\documentclass[aps,prx,twocolumn,superscriptaddress,longbibliography]{revtex4-2}
\usepackage[utf8]{inputenc}
\usepackage{amsmath, amssymb,amsfonts,dsfont,xspace,graphicx,relsize,bm,mathtools,xcolor,amsthm,soul}
\usepackage{tikz}
\usetikzlibrary{positioning, shadows,calc}
\definecolor{beige}{RGB}{245, 245, 220} 
\definecolor{lightmaroon}{RGB}{209, 170, 170} 

\usepackage{hyperref}

\usepackage[margin=1in]{geometry}
\hypersetup{
	colorlinks,
	linkcolor={red!100!black},
	citecolor={blue!100!black},
}

\usepackage{subcaption}
\usepackage{bbm}
\usepackage{mleftright}
\usepackage{hyperref}
\usepackage{tipa}
\usepackage{ulem}
\usepackage{enumitem}

\DeclareTextFontCommand{\emph}{\textit}

\def\01{\{0,1\}}

\DeclareDocumentCommand{\dist}{o}{%
  \IfNoValueTF{#1}{d}{d_{\mathrm{#1}}}%
}

\NewDocumentCommand{\Prob}{e{_} m}{%
  \IfNoValueTF{#1}{%
    \Pr \set*{#2}
  }{%
    \Pr_{#1} \set*{#2}
  }
}

\newcommand{\eps}{\epsilon}

\newcommand{\C}{\ensuremath{\mathcal{C}}}

\newcommand{\E}{\ensuremath{\mathbb{E}}}

\newcommand{\F}{\ensuremath{\mathcal{F}}}
\newcommand{\G}{\ensuremath{\mathcal{G}}}

\newcommand{\T}{\ensuremath{\mathsf{T}}}
\newcommand{\X}{\ensuremath{\mathsf{X}}}
\newcommand{\mN}{\ensuremath{\mathcal{N}}}
\makeatletter
\let\X\@undefined
\makeatother
\newcommand{\X}{\ensuremath{\mathcal{X}}}

\newcommand{\Exp}{\mathbb{E}}
\newcommand{\Tr}{\mathrm{Tr}}
\newcommand{\bra}[1]{\langle{#1}|}
\newcommand{\ket}[1]{|{#1}\rangle}

\newcommand{\ketbra}[2]{|{#1}\rangle\langle{#2}|}

\newcommand{\beq}{\begin{equation}}
\newcommand{\beql}[1]{\begin{equation}\label{#1}}
\newcommand{\eeq}{\end{equation}}
\newcommand{\eeqp}{\,\,\,.\end{equation}}
\newcommand{\eeqc}{\,\,\,,\end{equation}}

\theoremstyle{plain}
\newtheorem{problem}{Problem}

\newtheorem{defn}{Definition}

\newtheorem{theorem}{Theorem}
\newtheorem{proposition}{Proposition}
\newtheorem{lemma}{Lemma}
\newtheorem{corollary}{Corollary}
\newtheorem{fact}{Fact}

\usepackage{algorithm}
\usepackage{algpseudocode}

\begin{document}

\title{Learning quantum processes without input control}

\author{Marco Fanizza}
\affiliation{F\'{\i}sica Te\`{o}rica: Informaci\'{o} i Fen\`{o}mens Qu\`{a}ntics, Departament de F\'{\i}sica, Universitat Aut\`{o}noma de Barcelona, 08193 Bellaterra, Spain.}
\email{marco.fanizza@uab.cat}

\author{Yihui Quek} 
\affiliation{Dahlem Center for Complex Quantum Systems, Freie Universit\"{a}t Berlin, 14195 Berlin, Germany.}
\affiliation{Massachusetts Institute of Technology, Cambridge, MA, USA}
\email{yihuiquek3.14@gmail.com}

\author{Matteo Rosati} 
\affiliation{Dipartimento di Ingegneria Civile, Informatica e delle Tecnologie Aeronautiche, Università Roma Tre, Via Vito Volterra 62, I-00146 Rome, Italy.}
\email{matteo.rosati@uniroma3.it}

\date{\today}

\begin{abstract}
We introduce a general statistical learning theory for processes that take as input a classical random variable and output a quantum state. Our setting is motivated by the practical situation in which one desires to learn a quantum process governed by classical parameters that are out of one's control. This framework is applicable, for example, to the study of astronomical phenomena, disordered systems and biological processes not controlled by the observer. We provide an algorithm for learning with high probability in this setting with a finite amount of samples, even if the concept class is infinite. To do this, we review and adapt existing algorithms for shadow tomography and hypothesis selection, and combine their guarantees with the uniform convergence on the data of the loss functions of interest.
As a by-product we obtain sufficient conditions for performing shadow tomography of classical-quantum states with a number of copies which depends on the dimension of the quantum register, but not on the dimension of the classical one. 
We give concrete examples of processes that can be learned in this manner, based on quantum circuits or physically motivated classes, such as systems governed by Hamiltonians with random perturbations or data-dependent phase-shifts. 
\end{abstract}

\maketitle

\section{Introduction}

The goal of science is to gain a better understanding of nature. Because `nature isn't classical, dammit' \cite{Feynman-1986}, the tools of quantum information processing have been central to this pursuit. Insights transposed from statistical learning theory to the quantum domain have established rigorous guarantees for learners that predict properties of quantum states \cite{Aaronson2007, Aaronsonshadow20, BO21, Huang21manybody, Huangshadows}, classify phases of matter \cite{Huang21manybody}, learn quantum channels \cite{ChungLin21} or approximate models of physical dynamics \cite{huang2022quantum}. 

\begin{figure}[!t]
    \centering
    \includegraphics[trim = 1cm 1cm 1cm 1cm,width=.48\textwidth]{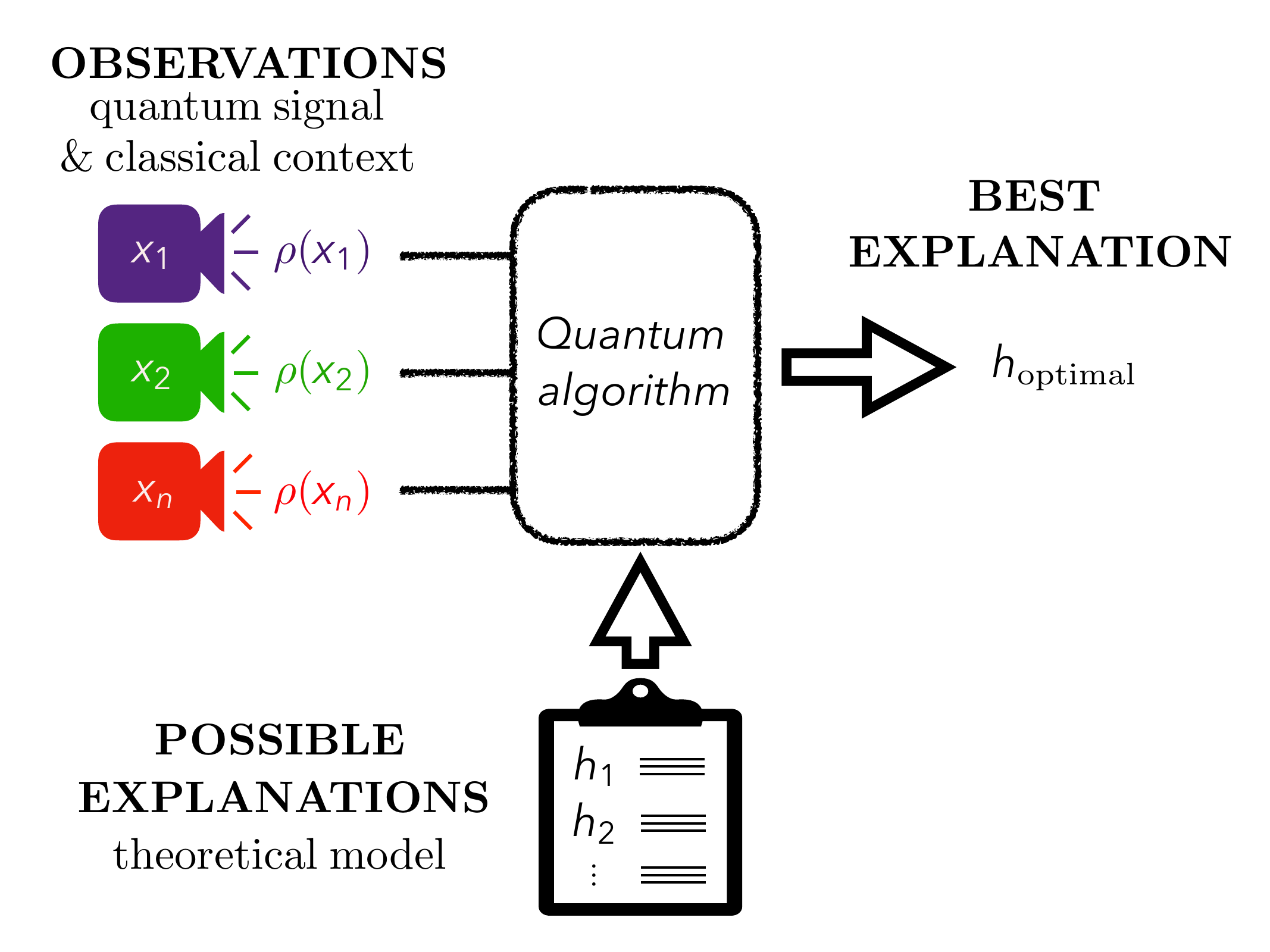}
    \caption{Schematic depiction of the experimental setting for which we provide theoretical learning guarantees. A size-$n$ training set comprising couples of a classical label $x$ and a corresponding quantum state $\rho(x)$ is produced by sampling from random sources parametrized by $x_1, x_2, \cdots, x_n$. The learner tries to find a theoretical explanation to the observed classical-quantum data using a hypothesis class of operator-valued functions $h(x)$, possibly of infinite cardinality.  We present quantum algorithms that are able to identify a near-optimal hypothesis $h_{\rm optimal}$ and obtain sample complexity guarantees in terms of covering numbers of the hypothesis class.}
    \label{fig:fig1}
\end{figure}

A feature shared by many of these works is that they require the quantum learner to have precise control of the unknown object -- for instance, the ability to request multiple identical copies of the unknown state or run the unknown process on a string of well-chosen inputs. That is, the unknown process is treated as a black box, to be applied on inputs specially designed by the learner. However, this assumption is not always satisfied in practice. A scientist can typically \textit{observe but not fully control} an unknown process of interest -- think of an astronomer analyzing signals generated by rare celestial events, a biologist probing molecular mechanics induced by biochemical signals in a noisy environment, or a physicist characterizing systems obeying Hamiltonians subject to random perturbations. In this work, we show that quantum processes can be learned even without the strong assumption of input control.

That is (see Fig.~\ref{fig:fig1}), the learner receives as examples the input-output pairs 
\beq
(x_1, \rho(x_1)), ...(x_n, \rho(x_n))
\eeq
where $x_i$ are classical inputs and the $\rho(x_i)$ are the quantum states output by the unknown process with these inputs. Here the lack of input control is reflected in the fact that the $x_i$ are not chosen by the learner but are samples from some distribution $\mathcal{D}$. Practically, the imperfect control could consist in the impossibility of choosing the input to a process or the inevitability of the input fluctuating throughout the various applications of the process. For instance, in optical imaging of celestial, atmospheric, planetary or biological events, $\rho(x)$ could describe the quantum state of an optical signal emitted or back-scattered by an unknown object, while $x$ is a set of environmental parameters influencing the state of the signal, e.g., temperature, optical depth, distance. 

We model the situations described above as follows: the learner has access to a source that outputs classical-quantum states
\begin{align}
&\sum_{x\in{\cal X}}{\cal D}(x)|x\rangle\langle x|\otimes\rho(x), \,\,\,\\&\text{where $\rho:\X \rightarrow \mathcal{S}$ and $x \overset{\mathcal{D}}{\gets} \X$}\nonumber,
\end{align}
{with $\mathcal{S}$ being the set of states of some Hilbert space $\mathcal{H}$ and $\rho$ the unknown process}. At each data collection, the input will be a product state $(\ketbra{x_1}{x_1}\otimes\rho(x_1))\otimes...\otimes(\ketbra{x_n}{x_n}\otimes\rho(x_n))$. This generalizes the setting of \cite{ChungLin21,caro2021binary} on probably approximately correctly (PAC)-learning quantum channels.

Our goal is to learn a possible classical-quantum source from which the data are sampled, and we construct algorithms to do so that have bounded sample complexity. A key object in the analysis is a certain distinguishability measure of the class of candidate classical-quantum processes, which we identify, that can be bounded even if this class is infinite. As an application, we also improve the shadow tomography procedure of \cite{Aaronsonshadow20} for an unknown classical-quantum state: a naive application of shadow tomography would not be guaranteed to work if the classical register is infinite-dimensional, while we show that only the dimension of the quantum register matters. 

Our theory also goes beyond previous works to encompass the \textit{agnostic case}: the case when the unknown process is \textit{not} included in the hypothesis set $\C$. {Note that the agnostic case was mentioned by~\cite{ChungLin21}, who provided a lower bound on the sample complexity, and tackled in a specific case by~\cite{caro2021binary}.} This is as opposed to the \textit{realizable} case, where the unknown process is guaranteed to be from the set. This is in keeping with the setting of learning Nature: agnostic learning models the situation when we learn a natural process using as hypotheses only the limited models that are within the reach of our theoretical understanding.

Our work overcomes a key technical hurdle not tackled by previous approaches: without input control, our learner cannot obtain identical copies of the process sampled at well-chosen points $\{x_i\}$. This assumption is crucial in, for instance, \cite{huang2022quantum} and \cite{LC22}. Instead, for every sampled $x_i$, she only gets a single copy of $\rho(x_i)$, which in general differs from the other copies $\rho(x_j)$, $j\neq i$. Nevertheless, we design a measurement strategy that learns even without the luxury of identical copies. 
Ref. \cite{ChungLin21} hinted that a VC-dimension-like quantity for this setting might be impossible to define. Here instead we establish a fundamental prerequisite for the definition of a statistically meaningful dimension (analogous to the fundamental theorem for uniform convergence of~\cite{vapnik1999overview}), providing sufficient conditions for learning concept classes of infinite cardinality with finite data: we introduce learning algorithms that succeed with high probability if a suitable \textit{covering number} of the quantum concept class $\mathcal C$ grows sufficiently slowly with the sample size. 

Our theory of learnability extends \textit{beyond} the usual setting of learning quantum processes given by quantum circuits. In fact, our learning model also encompasses the following scenarios:

\begin{itemize}
    \item We want to study how a small quantum system behaves in a variable environment. In this case, the classical random variable is a measurement of the status of the environment, for example a measure of classical fields. The copy of the quantum state is a copy of the state of the system corresponding to the measured state of the environment. One can imagine applying this scenario to molecules or nanostructures. Notice that the border between environment and object is arbitrary, therefore in the classical random variable one could include the outcomes of some predetermined measurement on the object of interest. With the same idea in mind, the quantum state could be also not the original state of the system but some post-processing of it, for example via a quantum sensor.
    \item We want to do imaging of a system, that is to associate to each point of the system a quantum state or channel. When the detector clicks, we receive as experimental data a pair comprising the position and the quantum state corresponding to that position. If our experimental setup uses spontaneous/stimulated emission, such that in a specific time interval we cannot guarantee that we can obtain an observation at a specific position, our model correctly represents the fact that we receive data from random positions and we cannot afford to receive multiple copies of the state corresponding to an arbitrary position. 
    \item We are studying a class of stars with a combination of classical and quantum sensors: for example, we get classical electromagnetic signals from a star, and a quantum sensor collects information about gravitational waves. We would like to study correlations between the electromagnetic and gravitational waves, but since these events are rare and unique we cannot repeat them at will.
\end{itemize}

Some toy models for concept classes inspired by these scenarios are discussed as applications. In these cases, we can find bounds on the covering number, from which sample complexity bounds can be obtained. 

\subsection{Setting: learning quantum processes with random classical input} We now go into more detail about our learning setting, which is a natural quantum generalization of supervised learning and builds on that of \cite{ChungLin21}. Suppose there is an unknown function $\rho: \mathcal{X}\rightarrow \mathcal{S}$ to be learned, $\mathcal{X}$ possibly of infinite cardinality, and a distribution $\mathcal{D}: \X \rightarrow [0,1]$. In fact, we will focus on processes $\rho$ that map from a \textit{classical} domain $\X$ to a \textit{quantum} set $\mathcal{S} \subseteq \mathcal{L}(\mathcal H)$, where $\mathcal{L}(\mathcal H)$ is the set of linear operators on a finite-dimensional Hilbert space $\mathcal{H}$.  When we want to keep track of the dimension $d$ of a Hilbert space, we use the notation $\mathcal{H}^{(d)}$. Furthermore, we will always be interested in the case where the unknown process $\rho$ outputs a \textit{quantum state}, that is, all operators in $\mathcal{S}$ are positive semi-definite and have unit trace.

The learner receives as input samples the pairs $(x_i, \rho(x_i))_{i=1}^n$ where $x_i \overset{\mathcal{D}}{\gets} \X$. Furthermore, the learner has a set of hypotheses, $\mathcal{C} = \{h:\mathcal{X}\rightarrow \mathcal{L}(\mathcal{H})\}$, and would like to use the smallest possible number $n$ of samples to choose a candidate $h$ from the class that accounts well for the observations. The accuracy of the learner's output $h$ relative to the true function $\rho$ will be measured by the \textit{true risk} $R_\rho:\mathcal{C} \rightarrow \mathbb R$, defined via a loss function $L:\mathcal{L}(\mathcal H)\times \mathcal{L}(\mathcal H)\rightarrow \mathbb R$:

\begin{equation}\label{eq:truerisk}
R_\rho(h):= \mathbb{E}_{x \sim \mathcal D} \left[L(\rho(x),h(x))\right]. \qquad \text{(True risk)}
\end{equation}
{Therefore, it is in the learner's interest to minimize the true risk, although she can only do so approximately, as detailed in the next section.}

In what is known as the \textit{realizable} setting of learning (studied by \cite{ChungLin21}), there is a promise that the unknown function comes from $\mathcal{C}$. We present results for this setting too, but go one step further. In learning Nature, our scientific models are but approximate descriptions that correspond more closely to reality at some scales than others. Thus, we primarily treat the \textit{agnostic} (or unrealizable) setting, in which no hypothesis in $\mathcal{C}$ is guaranteed to correspond exactly to the unknown function. 

We will focus on concept classes $\mathcal{C}$ that output two types of quantum objects:
\begin{itemize}
    \item \textbf{(Case 1) Quantum states:} $\C$ consists of hypotheses $h(x)=\sigma_h(x)$ which are state-valued functions. The true function $\rho(x)$ is also a state-valued function, and we use trace distance $L_s(\sigma_h,\rho)=d_{\text{tr}}(\sigma_h,\rho),$ a natural notion of distance between quantum states, as the loss function.
    \item \textbf{(Case 2) Quantum events/projectors:} $\C$ consists of hypotheses $h(x)=\Pi_h(x)$ which are projector-valued functions. Again, the true function $\rho(x)$ is a state-valued function and we will use as loss function the probability of not accepting the projector, i.e.,  $L_p(\Pi_h,\rho):=1-\Tr[\rho \Pi_h].$

\end{itemize}
 {We note that in Case 2  we can switch out the projectors for general POVM elements, by a standard dilation argument with which we can represent them as projectors on a larger space. However, the dilation is not unique and what we will say in the following will depend on the dilation. Therefore, for simplicity, we will always speak only about projectors.

A quick note on the motivation for defining projector-valued concept classes. In the classical case, when the label $y$ is not a deterministic function of $x$, one speaks of learning probabilistic concepts ~\cite{KEARNS1994464}. In~\cite{KEARNS1994464} two possible approaches are considered. The first is to learn a deterministic concept which maximises the probability of correct prediction, the second is to learn the conditional probability distributions $p(y|x)$ on average. Learning projector-valued classes is a generalization of the first approach, and learning state-valued classes is a generalization of the second approach.

Moreover, estimating $L_p(\Pi_h,\rho)$ for every $h\in \mathcal{C}$ (as we do later) encompasses shadow tomography \cite{Aaronson2007,Aaronsonshadow20}, a task where one is given a fixed state and a list of observables and has to output the expectation values of all observables in the list. Shadow tomography corresponds to the case where $\rho(x)$ and $\Pi_h(x)$ are constant as functions of $x$. Most importantly, our algorithm for this risk estimation problem on projector-valued functions is a key part of our strategy to attack Case 1. 

    }
\subsection{Results}
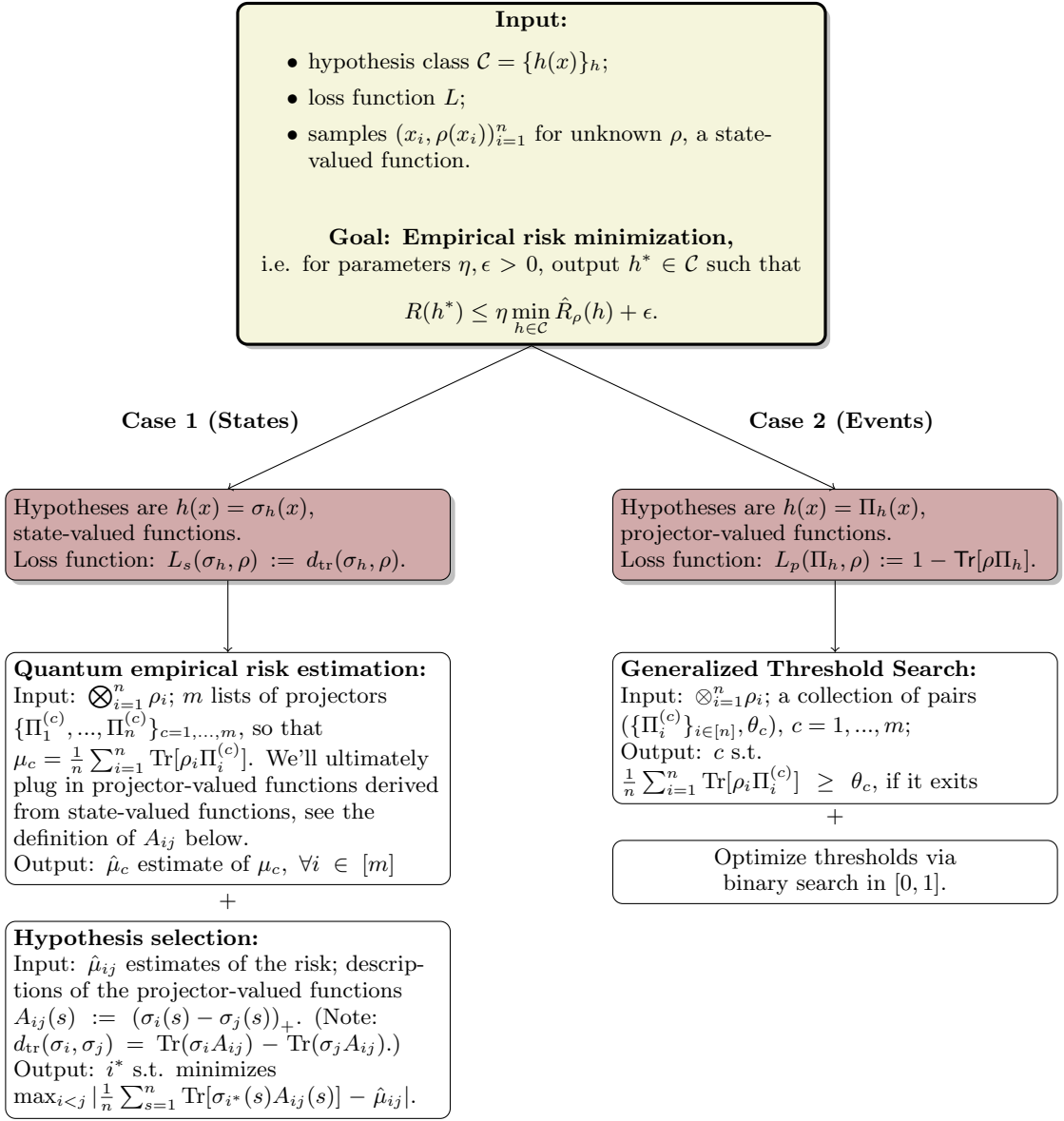
\begin{figure*}[!htbp]
\centering 

\begin{tikzpicture}
\node[draw, rectangle, rounded corners, align=center, text width=8cm, fill=beige, drop shadow, line width=1.2pt] (input) at (-5, 0){
\textbf{Input:}\\
\begin{itemize}
\item  hypothesis class $\mathcal{C}=\{h(x)\}_h$; 
\item loss function $L$; 
\item samples $(x_i, \rho(x_i))_{i=1}^n$ for unknown $\rho$, a state-valued function. \vspace{.5cm}
\end{itemize}
\textbf{Goal: Empirical risk minimization,}\\
i.e. for parameters $\eta,\epsilon>0$, output $h^*\in \mathcal{C}$ such that 
$$R(h^*) \leq \eta \min_{h\in \mathcal{C}} \hat{R}_{\rho}(h)+\epsilon.$$};

\node[draw, rectangle, rounded corners, align=left, text width=6cm, below left=2cm and -3cm of input, fill=lightmaroon, drop shadow] (case1) {
\raggedright Hypotheses are $h(x)=\sigma_h(x)$,\\ state-valued functions. \\
Loss function: $L_s(\sigma_h,\rho):=d_{\text{tr}}(\sigma_h,\rho)$.
};

\node[draw, rectangle, rounded corners, align=left, text width=6cm, below right=2cm and -3cm of input, fill=lightmaroon, drop shadow] (case2) {
\raggedright Hypotheses are $h(x)=\Pi_h(x)$, projector-valued functions. \\
Loss function: $L_p(\Pi_h,\rho):=1-\mathsf{Tr}[\rho \Pi_h]$.
};


\draw[->] (input.south) -- (case1.north) node[midway, below left, xshift=-10mm, yshift=2mm] {\textbf{Case 1 (States)}};
\draw[->] (input.south) -- (case2.north);

\node[below right, xshift=8mm, yshift=2mm, align=center] at ($(input.south)!0.5!(case2.north)$) {\textbf{Case 2 (Events)}};

\node[draw, rectangle, rounded corners, align=left, text width=6cm, below=1cm of case2] (summary) {
\textbf{Generalized Threshold Search:}\\
\raggedright{Input: $\otimes_{i=1}^{n}\rho_{i}$; a collection of pairs $(\{\Pi^{(c)}_i\}_{i\in[n]},\theta_c)$, $c=1,...,m$;\\
Output: $c$ s.t.\, $\frac{1}{n}\sum_{i=1}^n\mathrm{Tr}[\rho_{i}\Pi^{(c)}_{i}]\geq \theta_c$, if it exits}};

\draw[->] (case2) -- (summary);

\node[draw, rectangle, rounded corners, align=center, text width=6cm, below=0.5cm of summary] (binarySearch) {
  \raggedright Optimize thresholds via binary search in $[0,1]$.
};

\node[above=0.1cm of binarySearch] (plus) {
    +
};
\node[draw, rectangle, rounded corners, align=left, text width=6cm, below=1cm of case1] (theorem) {
\textbf{Quantum empirical risk estimation:}\\
\raggedright {Input: $\bigotimes_{i=1}^n \rho_i$; $m$ lists of projectors \\ $\{\Pi^{(c)}_{1},...,\Pi^{(c)}_{n}\}_{c=1,...,m}$, so that $\mu_c=\frac{1}{n}\sum_{i=1}^{n}\mathrm{Tr}[\rho_{i}\Pi^{(c)}_{i}]$. We'll ultimately plug in projector-valued functions derived from state-valued functions, see the definition of $A_{ij}$ below.\\
Output: $\hat \mu_c$ estimate of $\mu_c, \,\, \forall i\in [m]$
}};

\draw[->] (case1) -- (theorem);

\node[draw, rectangle, rounded corners, align=left, text width=6cm, below=0.5cm of theorem] (helstrom) {
    \textbf{Hypothesis selection:}\\
    Input: $\hat\mu_{ij}$ estimates of the risk; descriptions of the projector-valued functions $A_{ij}(s):=\left(\sigma_i(s)-\sigma_j(s)\right)_+$. (Note: $d_{\mathrm {tr}}(\sigma_i,\sigma_j)=\mathrm{Tr}(\sigma_i A_{ij}) - \mathrm{Tr}(\sigma_j A_{ij}).$)\\
    Output: $i^*$ s.t. minimizes $\max_{i<j}|\frac{1}{n}\sum_{s=1}^{n}\Tr[\sigma_{i^*}(s)A_{ij}(s)]-\hat\mu_{ij}|$.\\
};

\node[above=0.05cm of helstrom] (plusHelstrom) {
    +
};
\end{tikzpicture}
\caption{Summary of quantum empirical risk minimization algorithms} 
\label{fig:summarytech} 
\end{figure*}

{In this section, we state our main results, based on the algorithms summarized in Figure \ref{fig:summarytech} and obtained in Sections~\ref{sec:trnoniid},~\ref{sec:Empirical risk minimization},~\ref{sec:statlearn}. Recall that the goal is to learn an unknown function $\rho$ and the output of our learning algorithm will be some hypothesis $h\in \mathcal{C}$. While $\rho$ is not required to be in $\C$, our algorithms will always output some $h\in \C$ close to achieving the minimum \textit{empirical risk}, i.e., the average loss computed on the examples $(x_i, \rho(x_i))_{i=1}^n$:
\begin{equation}\label{eq:risk}
\hat{R}_\rho(h):=\frac{1}{n}\sum_{i=1}^n L(\rho(x_i),h(x_i)).\qquad \text{(Empirical risk)}
\end{equation}
That is to say, we use minimizing the empirical risk as a proxy for minimizing the true risk. {This principle, known as \textit{empirical risk minimization} (ERM), originated in classical statistical learning theory \cite{vapnik1999overview}; the contribution of this paper is to adapt it to learning quantum-valued classes.} {In this paper, we distinguish between two types of tasks: the term \textit{empirical risk minimization} refers to the task of outputting a \textit{single hypothesis} from $\C$ that approximately minimizes the risk. Along the way, we will also develop an algorithm for \textit{empirical risk estimation} (ERE), a term that refers to estimating the risk for \textit{every hypothesis} in $\C$.} 

As formulated by \cite{vapnik1999overview}, the success of ERM depends on a property of a concept class known as \textit{uniform convergence}~\cite{Anthony1999,Wolfnotes} of the empirical risks, which is controlled by a certain measure of the effective size of the concept class we define: the $\gamma_{1,q}$ covering number {(see general treatments of covering numbers, e.g.~\cite{Anthony1999, vershynin}).}

\begin{defn}[Covering number]\label{defcov}
Let $\mathcal{G} \subseteq \mathcal{L}(\mathcal H)^{\mathcal{X}}$ be a class of functions mapping to linear operators, and let $\eps\in(0,1]$. The \textit{covering number} of $\mathcal{G}$ is 
\begin{equation}
\gamma_{1,q}(n,\epsilon, \mathcal{G}):=\max \left\{ N_{in}\left( \epsilon, \mathcal{G},|| \cdot ||_{1,q,\vec{x}} \right)| \vec{x} \in \mathcal{X}^{n}  \right\},
\end{equation}
where for a pseudometric $d$, $N_{in}(\epsilon, \mathcal{G}, d)$ is the smallest cardinality of any internal $\epsilon$-cover of $\mathcal{G}$ according to the pseudometric $d$. { Here we have chosen as pseudometric the $\lVert \cdot\rVert_{1,q, \vec{x}}$-seminorm, which depends on the observed examples $\vec x$ as
\begin{equation}
    \lVert g_1 -g_2\rVert_{1, q, \vec{x}} := \frac{1}{|\vec{x}|}\sum_{i=1}^{|\vec{x}|} ||g_1(x_i)-g_2(x_i)||_q,
\end{equation}
 where $||A||_q$ is a Schatten norm (see definitions in section~\ref{sec:statlearnprel}) and $|\vec x|$ is the length of $\vec x$, i.e., the sample size. }
\end{defn}

{Intuitively, $\gamma_{1,q}$ describes the maximum number of hypotheses that can be pairwise distinguished on a dataset of size $n$ with resolution $\epsilon$, given information from the classical register only; in fact, the distinguishability between two hypotheses is measured by the average over classical outcomes of the appropriate distinguishability metric (operator norm for projectors and trace norm for states) for the corresponding values of the hypotheses.}

Since the rate of convergence of the empirical risk to the true risk is controlled by { the covering number $\gamma_{1,q}$~\cite{vapnik1999overview,Anthony1999,vershynin}, it is possible to minimize the risk by optimizing over an $\epsilon$-net of the concept class, which is finite-dimensional, rather than the class itself, which in general is infinite-dimensional.  This is the basis of our technique.} While the ERM principle is well-established classically, the non-trivial part in the quantum case is to minimize the empirical risk $\hat{R}_\rho(h)$ based on the string of samples $(x_i, \rho(x_i))_{i=1}^n$. { The main difficulties in this respect are the lack of identical copies of the input states and the fact that the observables associated to naive estimators of the empirical risks do not commute.} 

{We now present our results establishing a quantum variant of ERM, for both types of concept classes considered.

\subsubsection{ Quantum empirical risk minimization and estimation}
\label{sec:introerm}

{Let us first discuss projector-valued concept classes. Here,} the naive strategy of measuring each $\rho_{x_i}$ with $\Pi_h(x_i)$ { for all $h$} to estimate the empirical loss does not immediately work, since the projectors do not necessarily commute for different hypotheses. 
Nevertheless, we construct algorithms for both ERE and ERM on projector-valued concept classes, establishing the following Theorems.

\begin{theorem}[\sf{Quantum empirical risk minimization for projector-valued functions}]\label{th:theoere0}
Given access to a product state
\begin{equation}
\varrho=\rho_{1}\otimes...\otimes \rho_{n}
\end{equation} 
and a collection of lists of projectors $\{\Pi^{(c)}_{1},...,\Pi^{(c)}_{n}\}_{c=1,...,m}$, with 
\begin{equation}
\mu_c=\frac{1}{n}\sum_{i=1}^{n}\Tr[\rho_{i}\Pi^{(c)}_{i}],
\end{equation}
there is an algorithm which outputs $c^*$ and $\hat \mu_{c^*}$ such that
\begin{equation}
\Pr(|\hat{\mu}_{c^*} -\max_{c\in [m]}\mu_c|\geq \epsilon\cup |\hat{\mu}_{c^*} -\mu_{c^*}|\geq \epsilon)\leq \delta.\label{ineqerm}
\end{equation}
{ if $n$ is large enough. In fact we can take
\begin{equation}
    n=\tilde O\left(\frac{\log\frac{1}{\delta}\max(\log\frac{m}{\delta},\log^2 (e m))}{\epsilon^2}\right).
\end{equation}
}

\end{theorem}
Here, the notation $\tilde O$ hides logarithmic dependence on $\log m,\log\delta$, and $\epsilon$.}
{Note that the minimization in the caption of Theorem~\ref{th:theoere0} refers to the minimization of the loss function, which is $1-\mu_c$, whereas the theorem statement is about maximizing $\mu_c$, which is equivalent. 
}

 \begin{theorem}[\sf{Quantum empirical risk estimation for projector-valued functions}]\label{th:theoavsh0}

Given access to a product state
\begin{equation}
\varrho=\rho_{1}\otimes...\otimes \rho_{n}
\end{equation} 
and a collection of lists of projectors $\{\Pi^{(c)}_{1},...,\Pi^{(c)}_{n}\}_{c=1,...,m}$, with 
\begin{equation}
\mu_c=\frac{1}{n}\sum_{i=1}^{n}\Tr[\rho_{i}\Pi^{(c)}_{i}],
\end{equation}
there is an algorithm which outputs estimates $\hat \mu_{c}, c\in [m]$ such  such that
\begin{equation}
\Pr(\exists c \in [m]: |\hat{\mu}_{c} -\mu_c|\geq \epsilon)\leq \delta
\end{equation}
{if $n$ is large enough. In fact we can take 
\begin{equation}
    n=\tilde O\left(\frac{\log d \log\frac{1}{\delta} \max(\log \frac{m}{\delta},(\log^2 em))}{\epsilon^5}\right).
\end{equation}
}

\end{theorem}
{Here, the notation $\tilde O$ hides logarithmic dependence on $\log m,\log\delta,\log d$, and $\epsilon$.}
{ Note that estimation, in comparison with minimization, shows a worse dependence on $\epsilon$ and an explicit dependence on $d$. However, it represents a key subroutine to realize the ERM algorithm in the case of state-valued concept classes, which is the subject of the next theorem.}

{ Recall that, in this case, the loss function is different from the projector-valued case: it is the trace distance, rather than the overlap. This prevents us from immediately applying the techniques used for projector-valued classes, which only work to estimate expectation values. 

Nevertheless, using a trick from {\em quantum hypothesis selection} \cite{BO21}, we can reduce the task of risk \textit{minimization} for state-valued concept classes to estimating expectation values of Helstrom projectors constructed from the state-valued class.
The result is the following Theorem.}

\begin{theorem}[\sf{Quantum empirical risk minimization for state-valued functions}]\label{theo:states_finite0}
Let ${\cal C}=\{\sigma_i:[n]\rightarrow \mathcal{L}(\mathcal{H}^{(d)})\}_{i=1}^m$ be a class of state-valued functions and 
\begin{equation}
\varrho=\rho_1\otimes...\otimes\rho_n.
\end{equation}
There exists an algorithm which given $\varrho$ outputs $i^*$ such that
\begin{equation}\label{eq:ER_states0}
    \frac{1}{n}\sum_{s=1}^n d_{\rm tr}(\sigma_{i^*}(s),\rho_s)\leq 3\eta + 4\epsilon,
\end{equation}
where 
\begin{equation}\label{eq:optimalQHS10}
    \eta:=\min_{i\in [m]} \frac{1}{n}\sum_{s=1}^n d_{\rm tr}(\sigma_i(s),\rho_s),
\end{equation}
with probability of error less than $\delta$ if $n$ is large enough. { In fact we can take 
\begin{equation}
    n=\tilde O\left(\frac{\log d \log\frac{1}{\delta} \max(\log \frac{m}{\delta},(\log^2 em))}{\epsilon^5}\right).
\end{equation}
}
\end{theorem}
{Here, the notation $\tilde O$ hides logarithmic dependence on $\log m,\log\delta,\log d$, and $\epsilon$.}

\subsubsection{Learning via empirical risk minimization}
\label{sec:introlearnviaerm}

Combining the Theorems just stated and uniform convergence guarantees from statistical learning theory, we show the following sufficient conditions for learning.

\begin{theorem}[Learning quantum processes via ERM]\label{theo_ERM}

Suppose the concept class $\C$ consists of classical-quantum processes mapping to projectors or states and let $\epsilon>0$ be the accuracy parameter. { Furthermore, let $S = (x_i,\rho(x_i))_{i=1}^n$ be the training set, with $x_i \xleftarrow{\mathcal{D}} X$ and $\rho(\cdot)$ an unknown classical-quantum channel. 

Then, the appropriate ERM algorithm of Theorems~\ref{th:theoere0},~\ref{theo:states_finite0}, run on an $\epsilon$-net of the concept class $\mathcal{C}$ (according to the appropriate pseudometric determined by $x_1,...,x_n$), provide an \textit{agnostic} learning algorithm $\mathcal{A}:\X^n \times  \mathcal{L}(\mathcal{H}^{(d)})^{\otimes n} \rightarrow \C$. This algorithm} outputs an hypothesis ${\cal A}(S)$ satisfying, for some fixed $\eta,\xi\geq 1$, and $n$ large enough,
\begin{equation}\label{eq:ERM1}
\Pr_{S}[R_\rho(\mathcal{A}(S)) -\eta\inf_{h\in\mathcal C} R_\rho(h)<\xi\epsilon] > 1-\delta,
\end{equation}
if
\beq
\lim_{n\rightarrow \infty}\frac{\log^2 \gamma_{1,q}(n, \eps,\mathcal{C})}{n}=0,\qquad \forall\epsilon>0
\eeqp
In particular, this applies to risks defined via the loss functions $L_p$ (in this case $\eta=1$, $q=\infty$) and $L_s$ (in this case $\eta=3$, $q=1$) for projector-valued and state-valued concept classes $\mathcal{C}$ respectively. 

\end{theorem}

 We remark that Theorem \ref{theo_ERM} applies even to agnostic learners, going beyond Ref. \cite{ChungLin21} which considered only the realizable setting. In fact our methods also imply that we can learn with an \textit{infinite} concept class if $\lim_{n\rightarrow \infty}\log^2(\gamma_{1,q}(n,0,\C))/n=0$ (while in Ref. \cite{ChungLin21} show a sample complexity $O(\mathrm{poly}(\log|\C|))$). Moreover, in Appendix~\ref{subsec:warmup} we analyze the algorithm of \cite{ChungLin21} for learning pure states using a statistical learning theory approach, to show that $\lim_{n\rightarrow \infty}\log(\gamma_{1,1}(n,0,\C))/n=0$ suffices also with that algorithm. 

As a consequence of proving Theorem \ref{theo_ERM} for the case when the concept class consists of projectors, we are also able to speed up shadow tomography of classical-quantum states, vis-a-vis Ref. \cite{Aaronsonshadow20}:

\begin{theorem}[Shadow tomography of classical-quantum states]\label{theo_shadow}
Suppose the concept class $\C$ consists of quantum processes mapping to projectors, $\epsilon>0$ is the accuracy parameter, and the learner is given the same input as in the previous Theorem. 
Suppose also that

\begin{equation}
\lim_{n\rightarrow \infty}\frac{\log^2 \gamma_{1,\infty}(n,\eps,\mathcal{C})}{n}=0,\qquad \forall\epsilon>0.
\end{equation}

Then, the estimation algorithm of Theorem~\ref{th:theoavsh0} run on an $\epsilon$-net of the concept class $\mathcal{C}$ (according to the appropriate pseudometric determined by $x_1,...,x_n$), provide an \textit{agnostic} learning algorithm $\mathcal{A}:\X^n \times  \mathcal{L}(\mathcal{H}^{(d)})^{\otimes n} \rightarrow [0,1]^{|\C|}$, which output estimates of risks $\mu(h)$ for all concepts $h\in \C$, such that for some fixed $\xi\geq0$,
\begin{equation}\label{eq:ST}
\Pr_{S}[\,\forall h\in \C , \, \, |\mu(h)- R_\rho(h)|<\xi\epsilon] > 1-\delta,
\end{equation}
for $n$ large enough.

\end{theorem}
In words, we can not only find the minimum possible risk in the concept class, but also simultaneously estimate risks for all concepts. In fact, when $\C$ is finite, algorithm $\mathcal{A}$ performs shadow tomography of classical-quantum states with a copy complexity of $\tilde O(\mathrm{poly}(\log |\mathcal{C}|,\log d,1/\epsilon))$, where $d$ is the dimension of the quantum register only. By contrast, a naive application of shadow tomography~\cite{Aaronsonshadow20} has a copy complexity scaling as the dimension of the full space, which may be infinite if $|{\cal X}|=\infty$. 
\textit{Remark}: We mention that the algorithm of the aforementioned Theorems performs ERM correctly with high probability on \textit{any} possible data-sequence, { and not just on a subset of sequences which occur with high probability.} If one is willing to accept that sometimes ERM can fail with non-negligible probability, learning becomes possible when $\lim_{n\rightarrow \infty}\frac{\log \gamma_{1,\infty}(n,\eps,\mathcal{C})}{n}=0$, { rather than $\lim_{n\rightarrow \infty}\frac{\log^2 \gamma_{1,\infty}(n,\eps,\mathcal{C})}{n}=0$}. We argue for this in Appendix~\ref{appC}, by presenting a modification of the algorithms that perform ERM correctly only on a certain $\epsilon$-net on a subset of the data, with high probability. We illustrate this fact only for risk minimization for a projector-valued class, but the same can be shown also for risk estimation for a projector-valued class, and for risk minimization for a state-valued class.

\subsubsection{Examples of learnable classes}
In Section~\ref{sec:ex_covering}, we evince the applicability of our results by giving examples of quantum processes covered by our Theorems, for which we compute explicit upper bounds on the covering number $\gamma_{1,q}$ in terms of the dimensionality of the quantum systems and the \textit{fat-shattering dimension} of an appropriate concept class $\mathcal{F}$.
In all of the below examples, let $\mathcal{F}$ be a class of real-valued functions $g(x)$.
We build our concept classes using finite dimensional circuits and matrix functions that depend on the data via real functions $g(x)$ coming from concept classes $\mathcal{F}$. 
In fact, for the function classes with explicit data dependence that we consider, the covering number can slowly grow to infinity, but our Theorems show they are still learnable.\\ 

\textbf{(1) Quantum circuits in a particular architecture, possibly data-dependent.}
Concept classes consisting of $m$-qubit quantum circuits chosen from a set $S_m$ acting on arbitrary input states or projectors. That is, $\C_m = \{c_U(x) := U\rho(x) U^{\dag}\}_{U\in S_m}$ where $\rho(x)$ is a process that prepares the input state to the circuit, or $\C_m = \{c_U(x) := U\Pi(x) U^{\dag}\}_{U\in S_m}$, where $\Pi(x)$ is a projector. We give explicit upper bounds on the covering number for $S_m$ being\\
(a) the class of one-dimensional local quantum circuits on $m$ qubits of depth $\ell$ \cite{Brandao2016}, constructed by applying $\ell$ $2$-qubit nearest-neighbor gates on any pair of neighboring qubits,\\
(b) brickwork quantum circuits,\\
(c) the set of all unitaries,\\
(d) data-dependent circuits with any of the previous architectures, but modified by inserting in specific places in the sequence a number $\ell'$ of gates of the form $e^{i H_{j}g_{j'}(x)},\, j'=1,...,\ell'$, $H_{j'}$ fixed.\\

\textbf{(2) Gibbs states and low-energy projectors of a perturbed Hamiltonian.} Concept classes given by a set of Gibbs states obtained by perturbing a Hamiltonian with a field-dependent term, but the specific dependence on the field is not known, namely $H_0+g(x)V$. 
Similarly, concept class of projectors on low-energy eigenspace of $H_0+g(x)V$.\\

\textbf{(3) Phase-shifts with position-dependent depth. }
 Consider a spatially local channel that acts as a power of an unknown unitary channel at position $x$, according to some classical variable $g(x)$ at that position (for example a thickness), which we probe with some position-dependent state $\rho(x)$. 
 We can model this class as a state-valued (or projector-valued, when $\rho(x)$ are pure) function class:
$\mathcal C=\{f|f(x)=U^{g(x)}\rho(x){(U^\dagger)}^{g(x)}\,\, g\in\mathcal G, U=e^{i H}\in \mathrm{U}(d)\}.$
\\

\subsection{Related work}
There is a vast literature on how to learn quantum states and their properties. The exponential cost of full state tomography motivated to consider the problem of `pretty-good tomography'~\cite{Aaronson2007}, in which the goal of the learner is relaxed: she is content with obtaining a predictor function for \textit{properties} of the state instead of a full description of it. A similar task, consisting of estimating expectation values of a fixed list of observables, is known by the name \textit{shadow} tomography~\cite{Aaronsonshadow20, aaronson2019gentle,BO21}. Even though we aim to learn processes and not quantum states, we have leveraged the tools developed in the shadow tomography literature for our work. 

In this work, we delve into the considerably less studied realm of learning quantum \textit{processes} that have a classical input variable and a quantum output. As mentioned, this model is motivated by the observation that many processes in nature leave a classical \textit{imprint} of the conditions under which they occur, which the learner can detect with experimental equipment, but not observe. Our goal is to build a statistical learning theory for learning quantum processes that takes into account the unique structure of quantum measurements. 

Several previous works have investigated similar directions. Ref.~\cite{ChungLin21} first introduced the notion of Probably Approximately Correct (PAC)-learning of quantum channels, which captures our notion of lack of input control -- but solved the problem exclusively in the realizable setting. Their algorithms however have sample complexity scaling polynomially in $\log(|\C|)$, and thus cannot be applied when the concept class $\C$ is infinite. Our paper encompasses their setting (and improves their result for pure states) and solves the agnostic setting of channel learning. We also define an extensive measure of the concept class -- the covering number -- which controls learnability and can be finite even when $\C$ is infinite. The work~\cite{caro2021binary} addresses this aspect for the case of binary functions with quantum output (which are two states corresponding to 0 and 1 respectively), finding that Helstrom measurements are enough to get an upper bound on the sample complexity depending on the VC dimension of the associated classical function class.

Ref.~\cite{huang2022quantum} also considers agnostic learning for polynomial-time quantum processes. However, the paper models each process as an experiment over which the experimenter has full control. Accordingly, the learner runs the process on identical copies of the input in order to implement the hypothesis selection procedure (Section F.3.b). On the other hand, we are interested in the setting where the experimenter must contend with non-identical copies of the input. Additionally, while the model in~\cite{huang2022quantum} applies to channels with bounded input and output dimensions, where one can always find an $\epsilon$-net directly for the concept class (by its compactness), in our case instead we are interested in learning correlations between a classical and a quantum random variable, with the classical variable living possibly in an unbounded space. In this case, one cannot always define a finite $\epsilon$-net on the sources, but we show that learning can be still guaranteed in terms of $\epsilon$-nets of the datasets formed by pairs of classical and quantum variables. 

 Ref.~\cite{Huang21IT} also works in a PAC setting for learning quantum channels acting on classical inputs drawn from a distribution, but their notion of learning is to find a good predictor function (mapping to real values) for the expected value of a fixed set of observables on the output of an unknown quantum process. The focus of their work is whether quantum machine learning algorithms can have a large advantage over classical machine learning for various notions of prediction error. By contrast, our notion of learning is different and is not about predicting observables accurately. Furthermore, our work provides constructive quantum algorithms that achieve the sample complexity bounds we mention, while no such algorithms are presented in this other work. 
 
Ref.~\cite{Banchi2021b} also considered the problem of learning classical-quantum processes, providing generalization bounds in terms of entropic quantities. However, their setting is closer to a discrimination problem: the task is to find, among a set of possible projectors, one that best approximates joint distributions of classical data $x$ and a set of labels $c$, when measured on quantum states $\rho(x)$ representing embedding of $x$. Our setting is different, as their projectors do not depend on the specific value of $x$; moreover, their bounds on the generalization error depend on the source and not on the concept class, capturing a different aspect. 

Other relevant works on learning quantum processes are~\cite{PRXQuantum.4.040337}, which deals with estimating local properties of the output of an unknown circuit, with guarantees on the average error for certain distribution of input states, \cite{caro2023learning}, which shows how to efficiently learn the matrix elements of a channel in the Pauli basis, as a generalization of shadow tomography, and \cite{Aharonov2022, Chen22, 
Chen21}, which studies how to estimate and test properties of an unknown channel with or without quantum memory. 

Finally, the possibility of having generalization bounds for various loss functions evaluated on quantum circuits, using Rademacher complexities and covering numbers, has been considered extensively, especially with quantum machine learning applications in mind~\cite{PhysRevA.105.062431,Cai_2022, Popescu2021, Caro2021encodingdependent, PhysRevLett.128.080506, Gyurik2021, Caro2022, Rosati2022}. Relations between different combinatorial dimensions that cater to different modes of learning quantum states have also been explored in \cite{AQS21}. Our bounds on the covering numbers of processes described by quantum circuits are technically analogous to these existent results, which are devised for a different scenario with full control over the inputs. Indeed, such generalization bounds are valid if the sample complexity is defined as the number of unique data-point pairs $(x_i,\rho(x_i))$ seen by the learner, without restrictions on the number of times the pair $(x_i,\rho(x_i))$ can be actually produced (assuming that one is able to minimize the average of the loss function in some efficient way for every finite dataset). 
This analysis is not generically applicable in our case, since we do not have full control of the source and a good model for explaining the sample data cannot be in general inferred from finite samples. Obviously, since the covering numbers we consider are not constructed directly from the loss function but are dependent on the input data, they are generally worse in performance than those obtained for the cases with controlled input.

 \subsection{Learning algorithms: technical overview}

{In this section we explain the ideas behind our algorithms and their proofs (see also Figure~\ref{fig:summarytech}).}

\subsubsection{Threshold search for general product states (summary of Section~\ref{sec:trnoniid}) }\label{sec:introtr}

The key technical tool we establish is a generalization of the \textit{Threshold Search} algorithm in~\cite{BO21}, { which takes as input many identical copies of a state $\sigma$, a collection of pairs consisting of a projector and a threshold $(\Pi^{(c)},\theta_c)$, $c=1,...,m$ and reports if there is one projector whose average exceeds the threshold.}

In our generalization, the algorithm takes as input { only \textit{a single state} $\varrho$ --} a product (of possibly \textit{non-identical}) states, $\varrho=\otimes_{i=1}^{n}\rho_{i}$ -- {and a collection of pairs, each comprising a list of projectors and a threshold 
 $(\{\Pi^{(c)}_i\}_{i\in[n]},\theta_c)$; as output, the algorithm reports if there is one projector list whose average on the sample exceeds the threshold}: $\frac{1}{n}\sum_{i=1}^n\Tr[\rho_{i}\Pi^{(c)}_{i}]\geq \theta_c-\epsilon$.
Lemma~\ref{lem:threshold} gives the detailed formulation of this result. { In the following we outline its main technical steps.} 
For each list of projectors $\{\Pi^{(c)}_i\}_{j\in[n]}$, one can construct the collective projectors that accept at least $k$ times, denoted as $\{E^{(c)}_k\}_{k=0}^{n}$, 
\beq
E^{(c)}_k:= \sum_{\vec v\in \{0,1\}^n: ||\vec{v}||_1=k} \bigotimes_{i=1}^{n}(\Pi^{(c)}_i+v_i(\mathbf{1}-2\Pi^{(c)}_i)).
\eeqc

When the projective measurement $\{E^{(c)}_k\}_k$ is executed on $\varrho$, the random variable associated to the outcome $k$, denoted as $\mathsf{T}^{(c)}$, is a sum of Bernoulli variables, $\mathsf{T}^{(c)}=\sum_{i=1}^n \mathsf{T}^{(c)}_i$, where  $\mathsf{T}^{(c)}_i\sim (\Tr[\rho_{i}\Pi^{(c)}_{i}],1-\Tr[\rho_{i}\Pi^{(c)}_{i}])$. Via standard Chernoff bound arguments, to distinguish between 
\begin{itemize}
\item $\frac{1}{n}\sum_{i=1}^n\Tr[\rho_{i}\Pi^{(c)}_{i}]\geq \theta_c$ and 
\item $\frac{1}{n}\sum_{i=1}^n\Tr[\rho_{i}\Pi^{(c)}_{i}]\leq \theta_c -\epsilon $
\end{itemize} for a given $c$, it is sufficient to measure the binary event $\sum_{t >n( \theta_c -\epsilon/3)} E^{(c)}_t$, obtaining the correct answer with high probability with $n=O\left(\frac{1}{\epsilon^2}\right)$. However, after one of these measurements is applied $\varrho$ can change significantly, { so that the method of checking each threshold by measuring the threshold's accompanying projective measurement will not work naively. Therefore we construct a more clever measurement that is ``gentler": each threshold is associated with one such measurement. One performs the list of measurements sequentially on the same state and the measurement disturbs the state in a controlled way when it rejects. We now explain how such measurements are constructed in more detail.}

From the projectors  $\{E^{(c)}_k\}_k$, one can construct the events $B_c$:

\begin{equation}
B_c := \sum_{t=1}^n \Pr(\mathsf{X} + t > \theta n ) E^{(c)}_t,
\end{equation}
 where $\mathsf{X}$ is an exponential random variable, $p_{\lambda}(\mathsf{X}=x):=\lambda e^{-\lambda x}$, for some $0<\lambda<1$. 
 When $B_c$ is measured on $\varrho$, the probability of accepting is $\Pr(\mathsf{X} + \mathsf{T}^{(c)} > \theta n )$. The exponential random variable has the role of smoothing the event $B_c$, so that the probability of accepting is still exponentially suppressed in $n$ when $\frac{1}{n}\sum_{i=1}^n\Tr[\rho_{i}\Pi^{(c)}_{i}]\geq \theta_c-\epsilon$, but at the same time, when $B_c$ rejects, the state $\varrho$ does not change a lot. More explicitly, we can show in Theorem~\ref{thm:chi-stable-final} that for $\lambda= \frac{1}{D\sqrt{n}}$ and some constants $D>0$:

\begin{align}\label{eq:boundexp0}
\Tr[\varrho B_c] \leq e^{\left(-\frac{\sqrt{n}}{D}\left(\theta-\frac{1}{n}\sum_{i=1}^n\Tr[\rho_{i}\Pi^{(c)}_{i}]\right)\right){\color{black}+e/(2D^2)}},
\end{align}
and, if $\Tr[\varrho B_c]<1/4$, with $\varrho|_{\sqrt{\mathbf{1}-B_c}}:=\tfrac{\sqrt{\mathbf{1}-B_c}\varrho\sqrt{\mathbf{1}-B_c}}{\Tr[\varrho(\mathbf{1}-B_c)]}$ being the post-measurement state conditioned on rejecting $B_c$, for some constant $C>0$:
\begin{align} \label{eq:bures0}
d_{\mathrm{tr}}\left(\varrho,\varrho|_{\sqrt{\mathbf{1}-B_c}}\right) \leq C \Tr[\varrho B_c].
\end{align}

Note that the last equation is a consequence of the Gentle Measurement lemma for the Fidelity (e.g., Proposition 2.2 in~\cite{BO21}), together with a bound on the $\chi^2$-divergence of the distributions of $\mathsf{X} + \mathsf{T}^{(c)}$ and $\mathsf{X} + \mathsf{T}^{(c)}$ conditioned on $\mathsf{X} + \mathsf{T}^{(c)}\leq \theta n$.

{ Based on the previous definitions,} the algorithm (Algorithm~\ref{alg:threshold}) for threshold search consists in measuring the events $B_c$, $c=1,...,m$ in sequence until a first acceptance $c^*$ occurs. { Then } $c^*$ is declared as the index of the list of projectors with expectation value above threshold. Lemma~\ref{lem:threshold} shows that $n=O(\frac{\log^2(e m)}{\epsilon^2})$ is a sufficient condition for Algorithm~\ref{alg:threshold} to succeed with probability larger than $0.03$. The proof is based on a quantum union bound for sequential measurement proved in~\cite{BO21}.

 In fact, our contribution is to note that the proofs of~\cite{BO21} can be adapted to the case of non-identical states and projectors, using the fact that the concentration properties of the Poisson binomial distribution are sufficient to reproduce the same argument of~\cite{BO21} with some appropriate adjustments. We do not make additional comments here on the proof, and we redirect to~\cite{BO21} for additional explanations on the idea behind the algorithm and remarks on its connections with adaptive data analysis.

\subsubsection{Quantum empirical risk minimization (summary of Section~\ref{sec:Empirical risk minimization})  }
\label{sec:introermp}

{\textbf{Projector-valued concept classes (Theorem~\ref{th:theoere0} and Theorem~\ref{th:theoavsh0}) ---} 
{ The main idea that lets us perform ERM for projector-valued functions is to observe that, in this case, the empirical risk is the average of expectation values of a list of projectors. Hence, we can directly apply the techniques of Sec.~\ref{sec:introtr} to check if empirical risks are above or below a threshold. Furthermore, for our purposes, it is sufficient to consider concept classes of finite cardinality.} Thus, in the following we assume that our concept class can be described as $m$ lists of projectors $\{\Pi^{(c)}_i\}_{i\in[n]}$, $c\in[m]$. Our Algorithm~\ref{algo:proj2} for ERM works as follows.

{(I) \em Sampling-without-replacement step:}
given an input $\varrho=\otimes_{i=1}^{n}\rho_{i}${ a sufficiently large number of random batches $\varrho_s=\otimes_{k=1}^{l}\rho_{s,k}$ of size $l$, obtained by sampling without replacement from the list $i=1,...,n$. In the same way, one selects corresponding lists of projectors $\{\Pi^{(c)}_{s,k}\}_{k\in[l]}$ from the initial list $\{\Pi^{(c)}_j\}_{j\in[n]}$, for all $c$. 

{ Via an argument based on concentration inequalities for sampling without replacement, we show that this procedure has the effect} that the empirical risk on each batch $\varrho_s$, i.e., $\frac{1}{l}\sum_{j=1}^l\Tr[\rho_{s,j}\Pi^{(c)}_{s,j}]$, is $\epsilon$-close to that on the full training set $\varrho$. The precise relation between the batch size $l$, the number of batches, the batch-risk approximation error $\epsilon$, and the probability of error of this procedure is given in Lemma~\ref{lem:multisamplerep}.}

{(II) \em Binary search for optimal threshold:}
the empirical risks take values in $[0,1]$, therefore we start with a candidate minimum empirical risk $0.5$, and use Algorithm~\ref{alg:threshold} { on the first batch} to check if there is a candidate below this threshold. If we find one, we guess a new minimum empirical risk at $0.25$, otherwise we guess $0.75$ (these are approximate values, see the technical treatment for the details). {We iterate this procedure via binary search, using a new batch for each new threshold, and} selecting the upper or lower halves of the candidate interval, until we get to approximate the value of the empirical risk with precision $\epsilon$: we will terminate with a number of step $O(\log\frac{1}{\epsilon})$. 

Since Algorithm~\ref{algo:proj2} is guaranteed to work only with a constant probability, we may need to repeat it a certain number of times to ensure we are able to get a candidate below threshold with high probability, if there is one. Moreover, we also need to check that the candidate risk is indeed below threshold, which is done by a further measurement of the corresponding empirical average. 

{ Crucially, since in all of these steps we cannot reuse batches $\varrho_s$, the empirical risks must be close between different batches,} which is guaranteed by the sampling-without-replacement step. By a union bound on the probability of the several types of error, we can obtain our guarantee for ERM, Theorem~\ref{th:theoere0}. 
\\

{ A similar approach can be used to obtain Theorem~\ref{th:theoavsh0},} viewing ERE as a generalization of \textit{shadow tomography}. Shadow tomography~\cite{Aaronsonshadow20} consists of the following task: given $n$ copies of an unknown quantum state $\sigma$, and a list of projectors $\Pi_1,\ldots \Pi_m$, output approximations of $\Tr(\rho \Pi_c) \,\, \forall c\in [m].$ The Threshold Search algorithm of~\cite{BO21} was indeed used to obtain { the state-of-the-art sample complexity for shadow tomography}. 

The idea is to keep a list of candidate intervals for the true values of $\Tr(\sigma \Pi_c)$ for each $c$, and use their extremes as thresholds in the Threshold Search algorithm, to be run with projectors $\Pi_1,\ldots \Pi_m$ as well as $\mathbf{1}-\Pi_1,\ldots \mathbf{1}-\Pi_m$. If one of the true values of $\Tr(\sigma \Pi_{c})$ is far from the candidate values, one such $c^*$ with this property will be found after a certain number of attempts, with high probability. If this happens, the list of candidates for the true expectation values is updated. Given a specific way to compute the candidate intervals, the process is guaranteed to terminate with intervals of size $O(\epsilon)$, containing the true values, after $\tilde O(\frac{\log d}{\epsilon^3})$ rounds~\cite{Aaronsonshadow20}.

Our Algorithm~\ref{algo:shadow} generalizes the i.i.d case with two main modifications, as follows.
\\

{(I) \em Looking for bad estimates in the general product states case:} 
with our generalization of Threshold Search for general product states and our sampling-without-replacement step, we can follow the same scheme for ERE: given a list of candidate values for $\frac{1}{n}\sum_{i=1}^n\Tr[\rho_{i}\Pi^{(c)}_{i}]$, search for $c^*$ such that the true value of $\frac{1}{n}\sum_{i=1}^n\Tr[\rho_{i}\Pi^{(c)}_{i}]$ is $\epsilon$-far from the candidate. 
\\

{(II) \em Update rule for the general product states case:}
we generalize Aaronson's update procedure. We stress that in both Aaronson's and our generalization this step can be carried out by a classical computer, even if the computation involves quantum states. 

In the original work, the estimates $\mu_{t,c}$ for the projector $\Pi_c$ at update step $t$ are obtained starting from $q$ copies of a maximally mixed state, that is from the state $\rho_0^{\ast}:=\left(\frac{\mathbf{1}}{d}\right)^{\otimes q}$.
At each step, the estimates $\mu_{t,c}$ are the expectation of the empirical average of $\Pi_c$ on $\rho_t^{\ast}$, which is also the predicted fraction of acceptances when $\Pi_c$ is measured on each subsystem $\mathcal{H}_i$ in the tensor product $\bigotimes_{i=1}^q\mathcal{H}_i$. 
If at step $t=1,...,T$ some $c^*$ is detected as associated to a prediction smaller (larger) than the true value, $\rho_{t}^{\ast}$ is obtained post-selecting $\rho_{t-1}^{\ast}$ on the event for which measuring $\Pi_c$ on each subsystem accepts a fraction of times slightly larger (smaller) than what previously predicted. Clearly, this makes the estimates progressively more accurate. 

For the general product states case, the main change is that we start from the following classical-quantum state at step $t=0$:

\beq
\rho_0^{\ast} := \left(\frac{1}{n}\sum_{i=1}^n \ketbra{i}{i}\otimes \frac{\mathbf{1}}{d}\right)^{\otimes q},
\eeq

where $\ketbra{i}{i}$ are orthonormal projectors of an auxiliary system.
At each step $\rho_{t}^{\ast}$ is used as a guess for the classical-quantum state

\begin{equation}\sigma=\left(\frac{1}{n}\sum_{i=1}^n \ketbra{i}{i}\otimes \rho_i\right)^{\otimes q},
\end{equation}

in the following sense.
Note that the quantities $1-R_{\varrho_{s}}(c)$ can be expressed as expectation values of projectors 

\begin{equation}
\Pi^{(c)}:=\sum_{i=1}^n\ketbra{i}{i}\otimes \Pi_{i}^{(c)}
\end{equation}

on $\sigma$, $1-R_{\varrho_{s}}(c)=\Tr[\Pi^{(c)}\sigma]$. The empirical average of $\Pi^{(c)}$ on $\mathcal{H}^{\otimes q}$ is $\overline{\Pi^{(c)}}:=\frac{1}{q} (\Pi^{(c)}\otimes \mathbf{1}\otimes ....\otimes  \mathbf{1}+\mathbf{1}\otimes \Pi^{(c)}\otimes ....\otimes  \mathbf{1}+\mathbf{1}\otimes \mathbf{1}\otimes ....\otimes \Pi^{(c)})$. It follows that  $\Tr[\Pi^{(c)}\sigma]=\Tr[\overline{\Pi^{(c)}}\sigma^{\otimes q}]$. $\rho_t^*$ is a guess for $\sigma^{\otimes q}$ in the sense that, at each step $t$, the estimates of the risks are $\mu_{c,t}=\Tr[\overline{\Pi^{(c)}}\rho_t^*]$.
The algorithm now runs exactly as Aaronson's. The key point to conclude is that

\begin{equation}
\rho_0^{\ast} = \rho_0^{\otimes q}=\frac{1}{d^q}{\sigma}^{\otimes q}+\left(1-\frac{1}{d^q}\right)\omega,
\end{equation}
for some positive semi-definite and trace-1 $\omega$. Therefore, with probability at least $\frac{1}{d^q}$, $\rho_0^{\ast}$ behaves as $\sigma^{q}$, which would accept all the post-selection updates with high probability. On the other hand, if some estimate is incorrect at time $t$, one can show that $\rho_t^{\ast}$ should reject with high probability. These two facts give rise to contradiction unless the number of necessary rounds is $\tilde O(\frac{\log d}{\epsilon^3})$.

The proof of Theorem~\ref{th:theoavsh0} is then obtained through a union bound over all possible sources of errors of Algorithm~\ref{algo:shadow}. We also note that, for projector-valued functions, the estimation of the empirical risks is similar to the \textit{diverse-state setting} for shadow tomography considered in~\cite{aaronson2019gentle}, { where the state is a general product state but the projectors do not change with the subsystems.} 
\\

{\textbf{State-valued concept classes (Theorem~\ref{theo:states_finite0}) ---} Here the loss function used is different from the projector-valued case: it is trace distance, and not overlap. This prevents us from immediately recycling the technique used previously, which only works to estimate expectation values. 

Also in this case, it will suffice for the moment to consider finite cardinality classes, described by lists of states $\{\sigma_c(s)\}_{s\in[n]}$, $c\in[m]$. Equivalently, we can encode this information into a collection of classical-quantum states $\sigma_c=\frac{1}{n}\sum_{k=1}^n \ketbra{s}{s} \otimes \sigma_{c}(s)$, while the information about the input $\varrho=\otimes_{s=1}^{n}\rho(s)$ is encoded in the classical-quantum state $\sigma=\frac{1}{n}\sum_{k=1}^n \ketbra{s}{s} \otimes \rho(s)$. Then, the empirical risk for $\varrho$ and the hypothesis $\sigma_c(s)$ is simply $d_{\mathrm{tr}}(\sigma,\sigma_c)$.

The key idea is again inspired by~\cite{BO21}, which used it { for a task named} {\em quantum hypothesis selection}. They used the following observations: 
\begin{itemize}
\item The trace distance between pairs of states $\sigma_i,\sigma_{j}$ can be written as the difference of two expectation values via Helstrom's theorem. Indeed, defining $A_{ij}(s):=\left(\sigma_{i}-\sigma_j\right)_+$, where $(\cdot)_+$ is the projector on the positive part of the argument, we have that
\begin{equation}
    d_{\rm tr}(\sigma_i,\sigma_j)=\Tr(\sigma_i A_{ij}) - \Tr(\sigma_j A_{ij}).
\end{equation}
\item If a state $\rho$ is close to $\sigma_k$, 
 so are the expectation values. 
{ A good guess for $\rho$ is then the state $\sigma_{k^*}$ such that
\begin{equation}
   k^* = \underset{k\in[m]}{\mathrm{argmin}} \max_{i<j}|\Tr[\rho A_{ij}]-\Tr[\sigma_{k} A_{ij}]|.
\end{equation}
Indeed, suppose that $i^*$ is such the true minimizer of $d_{\mathrm{tr}}(\sigma_{i^*},\rho)\leq \eta$. }Then we have that
\begin{align}
    &\max_{i<j}|\Tr[\rho A_{ij}]-\Tr[\sigma_{k^*} A_{ij}]|\\
    &\leq \max_{i<j} |\Tr[\rho A_{ij}]-\Tr[\sigma_{i^*} A_{ij}]|\leq \eta
\end{align}
and via triangle-inequalities 
\begin{align}
    &d_{\rm tr}(\sigma_{k^*},\rho)
    \leq d_{\rm tr}(\sigma_{k^*},\sigma_{i^*}) + d_{\rm tr}(\sigma_{i^*},\rho) \\
    &= \left|\Tr(\sigma_{k^*}A_{k^*i^*}) - \Tr(\sigma_{i^*}A_{k^*i^*})\right| + \eta \\
    &\leq \left|\Tr(\sigma_{k^*}A_{k^*i^*}) - \Tr(\rho A_{k^*i^*})\right| \nonumber\\
    &+ \left|\Tr(\rho A_{k^*i^*} - \Tr(\sigma_{i^*}A_{k^*i^*})\right| + \eta \leq 3\eta.
\end{align}
\item Finally, if $\Tr[\rho A_{ij}]$ are known only with precision $\epsilon$, the above procedure is robust, and allows to find $k^*$ s.t. $d_{\mathrm{tr}}(\rho,\sigma_{k^*})\leq 3\eta +2\epsilon$ from approximate estimation of the expectation values of the Helstrom projectors.
\end{itemize}

Our crucial observation is that for the classical-quantum states $\sigma_c$, the Helstrom projectors are $A_{ij}:=\sum_{s\in[n]}\ketbra{x}{x}\otimes A_{ij}(s)$, with $A_{ij}(s):=\left(\sigma_i(s)-\sigma_j(s)\right)_+$. In turn, their expectation values can be estimated via our ERE algorithm for the lists of projectors $\{A_{ij}(s)\}_{s\in[n]}$, $i<j$, given $\varrho$  as input. This reduces our procedure for learning with state-valued functions to a post-processing of Algorithm~\ref{algo:shadow} for ERE for projector-valued functions and proves Theorem~\ref{theo:states_finite0}.

\subsubsection{Statistical learning for classical-quantum processes (summary of Section~\ref{sec:statlearn})  }
\label{sec:introstatlearn}

Building on the ERM algorithms just reviewed, we are able to give sufficient conditions for the minimization of the true risk in both cases of projector- and state-valued functions.
The algorithm works as follows.
\begin{itemize}
\item From the knowledge of the classical variables $x_i$ for all $i$, we can construct an $\epsilon$-net of the function class using the appropriate pseudometric: for two projector-valued functions $\Pi^{(c)}(x)$ and $\Pi^{(c')}(x)$, their pseudodistance on the data is $\frac{1}{n} \sum_{i=1}^{n}||\Pi^{(c)}(x_i)-\Pi^{(c')}(x_i)||$, while for two state-valued functions $\sigma_c(x)$ and $\sigma_{c'}(x)$, their pseudodistance on the data is $\frac{1}{n}\sum_{i=1}^{n}d_{\mathrm{tr}}(\sigma_c(x_i),\sigma_{c'}(x_i))$.
The cardinality of the $\epsilon$-net will be bounded by the $\gamma_{1,\infty}$ covering number in the case of projectors and by the $\gamma_{1,1}$ covering number in the case of states (see Definition~\ref{defcov}).
\item We then run ERM on $\varrho=\otimes_{i=1}^{n}\rho(x_i)$ using as concept class the $\epsilon$-net found at the previous step. 
\end{itemize}

{ Let us observe that our ERM algorithms are guaranteed to work if the covering numbers $\gamma_{1,q}$ grow slowly with $n$, as they take the place of $m$ in Theorems~\ref{th:theoere0},~\ref{th:theoavsh0} and \ref{theo:states_finite0}. On the other hand, via classical statistical learning theory, uniform convergence of the estimated empirical risk to the true risk is controlled by the covering numbers of the loss function, which depend on the unknown states $\rho(x)$. 

Nevertheless, we are able to show that these covering numbers can be bounded by $\gamma_{1,q}$, which does not depend on $\rho(x)$ and is in principle computable by the learner. Therefore, we get a sufficient condition, independent of the data, for learning in terms of the growth of $\gamma_{1,q}$ with $n$, proving Theorem~\ref{theo_ERM}.} Similarly, Theorem~\ref{theo_shadow} can be obtained by running the ERE algorithm of Theorem~\ref{th:theoavsh0} and checking that uniform convergence is also satisfied when $\gamma_{1,\infty}$ grows slowly with $n$.

\subsection{Outlook}
Our goal is to learn Nature. In the past few years, quantum information processing tools have been fundamental to building a theory of learning states produced by quantum circuits. However, the circuit picture may not be the most natural one to describe all quantum mechanical processes. Our work represents a first incursion into the territory of developing a statistical learning theory for physical processes beyond quantum circuits. Our intention is to build a general theory of what we can quantum-learn, going beyond the terra firma of quantum circuits and what they can model, reaching to characterization protocols in quantum mechanics that may be best expressed outside the circuit model, for instance, metrology, sensing, calibration and verification, and eventually building a unified foundation for designing physical experiments. 

\section{Preliminaries}

\subsection{Notation}
We will consider random variables $\mathsf X$ valued in the set $\mathcal X$, denoting as $\mathcal{D}(x)=\Pr_{x\sim\mathcal{D}}(\mathsf X=x)$ the probability that the random variable takes value $x\in \mathcal X$; accordingly, $\Pr_{x\sim\mathcal{D}}(E)$ is the probability of an event $E\subseteq \mathcal X$. We will denote the set $\{1,...,n\}$ as $[n]$. We denote the $n$-fold cartesian product of $\mathcal{X}$ as $\mathcal{X}^n$, elements in it as vectors $\vec x\in \mathcal{X}^n$, and we use the notation $|\vec{x}|:=n$ to refer to the length of the vector. $\mathcal{D}^n$ denotes the probability distribution of $\vec{x}$. We denote a Hilbert space of dimension $d$ by $\mathcal{H}^{(d)}$. We further denote by $\mathcal{L}(\mathcal{H}^{(d)})$ the set of linear operators on $\mathcal{H}^{(d)}$, and by $D(\mathcal{H}^{(d)}) \subseteq \mathcal{L}(\mathcal{H}^{(d)})$ the set of density matrices, that is, the subset of $\mathcal{L}(\mathcal{H}^{(d)})$ which is positive semi-definite and has unit trace. These matrices describe quantum states.
For brevity, the Hilbert space of $n$ qubits is denoted  $\mathcal{H}_n$. 

An arbitrary valid quantum operation on quantum
states can be expressed as a quantum channel $\Phi$, i.e., a completely positive trace-preserving map, and we denote the output of a channel applied to $\rho$ as $\Phi[\rho].$ A special case of quantum channels are unitary channels, defined from a \textit{unitary matrix} $U$ (meaning $UU^\dagger=U^\dagger U=\mathbf{1}$). An application of a unitary~$U$ to
the state $\rho$ results in the quantum state $U\rho U^\dagger$. Any quantum channel in finite dimension can be expressed in Kraus representation as $\Phi[\rho]=\sum_{i=1}^{m}K_i\rho K_{i}^{\dagger}$ for a suitable finite set of operators $\{K_i\}$. 

In order to extract \textit{classical} information out of a quantum state, one can perform a POVM (positive-operator valued measurement) which is specified by a set of $m$ positive semidefinite matrices $\{E_1, \ldots, E_m\}$ satisfying $\sum_i E_i=\mathbf{1}$. A measurement on a state $\rho$ using such a POVM returns a classical outcome $i \in [m]$ with probability~$\Tr[E_i \rho]$. For any operator $0\leq  E\leq \mathbf{1}$, we can define a two-outcome POVM $\{E,\mathbf{1}-E\}$, which we say that implements the measurement of the \textit{event} $E$, and the probability associated to $E$ is denoted as in the classical case as $\Pr(E)=\mathbb{E}_{\rho}[E]$. 
We will similarly use the notation $\mathbb E_{\rho}[A]:=\Tr[\rho A]$ to express expectation values of operators $A$. We will denote the standard deviation of a classical random variable $\mathsf{T}$ as $\operatorname{stddev}[\mathsf{T}]:=\sqrt{\mathbb{E}[(\mathsf{T}-\mathbb E[\mathsf{T}])^2]}$, where the expectation value is with respect to the probability distribution of $\mathsf{T}$.

We will focus on \textit{classical-quantum} states, i.e. states that can be written as $\rho=\sum_{x\in \mathcal X}\mathcal{D}(x)\ketbra{x}{x}\otimes \rho(x)$, where $\rho(x)$ are states of $\mathcal H^{(d)}$ and $\mathcal{D}$ is a probability distribution on $\mathcal X$, $\{\ket{x}\}_{x\in\mathcal{X}}$ are an orthonormal basis of a Hilbert space that we denote $\mathcal{H}_{\mathcal{X}}$. Note that for any operator $A\in \mathcal{L}(\mathcal{H}_{\mathcal{X}}\otimes \mathcal{H}^{(d)})$ and a classical-quantum state $\rho\in D(\mathcal{H}_{\mathcal{X}}\otimes \mathcal{H}^{(d)})$, $\mathbb E_{\rho}[A]=\sum_{x\in \mathcal X}\mathcal{D}(x) \Tr[A(x)\rho(x)]=\sum_{x\in \mathcal X}\mathcal{D}(x)\mathbb E_{\rho(x)}[A(x)] $, where $A(x):=\bra{x}A\ket{x}\in \mathcal{L}(\mathcal{H}^{(d)})$. For a set of $n$ operators on a Hilbert space $\mathcal H$, $A=\{A_1,...,A_n\}$, we denote as $\overline{A}$ the operator on $\mathcal{H}^{\otimes n}$ defined as $\overline{A}:=A_1\otimes \mathbf{1}\otimes ....\otimes  \mathbf{1}+\mathbf{1}\otimes A_2\otimes ....\otimes  \mathbf{1}+\mathbf{1}\otimes \mathbf{1}\otimes ....\otimes A_n$

\subsubsection{Distances between probability distributions, quantum states and quantum channels}
We consider the following distances, defined for probability distributions on discrete supports. Let $P,Q$ be two probability measures on the same support $\mathcal{X}$.  
\begin{itemize}
    \item \textbf{Total variation distance:}
    \begin{align}\label{eq:TV}
        TV(P,Q):= \sup_{A \subseteq \mathcal{X}} |P(A)-Q(A)| \nonumber\\= \frac{1}{2} \sum_{x\in \mathcal{X}} |P(x)-Q(x)|.
    \end{align}
    \item \textbf{Chi squared distance:}
    \beq
    d_{\chi^2}(P,Q) := \sum_{x \in \mathcal{X}} Q(x) \left(1-\frac{P(x)}{Q(x)}\right)^2 
    \eeqp
    \item \textbf{Bhattacharya Coefficient:} 
    \beq
    BC(P,Q) := \sum_{x\in \X}  \sqrt{P(x) Q(x)} 
    \eeqp
\end{itemize}

We will also need notions of continuity for matrices. First of all, we introduce some matrix norms. Let $M \in \mathbb{C}^{d\times d}$, then:

\begin{itemize}
\item The \textbf{trace norm} of $M$ is
\beq
\lVert M \rVert_1 := \Tr\left[\sqrt{M^{\dagger}M}\right] = \sum_{i=1}^{\text{rank}(M)} \sigma_i(M)
\eeqc
where $\sigma_i(M)$ are the singular values of $M$. 

\item For $p \in [1, \infty)$, the \textbf{Schatten $p$-norm} of $M$ is 
\beq
\lVert M \rVert_p := [\Tr((\sqrt{M^\dagger M}^p)]^{\frac{1}{p}},
\eeq
and $||M||_\infty=\lim_{p\rightarrow\infty}||M||_p$.
\item The \textbf{spectral norm} is the maximum singular value of $M$:

\beq
\lVert M \rVert := \max_{i \in [\text{rank}(M)]} \, \sigma_i (M),
\eeq
which coincides with $||M||_\infty$.

This is also the operator norm induced by the 2-norm for vectors:
\beq\label{eq:spec}
\lVert M \rVert = \max_{\lVert x\rVert_2 = 1, x \in \mathbb{C}^n} \lVert Mx \rVert_2
\eeqp
\end{itemize}
We will also use the following fact:
\begin{fact}\label{fact:1norm2norm}
For any two vectors $u,v \in \mathbb{C}^{n}$ (not necessarily normalized),
\beq
\lVert u v^\dag \rVert_1 = \lVert u\rVert_2 \lVert v \rVert_2
\eeqc
which can be verified by applying definitions. 
\end{fact}

We will need in particular some measures of distance between quantum states and channels, defined as follows. For $\rho, \sigma$ two quantum states in $D(\mathcal{H}^{(d)})$:
\begin{itemize}
\item \textbf{Trace distance:} 
\beq
d_{\text{tr}}(\rho,\sigma) := \frac{1}{2} \lVert \rho - \sigma\rVert_1 = \frac{1}{2}\Tr\left[\sqrt{(\rho -\sigma)^{\dagger} (\rho - \sigma)}\right]
\eeqp
Note that in analogy to Eq.\eqref{eq:TV}, this is half of the trace norm of the difference between $\rho$ and $\sigma$.
\item \textbf{Fidelity:}
\beq
F(\rho,\sigma):= \lVert \sqrt{\rho} \sqrt{\sigma}\rVert_1 = \Tr\left[\sqrt{\sqrt{\rho}\sigma \sqrt{\rho}}\right]
\eeqp
\item \textbf{Bures distance:}
\beq
d_{\text{Bures}}(\rho,\sigma)^2 := 2(1-F(\rho,\sigma))
\eeqp
\end{itemize}

\subsection{Sequential measurements}
The algorithms we will develop make extensive use of sequential measurements, which require to specify what is the state of the quantum system conditioned on the outcome of a measurement. We stick to the usual convention, which states that if the outcome of a POVM $\{E_i\}_{i=1}^r$ is $i$, the state after the measurement is $\rho_{|\sqrt{E_i}}:=\frac{\sqrt{E_i}\rho\sqrt{E_i}}{\Tr[E_i\rho]}$. 
The following facts will be useful. \footnote{This selection is influenced by~\cite{BO21}, as our algorithms are heavily based on their work. We refer the reader to~\cite{BO21} for a more extended presentation. Note also that we refer to lemmas 2.5 and 2.6 appearing in the v2 of the arXiv post.}
\begin{itemize}
\item ~\cite{fuchs1995mathematical, BO21}  Let $P$ be the probability distribution on~$[r]$ determined by the measurement $\mathcal M=\{\Pi_1, \dots, \Pi_r\}$ on~$\rho$, and let $Q$ instead be the distribution determined by $\mathcal{M}$ on $\rho|_{\sqrt{A}}$, where $A=\sum_{i=1}^r a_i \Pi_i$ for some $a_i>0$.  Then
  \begin{equation}F(\rho, \rho|_{\sqrt{A}}) =\frac{\mathbb{E}_{\rho}[\sqrt{A}]}{\sqrt{\mathbb{E}_{\rho}[A]}}= BC(P,Q).\label{eq:27}
  \end{equation}
\item  \cite{sen2012achieving, wilde2013sequential} For each $i = 1,  \dots,  m$, let $\Pi_0^{(i)}$ be projectors and $\Pi^{(i)}_1 = \mathbf{1} - \Pi_0^{(i)}$. The probability $p$ of getting always outcome $1$ for sequential measurements $\{\{\Pi_0^{(1)},\Pi_1^{(1)}\},...,\{\Pi_0^{(m)},\Pi_1^{(m)}\}\}$ on the state $\rho$ satisfies
\begin{align}\label{eq:unionbound}
       p&=\mathbb{E}_\rho\left[({\Pi^{(1)}_1} \dotsm {\Pi^{(m)}_1})({\Pi_1^{(1)}} \dotsm {\Pi^{(m)}_1})^{\dagger}\right] \nonumber\\
       &\geq 1 - 2\sqrt{\sum_{i=1}^m \E_\rho[\Pi^{(i)}_0]}.
\end{align}
\item Lemma 2.5 in~\cite{BO21}: Let $E_1, \dotsc, E_m$ be events and consider a sequence of two-outcome POVMs implementing them, on the state $\rho$, so that the state post-selected on the events occuring is $\rho_{|\sqrt{E_{i-1}}\dotsm\sqrt{E_1}}=\frac{\sqrt{E_{i-1}}\dotsm\sqrt{E_1}\rho\sqrt{E_{1}}\dotsm\sqrt{E_{i-1}}}{\Tr[\sqrt{E_{i-1}}\dotsm\sqrt{E_1}\rho\sqrt{E_{1}}\dotsm\sqrt{E_{i-1}}]}$  Let $q_1=\mathbb{E}_\rho[E_1]$ and $q_i=\mathbb E_{\rho_{|\sqrt{E_{i-1}}\dotsm\sqrt{E_1}}}[E_i]$, for $i>1$, while $p_i=\mathbb E_{\rho}[E_i]>0$ for all $i \in [m]$. We have that denoting $p_{[k]}=\prod_{i\in [k]} p_i$, $q_{[k]}=\prod_{i\in [k]} q_i$, 
  \begin{align}\label{eq:damagelemma}
    |p_{[m]}- q_{[m]}|
    &\le 2 \cdot \sum_{i=1}^{m-1}
      p_{[i]} d_{\text{tr}}(\rho_{|\sqrt{E_i}}, \rho),
  \end{align}
\item Lemma 2.6 in~\cite{BO21}:
  With the same notation as the previous point, let $p_0 = 1$, $\rho_0 = \rho$, and
  $\rho_i = \rho_{i-1}{}_{|\sqrt{E_{i}}}$ for all $i \in
  [m]$.
 
   For all $t \in [m]$, let $s_t$ denote the
  probability of observing outcomes $E_1, \dotsc, E_{t-1}, I-E_{t}$. It
  holds that
  \begin{align}\label{eq:fid-ineq}
    1
    &\le \sqrt{q_m} F(\rho, \rho_m) + \sum_{i=1}^m \sqrt{s_i} \sqrt{1-p_i}.
  \end{align}

Morover, given a subset $\mathcal{B}\in [m]$ such that $1-p_i\leq (\eta/m)^2$ for any $i\in \mathcal{B}$, it follows that (see proof of Lemma 4.2 in~\cite{BO21}):

\begin{equation}\label{eq:unionapp}
1-\sqrt{q_m}F(\rho,\rho_m)-\eta\leq \sqrt{\sum_{i\in[m]/\mathcal B}s_i}\sqrt{\sum_{i\in[m]/\mathcal B}(1-p_i)}.
\end{equation}

\item Fix any POVM $\{E_1,...,E_n\}$, a classical random variable $\mathsf{X}$ with values in $\mathbb R$, and $\theta\in\mathbb{R}$. Furthermore, consider the classical random variable $\mathsf{T}$, defined on a fixed quantum state $\rho$, which has the distribution $\Pr(\mathsf{T}=t)=\mathbb{E}_{\rho}[E_t]$.

The quantum event 
\beq
B=\sum_{t=1}^n\Pr(\mathsf{X}+t> \theta)E_t
\eeq
is such that when measured on a state $\rho$, it is accepted with probability $\mathbb E_{\rho}[B]=\Tr[\rho B] =\Pr(\mathsf{X}+\mathsf{T}> \theta)$. Furthermore, by Eq.~\eqref{eq:27}, $B$ has the property that
\begin{equation}\label{eq:bccond}
F(\rho,\rho_{|\sqrt{1-B}})=BC(\T, (\T|\T+\mathsf{X} \leq \theta)).
\end{equation}
\end{itemize}

\subsection{Naive expectation estimation}
A fundamental result on the concentration of sums of random variables is the following multiplicative Chernoff bound, with $\mathsf{X}_1 \ldots \mathsf{X}_n$ independent  random variables and $0\leq\epsilon\leq1$:
\begin{align}\label{eq:Chernoff2}
\Pr\left(\sum_{i=1}^n \mathsf{X}_i\geq (1+\epsilon) \sum_{i=1}^n \mathbb{E}[\mathsf{X}_i]\right)&\leq e^{-\epsilon^2\sum_{i=1}^n \mathbb {E}[\mathsf{X}_i]/3}\\
\Pr\left(\sum_{i=1}^n \mathsf{X}_i\leq (1-\epsilon)\ \sum_{i=1}^n \mathbb{E}[\mathsf{X}_i]\right)&\leq e^{-\epsilon^2\sum_{i=1}^n \mathbb {E}[\mathsf{X}_i]/2}.\label{eq:Chernoff1}
\end{align}

A consequence of this bound is the following Proposition, related to what has been called \textit{naive expectation estimation} by \cite{BO21}, but with i.i.d. observables. We retain the name: 

\begin{proposition}[Naive expectation estimation]\label{estave}
The random variable $\mathsf{X}$ obtained by measuring the observable $\overline{\Pi}:=
\Pi_{1}\otimes \mathbf{1}\otimes ....\otimes \mathbf{1}+\mathbf{1}\otimes \Pi_2....\otimes \mathbf{1}+\mathbf{1}\otimes ....\otimes \mathbf{1}\otimes \Pi_n$ on $\varrho=\rho_1\otimes...\otimes\rho_n$ satisfies, for $\epsilon<1$,
\begin{equation}
\Pr_{\varrho}[|\mathsf{X}-\mathbb{E}_{\varrho}[\overline{\Pi}]|\geq n\epsilon]\leq 2 e^{-n\epsilon^2/3}.\label{eqchernoff}
\end{equation}

Note that $\mathbb{E}_{\varrho}[\overline{\Pi}]=\sum_{s=1}^n\mathbb E_{\rho_s}[\Pi_s]:=n(1- R(h))$. This implies that, for given $\theta$ and $\epsilon$, and $n\geq\frac{27}{\epsilon^2}\log\frac{2}{\delta}$ this measurement allows to distinguish the two cases
\begin{itemize}
\item $|R(h)-\theta|\leq \epsilon$,
\item $|R(h)-\theta|\geq 2\epsilon$,
\end{itemize}
with probability at least $1-\delta$.
\end{proposition}

\subsection{Growth functions and covering numbers: Measuring the effective size of a concept class}\label{sec:statlearnprel}
In classical statistical learning theory, 
conditions that guarantee the success of ERM to learn an unknown function in the concept class $\C$ are given in terms of effective measures of the size of $\mathcal{C}$. We will be interested in two such measures, the growth function and the covering number.

For a finite sequence $x_1,\cdots,x_n$ and a set of functions $\mathcal C=\{f_{\alpha}:\mathcal X\rightarrow \mathcal Y\}$, we can identify a set of equivalence classes given by grouping together the functions that give the same outcomes on the inputs $x_1\cdots,x_n$. Naturally, if $\mathcal Y$ is finite, the cardinality of this set is finite. However, even if $\mathcal{Y}$ is infinite, $\C$ could still be such that there are only finitely many equivalence classes. Indeed, for a fixed function class, the maximum number of equivalence classes induced by a sample of length $n$ is called the \textit{growth function} $G(n)$. A formal definition of the growth function is as follows~\cite{Anthony1999,Wolfnotes}:
\begin{defn}[Growth function\label{def:growthandrestriction}]
Let $\mathcal{C} \subseteq \mathcal{Y}^{\mathcal{X}}$ be a class of functions with target space $\mathcal{Y}$. For every subset $\Xi \subseteq \mathcal{X}$ define the \textbf{restriction of $\mathcal{C}$} to $\Xi$ as 
\beq
\mathcal{C} \mid_{\Xi}:=\{f \in \mathcal{Y}^{\Xi} \mid \exists F \in \mathcal{C}\, \forall x \in \Xi: f(x)=F(x)\}.
\eeq
The growth function $G$ assigned to $\mathcal{C}$ is then defined for all $n \in \mathbb{N}$ as
$$
G(n):=\max _{\Xi \subseteq \mathcal{X}:|\Xi| \leq n}|\mathcal{C}|_{\Xi} \mid .
$$
\end{defn}

By a slight abuse of notation, we will also write $\mathcal C|_{\vec x}$ for a concept class restricted to a domain given by the set of values appearing in $\vec x$. We will also want to define $\epsilon$-nets over classes of functions and will be interested in their size. 
\begin{defn}[Covering number $N_{in}$] 
For a set $\mathcal{M}$ and a pseudometric~\footnote{The difference between a pseudometric and a metric is that the former is allowed to be zero for pairs of distinct points.} $d:\mathcal{M}\times \mathcal{M} \rightarrow [0,c]$, let $(\mathcal{M}, d)$ be a pseudometric space, let the sets $A, B \subseteq \mathcal{M}$ and fix $\epsilon>0$. 

The set $A\subseteq B$ is an internal $\epsilon$-net of $B$ if $\forall b \in B, \, \exists a \in A: \,d(a, b) \leq \epsilon$. 

The $\epsilon$-covering number of $B$, denoted by 
\begin{equation}
    N_{in}(\epsilon, B, d),
\end{equation} is the smallest cardinality of any internal $\epsilon$-net of $B$.
\end{defn}

The first pseudometric that will be of interest to us is the one built on the $\lVert \cdot \rVert_{p,\vec{x}}$ seminorm. 

\begin{defn}[$\lVert \cdot \rVert_{p,\vec{x}}$ seminorms]
For any set $\mathcal{X}, \vec{x} \in \mathcal{X}^{n}$ and any function class $\mathcal{G} \subseteq [0,c]^{\mathcal{X}}$ define the $\lVert \cdot\rVert_{p, \vec{x}}$-seminorm on the linear span of $\mathcal{G}$ for $p \in[1, \infty)$ as
\beq
\lVert g\rVert_{p, \vec{x}}:=\left(\frac{1}{n} \sum_{i=1}^{n}|g(x_{i})|^{p}\right)^{1 / p}, 
\eeq
and note that $\lVert g\rVert_{\infty, \vec{x}}:=\max _{i\in[n]}\left|g\left(x_{i}\right)\right|$.
The corresponding pseudometric is then given by $d: (g_{1}, g_{2})\mapsto \lVert g_{1}-g_{2}\rVert_{p, \vec{x}}$. 
\end{defn}

\begin{defn}[Loss-function covering numbers]
Let $\mathcal{G} \subseteq[0,c]^{\mathcal{X}}$ be a class of real-valued functions. For positive integer $p$, the $p$-th loss-function covering number is 
\beq
\Gamma_{p}(n, \epsilon, \mathcal{G}):=\max \left\{ N_{in}\left( \epsilon, \mathcal{G},|| \cdot ||_{p,\vec{x} } \right)| \vec{x} \in \mathcal{X}^{n}  \right\}
\eeqp
\end{defn}

We make three remarks: 
\begin{enumerate}
    \item For a finite target space $\mathcal{Y} = 1, 2, \ldots |\mathcal{Y}|$, and for any $\vec{x} \in \X^n$, observe that $||g-f||_{\infty,\vec{x} } < \eps < 1$ iff $g|_{\vec{x} } = f|_{\vec{x}} $. This follows from the definition of spectral norm. Thus

\begin{equation}
    N_{in}\left( \epsilon, \mathcal{F},|| \cdot ||_{\infty,\vec{x} } \right)  = \left|\mathcal{F}|_{\vec{x} } \right|
\end{equation}
and hence 
\begin{equation}
    \Gamma_{\infty}(n,\eps,\mathcal{F}) = G(n).
\end{equation}

\item The name of the above covering number comes from the fact that we will often be interested in choosing $\G$ to be the \textit{induced loss function class}, defined below:

\begin{defn}[Induced loss function class]\label{def:G}
For any function class $\mathcal{F} \subseteq \mathcal{Y}^{\mathcal{X}}$ and loss function $L: \mathcal{Y} \times \mathcal{Y} \rightarrow[0, c]$ define 
\begin{align}
    \mathcal{G}_{\F,L}&:=\{g: \mathcal{X} \times \mathcal{Y} \rightarrow[0, c] \nonumber\\
    &\mid \exists h \in \mathcal{F}: g(x, y)=L(y, h(x))\}.
\end{align}
\end{defn}
On a related note, observe that, when the loss function is $L(y,x)=|x-y|$, for any $h_1, h_2 \in \F$, their pseudodistance $\lVert h_1 - h_2\rVert_{1,\vec{x}}$ upper-bounds the difference between their empirical risks, $\hat{R}(h_1) - \hat{R}(h_2)$. This fact will be crucial in the next few paragraphs.
\end{enumerate}

We will also need to define covering numbers for classes mapping to operators instead of real intervals (in the following, $\mathcal{H}$ denotes some fixed Hilbert space). For this case we will use a different seminorm, defined for operator-valued functions: 
\begin{defn}[$\lVert \cdot \rVert_{p,q,\vec{x}}$ seminorms]
For any set $\mathcal{X}, \vec{x} \in \mathcal{X}^{n}$ and any function class $\mathcal{C} \subseteq \{h:\X\rightarrow M\}, \text{ where } M\subseteq \mathcal{L}(\mathcal{H})$, we define the $\lVert \cdot\rVert_{p, q,\vec{x}}$-seminorm on the linear span of $\mathcal{C}$ for $p \in[1, \infty)$, $q \in[1, \infty)$ as
\beq
\lVert g\rVert_{p,q, \vec{x}}:=\left(\frac{1}{n} \sum_{i=1}^{n}||g(x_{i})||_q^{p}\right)^{1 / p}, 
\eeq
where $\lVert \cdot \rVert_q$ is a Schatten $q$-norm, and note that $\lVert g\rVert_{\infty,q, \vec{x}}:=\max _{i\in[n]}||g(x_{i})||_q$.
The corresponding pseudometric is then given by $d: (g_{1}, g_{2})\mapsto \lVert g_{1}-g_{2}\rVert_{p, q,\vec{x}}$. 
\end{defn}

Accordingly, we will define the following covering number:
\begin{defn}[Operator-class covering numbers]
Let $\mathcal{C} \subseteq \mathcal{L}(\mathcal{H})^{\mathcal{X}}$ be a class of operator-valued functions. For positive integers $p,q$, the $(p,q)$-th operator-class covering number is 
\beq
\gamma_{p,q}(n, \epsilon, \mathcal{C}):=\max \left\{ N_{in}\left( \epsilon, \mathcal{C},|| \cdot ||_{p,q,\vec{x}} \right)| \vec{x} \in \mathcal{X}^{n}  \right\}
\eeqp
\end{defn}

For the concept classes of interest $\mathcal{C}: \mathcal{C} \subseteq \{h:\X\rightarrow M\}$ where $M\subseteq \mathcal{L}(\mathcal{H})$, we elect to use the following natural operator-class covering numbers, which we can also relate to the corresponding loss-function covering numbers:
\begin{itemize}
    \item \textbf{(Case 1: Quantum states)} In the case of $M$ being quantum states, we can take  $p=1,q=1$. The pseudodistance between two concepts becomes twice the average trace distance between the states,
    $||g_1-g_2||_{1,1,\vec x}=\sum_{i=1}^{n}\frac{1}{n}||g_1(x_i)-g_2(x_i)||_1$. Note that the empirical risk (defined via $L_s$) evaluated on the input source and the concepts $g_1$ and $g_2$ satisfies
    \begin{align}\label{contstates}
        &|\hat{R}_s(g_1)-\hat{R}_s(g_2)|\nonumber\\
        &= \frac{1}{n} \left|\sum_{i=1}^n L_s(\rho(x_i),g_1(x_i))-\sum_{i=1}^n L_s(\rho(x_i),g_2(x_i))\right|\\
        &= \frac{1}{2n}\left|\sum_{i=1}^n \left(\lVert g_1(x_i) - \rho(x_i) \rVert_1 -\lVert g_2(x_i) - \rho(x_i) \rVert_1\right)\right| \\
        &\leq \frac{1}{2}||g_1-g_2||_{1,1,\vec x}.
    \end{align}
    where the inequality follows from the triangle inequality. This means that if $N_{in}(\C, \eps, || \cdot ||_{1,1,\vec{x}} )$ is an $\eps$-net for $\C$, then one can construct an $\eps$-net for $\G_{\C,L_s}$, simply by associating every point $c\in N_{in}(\C, \eps, || \cdot ||_{p,q,\vec{x}})$ with a point $g_c = L_s \circ c$ which is in $\G_{\C,L_s}$. This implies in turn that 
    \beq \label{eq:gammavsGamma_s}
    \Gamma_1(n, \eps, \G_{\C,L_s}) \leq \gamma_{1,1}(n, 2\eps, \C)\leq \gamma_{1,1}(n, \eps, \C)
    \eeqp
    \item \textbf{(Case 2: Projectors)} In the case of $M$ being projectors, we can take  $p=1,q=\infty$. The pseudodistance between two concepts becomes the average operator norm distance between the projectors,
    $||g_1-g_2||_{1,\infty,\vec x}=\sum_{i=1}^{n}\frac{1}{n}||g_1(x_i)-g_2(x_i)||_{\infty}$. Note that the empirical risk (defined via $L_p$) evaluated on the input source and the concepts $g_1$ and $g_2$ satisfies 
    \begin{align}\label{contproj}
        &|\hat{R}_{p}(g_1)-\hat{R}_{p}(g_2)|\\
        &=\left|\sum_{i=1}^n\frac{1}{n}L_p(\rho(x),g_1(x))-\sum_{i=1}^n\frac{1}{n}L_p(\rho(x),g_2(x))\right|\\
        &\leq ||g_1-g_2||_{1,\infty,\vec x}.
    \end{align}
    Similarly to the above, we can conclude that 
    \beq \label{eq:gammavsGamma_p}
    \Gamma_1(n, \eps, \G_{\C,L_p}) \leq \gamma_{1, \infty}(n, \eps, \C)
    \eeqp
\end{itemize}
Note that both properties above are a generalization of the decreasing property of covering numbers for real-valued function classes under composition with Lipschitz functions~\cite{Anthony1999}. We need operator-class covering numbers since we can construct $\epsilon$-nets with respect to the $||\cdot||_{p,q,\vec{x}}$ seminorm, but not with respect to the loss function directly, since the value of the loss cannot be obtained immediately from the data.

\subsection{Statistical learning theorems}
 \label{sec:statlearntheo}

The principle of empirical risk minimization dictates that, once a loss function $L$ with values in the interval $[0,c]$ has been fixed, in order to find a function in the concept class that minimizes the true risk relative to the unknown concept $f$, 
\begin{equation}\label{eq:truerisk_2}
R_f(h):= \mathbb{E}_{x \sim \mathcal D} \left[L(f(x),h(x))\right] \qquad \text{(True risk)},
\end{equation}
it suffices with high probability to find one that minimizes the empirical risk over the given sample:
\begin{equation}\label{eq:risk_2}
\hat{R}_f(h):=\frac{1}{n}\sum_{i=1}^n L(f(x_i),h(x_i))\qquad \text{(Empirical risk)}
\end{equation}
In the classical case, the utility of this principle comes from the fact that the true risk cannot be estimated directly from the sample, but the empirical risk can. Crucially, the quality of the approximation and the rate of the above-mentioned convergence depend on the concept class under study. 

In particular, \textit{discrete-output} concept classes are very well-understood. In this special case, the rate of convergence of the empirical risk to the true risk is quantified by growth functions. Namely, a well-known fact (see for instance~\cite{Wolfnotes,Anthony1999, vapnik1999overview}) states that given a sample $S = (x_i, f(x_i))_{i=1}^n$, with probability at least $(1-\delta)$ w.r.t. repeated sampling of training data of size $n$ we have
\beq
\forall h \in \mathcal{F}:|R_f(h)-\hat{R}_f(h)| \leq c \sqrt{\frac{8 \ln \left(G(2 n) \frac{4}{\delta}\right)}{n}}
\eeqp
However, in our setting, we would like to be able to characterize channels that map to a continuous set of states. Since the restrictions of such concept classes to a finite sample are potentially infinite, the above statement is not useful. Nevertheless, a crucial observation is that, even for infinite concept classes $\F$, we can obtain a PAC bound for learning $\F$ via the covering number of a class $\G$ related to $\F$. $\G$ is induced by composing the chosen loss function with the functions in $\F$, as earlier stated in Definition \ref{def:G}.

The learning theorem is then:
\begin{theorem}[PAC bound via uniform covering numbers (reported as in~\cite{Wolfnotes}, also equivalent to Theorem 17.1 in~\cite{Anthony1999}]~\label{convergence_proj}
For any concept class $\F$ and loss function $L$ with values in $[0,c]$, define $\G_{\F,L}$ as in Def \ref{def:G}. For any $\epsilon>0$ and any probability measure ${\cal D}$ on $\mathcal{X} \times \mathcal{Y}$ it holds
\begin{align}\label{fundamentalgrowth}
&\Pr_{S \sim \mathcal{D}^{n}}[\exists h \in \mathcal{F}:|R(h)-\hat{R}(h)| \geq \epsilon] \\&\leq 4 \Gamma_{1}(2 n, \epsilon / 8, \mathcal{G}_{\F,L}) e^{-\frac{n \epsilon^{2}}{32 c^{2}}},
\end{align}
 where $S = ((x_i,y_i))_{i=1}^n$ is the training sample and $(x_i,y_i) \sim \mathcal{D}$:
\end{theorem}

That is to say, the empirical risk $\hat{R}$ converges to the true risk $R$ depending on the speed of growth of $\Gamma_1$. 

As noted above, for the loss functions we consider the covering number $\Gamma_1(n,\epsilon,\mathcal G_{\mathcal F,L})$ can be bounded by the corresponding operator-class covering number $\gamma_{1,q}(n,\epsilon,\mathcal C)$, which effectively controls our learning algorithms.

\section{Threshold search for non-identical states}\label{sec:trnoniid}

Recall the discussion in the introduction about the difficulties encountered in naively generalizing the strategy of classical ERM to our quantum setting: they stem from not having an arbitrary number of identical copies of a quantum state $\rho(x_i)$ and from not being able to naively estimate all the empirical risks at the same time. In this section, we nevertheless introduce a tool for performing our ``quantum" version of ERM: an adaptation of \textit{quantum threshold search} \cite{Aaronsonshadow20,BO21}. The gist of the algorithm is to use identical copies of some state of interest $\rho$, to pick an (observable, threshold) pair from amongst a set of such pairs, with the property that the observable exceeds the threshold on $\rho$ -- or, reports that no such pair is available. We define it more formally:

\begin{problem}[Quantum Threshold Search \cite{Aaronsonshadow20,BO21}]
Given as input
\begin{itemize} 
    \item Parameters $\eps, \delta > 0$
    \item Access to  \textit{identical} copies of a state $\rho\in D(\mathcal{H}^{(d)})$
    \item A collection of known projectors $\{\Pi_c\}_{c=1}^m$ where $\Pi_c \in \mathcal{L}(\mathcal{H}^{(d)})$.
    \item A collection of known thresholds $\{\theta_c\}_{c=1}^m$ where $\theta_c \in [0,1]$.
\end{itemize}
Output either 
\begin{itemize}
    \item $\Tr[\rho \Pi_c] > \theta_c-\eps$ for some particular $c$,
    \item $\Tr[\rho \Pi_c] \leq \theta_c$ for all $c$,
\end{itemize}
with probability of a correct statement at least $1-\delta$.
\end{problem}

This section is devoted to presenting an algorithm for threshold search that relaxes the need for access to identical copies of $\rho$, solving the same task on \textit{general product states}. Intuitively, this is necessary for our setting because the learner, lacking control over the input to the process, receives as examples the sequence $(x_i, \rho(x_i))_{i=1}^n$ where $\rho(x_i)$ are not identical to each other. Thus, in contrast to the measurement in \cite{BO21} which works on $\rho^{\otimes n}$, our key tool will be a measurement on the product state $\rho(x_1)\otimes \ldots \otimes\rho(x_n)$. The measurement reports if a threshold has been exceeded. We will eventually use this measurement to perform ERM for our setting. We define the properties of this measurement in Lemma \ref{lem:threshold}. Our general proof strategy follows closely that of~\cite{BO21}, adapting when needed.

\begin{lemma} [Quantum threshold search on non-identical states\label{lem:threshold}]
Given as input
\begin{itemize} 
    \item Parameter $1>\eps > 0$
    \item Access to \textit{a single} product state $\varrho := \rho_1\otimes \ldots \otimes \rho_n \in (D(\mathcal{H}^{(d)}))^{\otimes n}$ which is a product of generally \textit{non-identical} qudit states.
    \item A collection of lists of known projectors
    \beq
    \{\Pi_{1}^{(c)} , \ldots, \Pi_{n}^{(c)}\}_{c=1}^m
    \eeq
    where each $\Pi_{i}^{(c)} \in \mathcal{L}(\mathcal{H}^{(d)})$.
    \item A collection of known thresholds $\{\theta_c\}_{c=1}^m$ where $\theta_c \in [0,1]$.
\end{itemize}
If the projectors and thresholds obey the promise that 
\beq\label{eq:promise2}
\frac{1}{n}\sum_{i=1}^n\Tr[\Pi_{i}^{(c)}\rho_i] > \theta_c
\eeq
for at least one~$i$, there is an algorithm such that: 
  \begin{itemize}
      \item At each step $c=1,...,m$ it performs a two-outcome measurement $\{B_c,\mathbf{1}-B_c\}$ based on the projectors $\{\Pi_{1}^{(c)}, \ldots, \Pi_{n}^{(c)}\}$.
     { \item If the measurement accepts $B_c$, then the algorithm halts and outputs $c$, otherwise it passes to $c+1$.}
     
  \end{itemize}
  This algorithm is such that if 
  \begin{equation}
  (\log m+C_2)^2<C_1 n \epsilon^2
  \end{equation}
  for appropriate constants $C_1,C_2>0$, it outputs $c$ such that 
  \begin{equation}
  \frac{1}{n}\sum_{i=1}^n\Tr[\Pi_{i}^{(c)}\rho_i]  \geq \theta_c-\epsilon
  \end{equation}
  with probability at least~$0.03$. 
  
\end{lemma}

In order to get this improvement, we need to generalize the measurement constructed in~\cite{BO21}, which in a certain sense is gentler than the projectors which can check if the expectation is above or below threshold from Proposition~\ref{estave}. Roughly speaking, gentle means that the state does not change much after the measurement, with high probability. Implementations of product measurements that are gentle on product states had already been obtained in~\cite{Aaronsonshadow20}, and then reconsidered by~\cite{BO21} with a simpler analysis on identical states, which gave an improvement in the sample complexity of shadow tomography~\cite{BO21}. In these implementations, a parameter $\lambda$ can be increased in order to obtain a gentler version of a projective measurement, sacrificing the amount of information that the measurement can reveal.
We go one step further and generalize the stronger statements of~\cite{BO21} so that they apply to our setting, which involves products of non-identical states. 

To do this we need to study the behavior, under perturbation, of a random variable which is the sum of non-identical Bernoulli random variables:
\beq\label{eq:T}
\mathsf{T}=\sum_{i=1}^n \mathsf{T}_i \qquad \text{where } \mathsf{T}_i\sim (p_i,1-p_i)
\eeq
The distribution of $\mathsf{T}$ is also known as the \textit{Poisson binomial distribution} and we denote this as $T \sim PB(p_1,\ldots p_n)$. 

We also need to consider exponential random variables $\mathsf X$, which have a density $p_{\lambda}(x):=\lambda e^{-\lambda x}$, for some $\lambda>0$, and satisfy $\mathbb E [\mathsf X]=1/\lambda$.

{ We first present an entirely classical Theorem. This Theorem says that a measurement that checks for the event that $\mathsf{T}+\mathsf{X}$ exceeds a certain threshold is gentle when it rejects, in the sense that it only \textit{gently} perturbs the distribution of $\mathsf{T}$.}  
This is a generalization of Theorem 1.2 of~\cite{BO21}, which shows that adding exponential noise to a \textit{binomial} random variable allows for gentle measurements. 

\begin{theorem}[Gentle classical measurement on Poisson binomial distribution]\label{bernoulli_new}
    Let $\T\sim PB(p_1,\ldots, p_n)$, and write $q_i = 1-p_i$. Assume that $\mathsf{X}$ is an independent exponential random variable
    with mean $1/\lambda$ at least~$\mathrm{stddev}[\mathsf{T}] = \sqrt{\sum_{i=1}^{n} p_i q_i}$ (and also at
    least~$1$).  
    Let $\mathfrak{B}$ be the event that $\mathsf{T} + \mathsf{X} > \theta n$, and assume that $\Pr[\mathfrak{B}] < \frac14$.
    Then
    \[
        \dist[\chi^2]((\mathsf{T} \mid \overline{\mathfrak{B}}), \mathsf{T}) \leq C
        \left(\Pr[\mathfrak{B}] \cdot \frac{\mathbf{\mathrm{stddev}}[\mathsf{T}]}{\mathbb{E}[\mathsf{X}]}\right)^2,
    \]
    for a sufficiently large constant $C>0$.
\end{theorem}
\begin{proof}
See Appendix \ref{bernoulli_new_proof}. The proof of Theorem \ref{bernoulli_new} requires to modify some key details of the original proof. The main observation is that the crucial properties of the binomial distribution used in the proof, i.e., the form of the generating function and the probability of exceeding the expectation, also hold, with some caveats, for the Poisson binomial distribution.
\end{proof}

Now let us turn to the quantum problem. The reason we must consider Poisson binomial random variables is that they describe the probability distribution of outcomes when measuring a sum of (possibly non-identical) local projectors on a product quantum state. The link to the learning setting is that the local projectors are exactly the ones in the set $\{\Pi_1^{(c)}, \ldots, \Pi_n^{(c)}\}$, which describe the $c$-th hypothesis under consideration, in a way we will make formal in the next section.

In the agnostic learning setting, the learner receives the string of classical-quantum examples $(x_1, \rho(x_1)) \otimes \ldots \otimes (x_n, \rho(x_n))$. The quantum part of these examples are $n$ non-identical quantum states. We assume there is at least one hypothesis whose empirical risk on the examples goes below some threshold; this is captured by the promise in Lemma \ref{lem:threshold}, that there is at least one set of projectors $\{\Pi_1^{(c)} ,\ldots, \Pi_n^{(c)}\}$ such that 
\beq\label{eq:guarantee}
\frac{1}{n}\sum_{i=1}^n \Tr[\Pi_i^{(c)} \rho_i] > \theta_c
\eeqp
In order to perform ERM, therefore, it suffices to find some $c$ satisfying such a guarantee. The next theorem (Theorem \ref{theo:Bbiid}) shows that the learner can perform a gentle \textit{quantum} measurement, in the sense we next define, on the product state, in order to find such $c$ with high probability. This is also the measurement referred to in Lemma \ref{lem:threshold}. It is sequential and adaptive, where the learner is presented with a projectors-threshold pair $(\{\Pi^{(c)}_1,\ldots,\Pi^{(c)}_n\}, \theta_c)$ at each step, and responds by making an appropriate measurement on her state (which is reused for multiple measurements) depending on past measurement outcomes. The $c$-th measurement depends on the presented projectors-threshold pair and the outcomes of the first $c-1$ measurements. 

Theorem \ref{theo:Bbiid} should be viewed as a quantum counterpart to the classical Theorem \ref{bernoulli_new}, as it introduces the gentle quantum measurement that is at the core of our learning algorithm. Here we show that to every list of projectors $\{\Pi_1 ,\ldots, \Pi_n\}$, one can associate a gentle quantum observable $B$ that, similarly, only perturbs a quantum state by a small amount, if a measurement of $B$ rejects. In Theorem \ref{bernoulli_new}, the gentleness was in the sense that $d_{\chi^2}((\mathsf{T}|\overline{\mathfrak{B}}),\mathsf{T})$ was bounded by $\Pr(\mathfrak{B})^2$. In the quantum Theorem \ref{theo:Bbiid}, the gentleness is in the sense that the fidelity between the pre- and post-rejecting-measurement quantum state is exactly  $BC((\mathsf{T}|\overline{\mathfrak{B}}),\mathsf{T})$: this is analogous to Lemma 3.4 of \cite{BO21}. The two classical distance measures on probability distributions that we have mentioned are related via the inequality 
\beq\label{ineq:bcchi}
1 - BC(P, Q) \leq d_{\chi^2}(P, Q)
\eeq
for probability distributions $P,Q$.

\begin{theorem}[Gentle quantum measurements on non-identical product states] \label{theo:Bbiid}
Fix $0\leq \theta \leq 1$. Let $\sf{X}$ be any classical random variable taking values in $[0,\infty)$. For any list of projectors  $\{\Pi_{1}, \ldots, \Pi_n \},\, \Pi_i\in \mathcal{L}(\mathcal{H}^{(d)})$, there exists a quantum event $B\in L((\mathcal{H}^{(d)}){}^{\otimes n}))$ such that when $B$ is measured against a product state $\varrho$, i.e.:
\begin{equation}
    \varrho := \rho_1\otimes \ldots \otimes\rho_n,
\end{equation}
we have
\beq\label{eq:fidBC}
\mathbb{E}_{\varrho}[B] =\Pr(\mathsf{T} + \mathsf{X}> \theta n),
\eeq
where $\T\sim PB(p_1,\ldots, p_n)$ and $p_i = \mathbb{E}_{\rho_i}[\Pi_i]$. 
Furthermore,
\beq
F(\varrho , \varrho|_{\sqrt{\mathbf{1}-B}}) = BC((\T|\T+\mathsf{X} \leq \theta n),\T)
\eeqp

\end{theorem}

\begin{proof}
Observe that we may write the distribution of $\mathsf{T}$ explicitly as 
\beq\label{eq:PB}
\Pr[\mathsf{T}=t] = \sum_{\vec{v} \in \01^n: ||\vec{v}||_1=t} \prod_{i=1}^n p_i^{v_i} (1-p_i)^{(1-v_i)}
\eeqp

Then consider the projectors $\{E_k\}_{k=0}^{n}$, 
\begin{equation}
E_k:= \sum_{\vec v\in \{0,1\}^n: ||\vec{v}||_1=k} \otimes_{i=1}^{n}(\Pi_i+v_i(\mathbf{1}-2\Pi_i)).
\end{equation}

It is easy to see that the event $B$ 
\beq
B := \sum_{t=1}^n \Pr(\mathsf{X} + t > \theta n ) E_t
\eeqp
fulfills Eq.~\eqref{eq:fidBC}. The `furthermore' part of the statement of the Lemma follows immediately from Eq.~\eqref{eq:bccond}.
\end{proof}

Theorem \ref{theo:Bbiid} applies to $\mathsf{X}$ being any classical random variable. Now we specialize to $\mathsf{X}$ being an exponential random variable. This specification allows us  
to obtain a bound on the Bures distance between the state before and after the rejecting measurement, together with a bound on the probability of $B$. {\color{black} For the latter, at variance with~\cite{BO21}, we employ here a second-order Lagrange remainder that allows us to bypass the reduction to a fixed threshold, thus simplifying the overall proof procedure.}

\begin{theorem}[Gentle quantum measurements on non-identical product states via exponential noise]\label{thm:chi-stable-final}
Let $\varrho := \rho_1\otimes \ldots \otimes \rho_n \in L((\mathcal{H}^{(d)})^{\otimes n})$ and consider a list of projectors $\{\Pi_{1}, \ldots, \Pi_n\},\, \Pi_i\in \mathcal{L}(\mathcal{H}^{(d)})$. Let $1>\lambda>0$, and let $\theta \in[0,1]$ be an arbitrary threshold. Let $\T\sim PB(p_1,\ldots, p_n)$ where $p_i=\Tr[\rho_i\Pi_i]$, and let the exponential random variable $\mathsf{X}$ be as in Theorem \ref{bernoulli_new}.

Then there exists a quantum event $B \in\left(\mathbb{C}^{d \times d}\right)^{\otimes n}$ such that 
\beq\label{eq:gentle}
\mathbb{E}_{\varrho}[B]=\Pr(\T+{\sf{X}}>\theta n)
\eeqc
and
\beq \label{eq:bures}
d_{\mathrm{Bures}}\left(\varrho,\varrho|_{\sqrt{\mathbf{1}-B}}\right) \lesssim \mathbb{E}_{\varrho}[B] \cdot \frac{\operatorname{stddev}[\T]}{\Exp[\mathsf{X}]}
\eeqp
Moreover,
\beq\label{eq:boundexp}
\Exp_{\varrho}[B] \leq \exp{\left(-n \lambda\left(\theta-\frac{1}{n}\sum_{i=0}^n p_i{\color{black}-e\frac\lambda2}\right)\right)}
\eeqp
\end{theorem}

\begin{proof}
Eq.~\eqref{eq:bures} is an easy consequence of Theorems \ref{bernoulli_new}, \ref{theo:Bbiid} and Eq.~\eqref{ineq:bcchi}. 

Moreover, Eq.~\eqref{eq:boundexp} follows from 
\begin{align}
\mathbb{E}_{\varrho}[B]&:=\Pr[\T+\mathsf{X}>\theta n] \\
&= \sum_{t=0}^n \Pr(\mathsf{T}= t) \Pr(\mathsf{X} > \theta n-t)\\
&\leq \Exp[\exp (-\lambda(\theta n-\T))]\\
&=\exp (-\lambda \theta n) \Exp[\exp (\lambda \T)]\\
&=\exp (-\lambda \theta n) \prod_{i=0}^n ( 1- p_i + p_i e^{\lambda})\\
&\leq \exp (-\lambda \theta n)\left(\frac{1}{n} \sum_{i=0}^n (1+ p_i (e^{\lambda}-1))\right)^{n}\\
&\leq\exp (-\lambda \theta n)\left(1+ \frac{\lambda}{n}\sum_{i=0}^n  p_i + e\lambda^2/2\right)^{n} \\
&\leq \exp (-\lambda \theta n) \exp \left(\lambda \sum_{i=0}^n  p_i+en\lambda^2/2\right)\\
&= \exp \left(-n \lambda\left(\theta-\frac{1}{n}\sum_{i=0}^n p_i-e \lambda/2\right)\right)
\end{align}
as desired. In the first inequality, we used that $\Pr[\mathsf{X}>t] \leq \exp (-\lambda t)$; in the fourth equality we used the MGF of a Poisson binomial random variable; in the second inequality we used the AM-GM inequality; in the third inequality we used that, from the second-order Lagrange remainder:
$$e^\lambda=1+\lambda+\lambda^2\frac{e^{t^*}}{2}\leq 1+\lambda+e\lambda^2/2,$$
with $t^*\in[0,\lambda]$.
Finally, in the fourth inequality, we used that $1+x \leq e^{x} \text { for } x \in \mathbb{R}$.
\end{proof}

With these ingredients, we can prove Lemma~\ref{lem:threshold}. The algorithm that solves quantum threshold search is as follows:

\begin{algorithm}[H]
\textbf{Parameters}: $\eps,\delta, D > 0$. \\
\textbf{Promise}: Of the pairs $\{(\Pi_{1}^{(c)}, \ldots, \Pi_{n}^{(c)}, \theta_c)\}_{c=1}^m$, there is some $c\in [m]$ such that \beq\label{eq:promise1}
\frac{1}{n}\sum_{i=1}^n\Tr[\Pi_{i}^{(c)}\rho_i] > \theta_c
\eeq
\textbf{Initialize}: $\varrho^{(0)} \leftarrow \rho_1\otimes \ldots \otimes \rho_n$.
\vspace{5pt}
\begin{algorithmic}[1]
\For{$c = 1, \ldots, m$}
\vspace{1mm}
    \State Upon being presented with the pair 
    \beq
    (\{\Pi_{1}^{(c)}, \ldots , \Pi_{n}^{(c)}\}, \theta_c)
    \eeqc
    measure the two-outcome POVM $\{B_{c},\overline B_{c}:=\mathbf{1}-B_{c}\}$, constructed as in Theorem~\ref{theo:Bbiid} with threshold $\theta_c-\epsilon$ and with $1/\lambda=D\sqrt{n}$, on $\varrho^{(c-1)}$. \;
    \State If the measurement accepts, output $c$ and break. If the measurement rejects, let $\varrho^{(c)}$ denote the post-measurement quantum state. \;
\EndFor
\State If none of the measurements were accepted, output ``pass on all". 
\end{algorithmic}
\caption{Quantum threshold search on non-identical states ($\mathsf{\mathsf{ThresholdSearch}}$)\label{alg:threshold}}
\end{algorithm}
In Lemma \ref{lem:threshold}, we claimed that Algorithm \ref{alg:threshold} is such that if 
  \begin{equation}
  (\log m+C_2)^2<C_1 n \epsilon^2,
  \end{equation}
 for an appropriate constants $C_1$ and $C_2$, (and when $\epsilon<1$), with probability at least 0.03, the algorithm halts and outputs a projector exceeding the threshold, i.e. it halts on $i$ such that \beq\label{eq:exceeds}
  \frac{1}{n}\sum_{i=1}^n\Tr[\Pi_{i}^{(c)}\rho_i]  \geq \theta_c-\epsilon
  \eeqp
We now prove this. We will adopt the following notation, which mirrors the notation in~\cite{BO21}. For $c = 1, \dotsc, m$, let:
  \begin{enumerate}
  \item $\mathsf{T}^{(c)}=\sum_{i=1}^n \mathsf{T}^{(c)}_i$, with $\mathsf{T}^{(c)}_i\sim (\mathbb{E}_{\rho_{i}}[\Pi_{i}^{(c)}],1-\mathbb{E}_{\rho_{i}}[\Pi_{i}^{(c)}])$ Bernoulli variables,\\ implying  $\mathsf{T}^{(c)} \sim PB(\mathbb{E}_{\rho_{1}}[\Pi_{1}^{(c)}],\ldots, \mathbb{E}_{\rho_{n}}[\Pi_{n}^{(c)}])$. 
  \item $p_c = \mathbb{E}_{\varrho}[B_c]$;
  \item $\varrho^{(0)} = \rho_1\otimes \ldots \otimes \rho_n$ and $\varrho^{(c)}$ be the
    quantum state after the $c$-th measurement, \textit{conditioned} on the event $\overline{B}_j$ occurring for all $1 \le j \le c$;
  \item {$r_c = \mathbb{E}_{\varrho^{(c-1)}}[\overline{B}_c]$} 
  be the probability that the
    event $\overline{B}_i$ occurs assuming all the events $\overline{B}_j$ with $1 \le j
    \le c - 1$ occurred;
  \item $q_c = r_1 \dotsm r_c$ be the probability that \textit{all} of the
    events $\overline{B}_j$ with $1 \le j \le c$ occur;
    \item $s_c = q_{c-1}\cdot \mathbb{E}_{\varrho_{c-1}}[B_c]$ be the probability of observing outcomes $\bar B_1 \cdots \bar B_{c-1} B_c$.
  \end{enumerate}
  
As in~\cite{BO21}, our proof works by bounding the probabilities of two bad events. We now state what these events are and give an intuition for why their probabilities should be bounded, before we formalize the intuition. 
\begin{enumerate}
    \item The algorithm outputs a \textit{false negative}: it passes on all projector-threshold pairs, even though there was one that fulfilled the promise. 
    \item The algorithm outputs a \textit{false positive}: it outputs $c$ that actually doesn't fulfil the promise, i.e.
    \beq
    \frac{1}{n}\sum_{i=1}^n\Tr[\Pi_{i}^{(c)}\rho_i] \leq \theta_c - \eps
    \eeq
\end{enumerate}
\textit{Intuition for why these two probabilities are bounded}: Suppose hypothetically that the algorithm had the luxury to measure the given projector on $R$ \textit{fresh} copies of $\rho_1\otimes \ldots \otimes \rho_n$ each time. Then the Promise of Lemma \ref{lem:threshold}, and Chernoff's bound, ensure that with $R=O(1/\eps^2)$-many copies, one could identify some projector fulfilling \eqref{eq:exceeds} with high probability, so that neither false negative nor false positive would come to pass. 

However, in Algorithm \ref{alg:threshold}, we do not use $R$ fresh copies on each of the $m$ iterations of the For loop; instead, we make successive measurements on a single copy. Then the Damage Lemma (Eq.~\eqref{eq:damagelemma}), together with the properties of the specially-constructed $B_c$ in Theorem \ref{theo:Bbiid}, ensures that the usage of a `damaged' copy does not affect the probability of acceptance on a new measurement too much, relative to using $R$ fresh copies. The caveat is that `low damage' is only guaranteed if all of the previous measurements have rejected; care must then be taken to account for this. 

\begin{proof}[Proof of Lemma~\ref{lem:threshold}]

\textbf{Controlling the false negative probability} ---
The promise of Lemma \ref{lem:threshold} is that there is some $c\in [m]$ such that \beq\label{eq:promise3}
\frac{1}{n}\sum_{i=1}^n\Tr[\Pi_{i}^{(c)}\rho_i] > \theta_c
\eeqp
For this particular $c$, let us bound the probability of acceptance of the corresponding measurement $B_c$ (which we previously denoted $p_c$).

To do so, let us recall that Theorem \ref{thm:chi-stable-final} guarantees that the $B_c$'s are constructed in such a way that their measurement outcome statistics follow from that of a classical Poisson binomial random variable, $\mathsf{T}^{(c)} := \sum_{i=1}^n \mathsf{T}^{(c)}_i$. This is a sum of independent and not identically distributed binary random variables. Their concentration is captured by the multiplicative Chernoff bound in Eqs.~\eqref{eq:Chernoff2} and~\eqref{eq:Chernoff1}. 
Applying this to the sum $\mathsf{T}^{(c)} := \sum_{i=1}^n \mathsf{T}^{(c)}_i$, which has expectation $\mathbb{E}[\mathsf{T}^{(c)}] =n\overline p_c := \sum_{i=1}^n \Tr[\rho_i\Pi_{i}^{(c)}]$, we have, by the promise, $\theta_c \leq \overline{p}_c$ and for any non-negative random variable $\mathsf{X}$:
\begin{align}
 p_i &:= \mathbb{E}_{\varrho}[B_c]=\Pr(\mathsf{T}^{(c)}+\mathsf{X}>(\theta_c-\epsilon) n)\\
 &\geq \Pr(\mathsf{T}^{(c)}>(\theta_c-\epsilon)n)\geq \Pr(\mathsf{T}^{(c)}>(\overline{p}_c-\epsilon)n)\\
 &\geq 1- e^{-n \epsilon^2/(2\overline{p_i})}\geq 1-e^{-1/4}
\end{align}
for $n=\Omega(1/\epsilon^2)$, where the second equality on the first line is by Eq.~\eqref{eq:gentle}, and the first inequality on the third line follows from Eq.~\eqref{eq:Chernoff1}.

Since we have now established that there is some $i$ such that $1-p_i \leq e^{-1/4}$, so there must also be some minimum $t \in [m]$ such that $(1-p_1)...(1-p_t)\leq e^{-1/4}$. If $t=1$, then $r_1=1-p_1\leq e^{-1/4}\leq 4/5$, that is $p_1\geq 1/5$. 

Otherwise, $t>1$. Since $t$ is minimal, $(1-p_1)...(1-p_{t-1})> e^{-1/4}$.
Taking logs on both sides, this implies that $\sum_{c=1}^{t-1}p_c\leq -\log \prod_{c=1}^{t-1} (1-p_c)\leq 1/4$ (where we have used the inequality $x< -\log(1-x)$).
If $t>1$, by Eq.~(\ref{eq:damagelemma}), Theorem~\ref{bernoulli_new} and the fact that $d_{\text{tr}}(\rho,\sigma)\leq d_{\text{Bures}}(\rho,\sigma)$ (by the standard Fuchs-Van de Graaf inequality~\cite{fuchs1995mathematical}), we have
\begin{align}
&|(1-p_1)...(1-p_t)-q_t|\\
&\leq  2 \sum_{c=1}^{t-1}
      d_{\text{tr}}(\varrho_{|\sqrt{1-B_c}}, \varrho)
     \\ 
     &\lesssim 2\sum_{c=1}^{t-1} \mathbb{E}_{\varrho}[B_c] \cdot \frac{\operatorname{stddev}[\T^{(c)}]}{\Exp[\mathsf{X}]}\nonumber\\
      &\leq \left(\sum_{c=1}^{t-1} p_c\right) \cdot \frac{\sqrt{n}}{\Exp[\mathsf{X}]}\leq
      \frac{1}{4} \cdot \frac{\sqrt{n}}{\Exp[\mathsf{X}]},
\end{align}
and also
\begin{align}
|(1-p_1)...(1-p_{t-1})-q_{t-1}|\lesssim
      \frac{1}{4} \cdot \frac{\sqrt{n}}{\Exp[\mathsf{X}]}.
\end{align}

Recall that $\mathsf{X}$ is an exponential random variable with $\mathbb{E}[\mathsf{X}] = D\sqrt{n}$, meaning that choosing $D$ large enough we have $q_t\leq 4/5$ and $q_{t-1}\geq 3/4 $. The former means that the algorithm will output some $c\leq t$ with probability larger than $1/5$. \newline \newline

\textbf{Controlling the false positive probability} --- Now we restrict our attention to the first $t$ projectors presented to the algorithm, since we have established in the previous section that the algorithm terminates with more than some constant probability after at most $t$ iterations of the For loop. We now need to show that the algorithm, with sufficiently high probability, does not output a $c$ with $\mu_c= \frac{1}{n}\sum_{i=1}^n\Tr[\Pi_{i}^{(c)}\rho_i]  \leq \theta_c-\epsilon$. Let us denote the set of such $c$ as $\mathcal{B}$ (for Bad):
$$\mathcal B=\left\{c\leq t \Bigg| \mu_c \leq \theta_c-\epsilon \right\}.$$ 

By Eq.~(\ref{eq:boundexp}) in Theorem \ref{thm:chi-stable-final}, whenever $\mu_c  \leq \theta_c-\epsilon$ (i.e. for all $c \in \mathcal{B}$), we have that
\begin{align}\label{eq:boundwrong}
p_c  &\leq \exp \left(-n \lambda\left(\theta_c-\mu_c-e\lambda/2\right)\right)\nonumber\\
&\leq {\color{black}\exp \left(-D^{-1} \sqrt{n}\epsilon+eD^{-2}/2\right)} 
\end{align}
where we used that $\mathbb{E}[\mathsf{X}]=1/\lambda=D\sqrt{n}$ for some $D>0$. Therefore $p_c\leq (\eta/m)^2 < 1/5$ if $(2\log(m/\eta)+e D^{-2}/2)^2\leq D^{-1} n \epsilon^2$, and $\eta\leq0.01$. 

The argument is now identical to~\cite{BO21}, proof of Lemma 4.2. There, using Eq.~\eqref{eq:unionapp}, $F(\varrho,\varrho_t)<1$, and $\sum_{c\in[t]/\mathcal{B}}p_c<1/4$, $q_t<4/5$, one obtains 
  \begin{align}
    &\frac{1}{2} \sqrt{\sum_{c \in [t]/\mathcal{B}} s_c}
    \ge 0.99 - \sqrt{4/5}\\
    &\rightarrow \sum_{c \in [t]/\mathcal{B}} s_c
      \ge 4 \cdot (0.99 - \sqrt{4/5})^2
      \geq 0.03.
\end{align}
  Since $\sum_{c \in [t]/\mathcal{B}} s_c$ is the probability that the algorithm
  returns an index $c \in [t]$ with $\mu_c\geq \theta_c-\epsilon$, it follows
  that the algorithm is correct with probability at least $0.03$.

\end{proof}

\section{Quantum empirical risk minimization}\label{sec:Empirical risk minimization}
In this section we present our main algorithm to perform ERM for both projector- and state-valued classes of quantum processes.
\subsection{Empirical risk minimization for projector-valued functions}
We first show how to estimate the empirical risk of a fixed number of projector-valued functions on a given product state. The main theorem we prove is the following.

\begin{theorem}[{\sf{Quantum empirical risk minimization - Theorem~\ref{th:theoere0}, refined}}]\label{th:theoere}
Given access to a product state
\begin{equation}
\varrho=\rho_{1}\otimes...\otimes \rho_{n}
\end{equation} 
and a collection of lists of projectors $\{\Pi^{(c)}_{1},...,\Pi^{(c)}_{n}\}_{c=1,...,m}$, with 
\begin{equation}
\mu_c=\frac{1}{n}\sum_{i=1}^{n}\Tr[\rho_{i}\Pi^{(c)}_{i}],
\end{equation}
(i.e. $\mu_c$ is $1-R_{\varrho}(c)$, where the second term is the empirical risk of concept $c$ on the entire product state) there is an algorithm which outputs $c^*$ together with an estimate $\hat{\mu}_{c^*}$ of $\mu_{c^*}$ such that
\begin{equation}
\Pr(|\hat{\mu}_{c^*} -\max_{c\in [m]}\mu_c|\geq \epsilon\cup |\hat{\mu}_{c^*} -\mu_{c^*}|\geq \epsilon)\leq \delta,
\end{equation}
if $n$ is large enough. In fact we can take 
\begin{align}
    n&=\frac{1}{\epsilon^2}\log\frac{1}{\delta}\log^2\frac{1}{\epsilon}\nonumber\\
    &\cdot O\left(\max\left(\log\left(\frac{m}{\delta}\log\frac{1}{\delta}\log^2\frac{1}{\epsilon}\right),(\log m+C_1)^2\right)\right)
\end{align}
for some $C_1>0$.
\end{theorem}

Let us now convey the intuition behind the claimed algorithm (which is Algorithm \ref{algo:proj2}). Observe that the task of finding the concept that attains the maximum overlap, $\max_{c\in [m]}\mu_c$, can be accomplished by binary search over the interval $[0,1]$ -- namely, we search for the largest threshold such that there is at least one risk above it, out of the $m$ possible ones. The remainder of this section gives the details of how to do so. The key idea is to divide the given product state $\varrho$ into \textit{blocks} of $l=n/2Tk$ states,
that is
\beq\label{eq:blocks}
\varrho = \varrho_1 \otimes \varrho_2 \cdots \otimes \varrho_{2Tk} \quad\text{where}\quad \varrho_s = \rho_{s,1}\otimes \cdots \otimes\rho_{s, l}
\eeqc
and take the block size $l$ large enough so that the average risk on each block concentrates towards $\mu_c$ for each $c$. This gives us sufficient confidence to apply our Algorithm \ref{alg:threshold} ($\mathsf{ThresholdSearch}$) on each block to check if there exists some concept that exceeds the current candidate threshold. Depending on the results of this check, we adjust the candidate threshold accordingly. The guarantees of this Algorithm (Algorithm \ref{algo:proj2}) are given in Lemma \ref{lem:simpleere}. Subsequently, in Lemma \ref{lem:multisamplerep}, we derive the block size $l$ we need to ensure the desired concentration. 

\begin{lemma}[Quantum empirical risk minimization given large product states\label{lem:simpleere}]
Given access to $2Tk$ blocks of states as in Eq.~\eqref{eq:blocks} and a collection of lists of projectors $\{\Pi^{(c)}_{s,j}\}_{c=1,...,m,s=1,...,2Tk, j=1,...,l}$, 
suppose that the following conditions hold:
\begin{enumerate}
    \item For appropriate constants $C_1,C_2>0$
\begin{equation}\label{cond:conv2}
  (\log m+C_2)^2<C_1 l \epsilon^2;
\end{equation}
\item  At the same time
\begin{equation} 
  l>\frac{\log(Tk/\delta)}{\epsilon^2},
  \end{equation}
for large enough $T=O(\log\frac{1}{\epsilon})$ and $k=O\left(\log\frac{1}{\delta} \log\frac{1}{\epsilon}\right)$.
  \item the numbers $0\leq\mu_c\leq 1$ are approximations of the expected value of $\{\Pi^{(c)}_{s,j}\}_{s,j}$ with respect to the blocks of states for all $c\in[m]$, i.e., they satisfy 
\begin{equation}\label{cond:conv1}
\left|\frac{1}{l}\sum_{j=1}^{l}\mathbb{E}_{\rho_{s,j}}[\Pi^{(c)}_{s,j}]-\mu_{c}\right|\leq \epsilon/4, \, \forall\,c\in [m], s\in[2Tk].
\end{equation}

\end{enumerate}

Then there is an algorithm that outputs $c^*$ and $\hat\mu_{c^*}$ such that 
\begin{equation}
\Pr(|\hat\mu_{c^*} -\max_{c\in [m]}\mu_c|\geq 6\epsilon \cup |\hat{\mu}_{c^*} -\mu_{c^*}|\geq 6\epsilon)\leq\delta.
\end{equation}
 
\end{lemma}

Here we note that `given large product states' in the title of the Lemma refers to the requirement that the $n=2Tkl$ is large enough to guarantee that Eq.~(\ref{cond:conv2}) holds.

\begin{proof}
We use Algorithm \ref{algo:proj2}, and the following notation:
\begin{itemize}
\item $h_{s}^{(c)}=\{\Pi_{s,1}^{(c)},...,\Pi_{s,n}^{(c)}\}$
\item $\overline{h_{s}^{(c)}}=\Pi_{s,1}^{(c)}\otimes \mathbf{1}\otimes ....\otimes  \mathbf{1}+\mathbf{1}\otimes \Pi_{s,2}^{(c)}\otimes ....\otimes  \mathbf{1}+\mathbf{1}\otimes \mathbf{1}\otimes ....\otimes \Pi_{s,n}^{(c)}$
\item $R_{\varrho_{s}}(c)=1-\frac{1}{n}\mathbb{E}_{\varrho_{s}}[\overline{h_{c,s}}]$ 
\item $\mathsf{X}_{c,s}$ is the random variable obtained by measuring $\overline{h_{c,s}}$ on $\varrho_{s}$.
\end{itemize}

\begin{algorithm}[H]
\textbf{Input:} $2T{}k$ product states $\{\varrho_{s}\}_{s=1,\ldots 2Tk}$ and $2Tkm$ sets of projectors $\{\{h_{s}^{(c)}\}_{c=1,...,m}\}_{s=1,...,2Tk}$, with $h_{s}^{(c)}=\{\Pi_{s,1}^{(c)},...,\Pi_{s,n}^{(c)}\}$.\\

\textbf{Parameters}: $\eps,\delta,k > 0$.      \\
\textbf{Initialize}:  $\theta=1/2, low=0,high=1, failures=0, s = 0$.
\vspace{3pt}\\
\begin{algorithmic}[1]
\While{$high - low \geq 6\eps$}
        \If{$failures < k$}
            \State $s += 1$
            \State \textbf{Threshold search}: Apply $\mathsf{ThresholdSearch}$ (Algorithm \ref{alg:threshold}) on the set of projectors-threshold pairs $\{(h_{2s}^{(c)},\theta-\eps)\}_{c=1}^m$
            with the parameter $\epsilon/4$, to the product state $\varrho_{2s-1}$.
            \If{\textbf{Threshold search} doesn't output a concept}:
                \State $failures += 1$
            \Else{ \textbf{Threshold search} outputs concept $c$}
                \State \textbf{Check}: measure $\overline{h_{2s}^{(c)}}$ on $\varrho_{2s}$ and check if $\mathsf{X}_{c,2s}\geq n(\theta-7/4\epsilon).$
                \If{\textbf{Check} outputs `yes'}:
                    \State $low \gets \theta-2\eps, \, high \gets high, \, \theta \gets \frac{1}{2}(high+low)$ \Comment{Update interval to upper half}  
                    \State $failures \gets 0.$ 
                \Else{ \textbf{Check} outputs `no' and}
                    \State $failures += 1$
                \EndIf
            \EndIf
        \Else{ there have been $k$ consecutive failures, so}
            \State $low \gets low, \,high \gets \theta,\, \theta \gets  \frac{1}{2}(high+low)$ \Comment{Update interval to lower half} 
            \State $failures \gets 0$ 
        \EndIf
\EndWhile
\State Output $\theta$ and the last selected concept, if there is one, otherwise pick it randomly.
\end{algorithmic}
\caption{\label{algo:proj2}Learning projector-valued functions}
\end{algorithm}
\textbf{Algorithm in words:} Algorithm \ref{algo:proj2} runs binary search to find an interval containing $\max_c \mu_c$, starting with the candidate interval $[0,1]$, determining whether the desired value lies in the upper or lower half, and then updating the candidate interval to the relevant half and recursing. To determine in which half of the candidate interval the desired value lies, the algorithm uses up two blocks of samples, $\varrho_{2s-1}$ and $\varrho_{2s}$: on the first block, we run $\mathsf{ThresholdSearch}$, which also outputs a concept that exceeds the current candidate value $\theta_c$. On the second block, we run a check, that confirms that this concept indeed exceeds the threshold $\theta_c$ (as $\mathsf{ThresholdSearch}$ only succeeds with probability $0.03$). We declare a $failure$ if by the end of this process, we have not received a `yes' from both algorithms; after $k$ consecutive failures, we conclude that there was actually no concept exceeding the candidate threshold, and move on by decreasing the candidate threshold (Line 17). Conversely, if at any point we receive a `yes' from both algorithms, we conclude that the candidate threshold was too low and increase it (Line 10).

\textbf{Error analysis:} Let us now analyze the error in this algorithm. The algorithm will err if either Line 17 or Line 10 updates the candidate threshold wrongly. We are interested in the probability of either of these two events happening for a fixed candidate threshold; by a union bound, we will then multiply this probability by the number of candidate thresholds that are examined, which is $O(\log(1/\epsilon))$ -- as the interval does not get updated once it becomes smaller than $6\eps$, and every update approximately halves the size of the interval (see Lines 10 and 17). 

We thus define the following four error probabilities and their associated events:
\begin{itemize}
\item $p_{FPTS}$: $\mathsf{ThresholdSearch}$ outputs a false positive, i.e. it outputs that there is a concept above threshold when there is none. 
\item $p_{FNTS}$: $\mathsf{ThresholdSearch}$ outputs a false negative, i.e. it outputs `no concept above threshold' when, in fact, there was one.
\item $p_{FPC}$: Check outputs a false positive, i.e. $\mathsf{X}_{c,2s}\geq l(\theta-7/4\epsilon)$ when in fact $\mu_c < \theta-2\eps.$
\item $p_{FNC}$: Check outputs a false negative, i.e. $\mathsf{X}_{c,2s}< l(\theta-7/4\epsilon)$ when in fact $\mu_c \geq \theta.$
\end{itemize}
Note that, by Lemma~\ref{lem:threshold}, it holds
$p_{FNTS} <0.97$, whereas $p_{FNC}$ and $p_{FPC}$ can be made exponentially small in the block size $l$ by the multiplicative Chernoff bound (Proposition~\ref{estave}).

If Line 10 updates wrongly, it can only be because $\mathsf{ThresholdSearch}$ outputs a concept even though there is no concept that exceeds the threshold \emph{and} Check outputs `yes' on that wrong concept. The probability of this happening for a given interval is at most $k\, p_{FPTS}\, p_{FPC}$, as there are $k$ rounds where this could potentially happen and any such event triggers an update of the interval. 

If Line 17 updates wrongly, there must have been $k$ consecutive failures when there actually was a concept above threshold. Each such failure is caused by one of the following events: 
\begin{itemize}
\item Threshold search executes correctly at some point but Check falsely outputs no. The probability that this happens in $k$ rounds is at most $p_{FNC}$. To see this, call $\tau_i$ the event that the first false negative check occurs at time $i$. $\{\tau_i\}$ are mutually exclusive events, and $\Pr(\tau_i)\leq (p_{FNTS})^{i-1}(1-p_{FNTS})p_{FNC}$. The probability that there is a false negative at some time $i$ is $\sum_{i=1}^{k}\Pr(\tau_i)\leq  p_{FNC}$; 
\item Threshold search wrongly outputs that no concept was above threshold (probability $p_{FNTS}$) for $k$ consecutive times. The probability of these events is bounded by $0.97^k$, as explained above.
\end{itemize}

Summing over the $T$ rounds, the probability of error is then upper bounded as
\begin{equation}
p_{err}\leq T{}(0.97^k+k p_{\mathrm{FPC}}(l)+p_{\mathrm{FNC}}(l)).
\end{equation}

With the choices made for the parameters of the algorithm, we obtain the guarantee given in the Lemma statement.

\end{proof}

Next, we control the probability that the (possibly non-identical) product states give an accurate approximation to the expectation values of the projectors associated with the possible concepts, i.e. the probability that Eq.~(\ref{cond:conv1}) is satisfied. The main technical ingredient is the following Lemma which shows, essentially, that H\"{o}ffding's inequality is effective even for proving the concentration of the empirical mean of samples without replacement from a finite population.

\begin{lemma}\label{lem:multisamplerep}
Let $\mathcal{Y}=\{1,...,n\}$, $n\geq 3Kl$, and let the $K$ sets of indices $\{\{X_{lk+i}\}_{i=1}^l\}_{k=0}^{K-1}$ be random samples drawn without replacement from $\mathcal{Y}$. Furthermore, consider $m$ different finite populations of $n$ numbers $\{\{x_{c,j}\}_{j=1}^{n}\}_{c=1}^m$, with $0\leq x_{c,j}\leq 1$. From each population $\{x_{c,j}\}_{j=1}^{n}$ obtain $K$ subsets of size $l$, as $\{\{x_{c,X_{lk+j}}\}_{j=1}^l\}_{k=0}^{K-1}$. 
Then it holds
\begin{align}
&\mathrm{Pr}\left(\underset{\substack{ c=1,...,m\\ k=0,...,K-1}}{\max}\Big|\frac{1}{l}\sum_{j=1}^{l}x_{c,X_{lk+j}}-\mu_{c}\Big|\geq \epsilon\right)\\
&\leq 2Kme^{-2l \epsilon^2/4},
\end{align}
where $\mu_c = \frac{1}{n}\sum_{j=1}^{n} x_{c,j}$.
\end{lemma}

The proof is given in the appendix~\ref{app:worep} (see the equivalent Theorem~\ref{thmappsampl}). With these two Lemmas in hand, we are now able to prove that Algorithm \ref{algo:proj2} works for finding the empirical risk minimizer amongst projector-valued functions.

\begin{proof}[Proof of Theorem~\ref{th:theoere}]
Let $n\geq 6Tkl$ be the number of observations of the output of the unknown process (where we remind readers that each observation corresponds to a quantum state). That is, we set $K = 2Tk$ in Lemma \ref{lem:multisamplerep}. Sampling without replacement $2T{}k$ lists of length $l$ from $[n]$, let the $s$-th list define the product state
\begin{equation}
\varrho_{s}=\rho_{s,1}\otimes...\otimes \rho_{s,l}, \qquad s=1,..., 2T{}k
\end{equation}
and for each concept $c \in [m]$, the set of projectors
\begin{equation}
\{\Pi^{(c)}_{s,1},... \Pi^{(c)}_{s,l}\}, \qquad s=1,..., 2T{}k.
\end{equation}
Now by identifying $x_{c,j}=\Tr[\Pi^{(c)}_{j}\rho_{j}]$ in Lemma~\ref{lem:multisamplerep}, we obtain the desired concentration
\begin{equation}\label{cond:conv22}
\left|\frac{1}{l}\sum_{j=1}^{l}\mathbb{E}_{\rho_{s,j}}[\Pi^{(c)}_{s,j}]-\mu_{c}\right|\leq \epsilon/4, \, \forall\,c\in [m], s\in[2T{}k],
\end{equation}
{ with probability $p_{err}^{(1)}\leq 4T{}kme^{-2l \epsilon^2/64}$,}

If Eq.~(\ref{cond:conv22}) is true, then condition Eq.~(\ref{cond:conv1}) is satisfied and we can apply Algorithm~(\ref{algo:proj2}) to obtain a good estimate with probability of error $p_{err}^{(2)}\leq T{}(0.97^k+(k+1)e^{-l \epsilon^2/72})$.

We can make $p_{err}^{(1)}+p_{err}^{(2)}\leq \delta$ with the choice of $T{}, k$ as in Lemma~\ref{lem:simpleere} and $l=O(\max(\log (T{}km/\delta)/\epsilon^2), (\log m+C_2)^2/\epsilon^2))$. 
\end{proof}

\subsection{Empirical risk estimation for projector-valued functions}
In this section we will show that if the size of a product state is sufficiently large with respect to the logarithm of the local dimension, we can not only identify the concept yielding the highest minimum risk and estimate that risk, but in fact estimate the risks of \textit{all} of the concepts at the same time. The key idea is built on the philosophy of shadow tomography: we find a function that predicts the expectation values \begin{equation}
\mu_c=\frac{1}{n}\sum_{i=1}^{n}\Tr[\rho_{i}\Pi^{(c)}_{i}],
\end{equation}
where $c$ runs over the possible different concepts, given access to the product state
\begin{equation}
\varrho=\rho_{1}\otimes...\otimes \rho_{n}.
\end{equation} 

In the case of lists of identical projectors ($\Pi_i^{(c)}$ independent of $i$), the task went under the name of \textit{diverse-state setting} for shadow tomography in~\cite{aaronson2019gentle}, where it was noted that the algorithm proposed there also works in this setting. Here we show that the improvements given by~\cite{BO21} carry through in this setting with the appropriate generalizations. We also make explicit an appropriate modification of the procedure of~\cite{Aaronsonshadow20} to \textit{update} the state guess, where a reference quantum state (possibly stored on a classical computer as a matrix) is used to compute the expectation values of the projectors, and updated to agree with the data from the experiment. 
Here we start with $m$ lists of projectors
\begin{equation}
h_{c}=\{\Pi^{(c)}_{1},...,\Pi^{(c)}_{n}\}, \qquad c\in[m]
\end{equation}
and an unknown product state
\begin{equation}
\varrho=\rho_{1}\otimes...\otimes \rho_{n}.
\end{equation}
As in the previous section we use the notation
\begin{itemize}
\item $h_{c,s}=\{\Pi^{(c)}_{s,1},...,\Pi^{(c)}_{s,l}\}$, $\overline{h_{c,s}}=\Pi^{(c)}_{s,1}\otimes \mathbf{1}\otimes ....\otimes  \mathbf{1}+\mathbf{1}\otimes \Pi^{(c)}_{s,2}\otimes ....\otimes  \mathbf{1}+\mathbf{1}\otimes \mathbf{1}\otimes ....\otimes \Pi^{(c)}_{s,l}$
\item $1-R_{\varrho_{s}}(c)=\frac{1}{l}\mathbb{E}_{\varrho_{s}}[\overline{h_{c,s}}]$ 
\item $\mathsf{X}_{h_{c,s}}$ is the random variable obtained by measuring $\overline{h_{c,s}}$ on $\varrho_{s}$.
\end{itemize}
Note that the quantities $1-R_{\varrho}(c)$ can be also expressed as expectation values of projectors \begin{equation}
\Pi^{(c)}:=\sum_{i=1}^n\ketbra{i}{i}\otimes \Pi^{(c)}_{i}
\end{equation}
on the classical quantum state 
\begin{equation}\sigma=\frac{1}{n}\sum_{i=1}^n \ketbra{i}{i}\otimes \rho_i.
\end{equation}
 
The algorithm works by keeping track of a classical estimate of $\sigma$ that is updated sequentially. The estimate is initialized at time $t=0$ as
\beq\label{eq:rho0}
\rho_0^{\ast} := \left(\frac{1}{n}\sum_{i=1}^n \ketbra{i}{i}\otimes \frac{\mathbf{1}}{d}\right)^{\otimes q},
\eeq
and at each time $t$, $\rho_t$ is obtained by picking one system of the $q$ at random and computing the marginal of $\rho_t^\ast$ in that system. The form of the state to update is the main difference between what we propose and the strategy in \cite{Aaronsonshadow20}. We give the proof in its entirety for convenience.

At each time step $t$, the algorithm is also provided with some $c\in [m]$ for which the current best estimate $\rho_t$ makes poor predictions relative to $\sigma$, 
that is, 
\beq
\mu_{c,t}=\mathbb E_{\rho_t}[\Pi^{(c)}] = \Tr[\rho_t \sum_i \ketbra{i}{i}\otimes \Pi^{(c)}_{i}]
\eeq
is $\eps$-far from 
\beq
\mu_c = \mathbb E_{\sigma}[\Pi^{(c)}] = \sum_{i}\frac{1}{n}\Tr[\Pi^{(c)}_{i}\rho_i]
\eeqp
(Later, in Lemma \ref{lem:averageshadow}, we'll explain how to find such a $c$.) The updating procedure then makes use of $\Pi^{(c)}$
to update $\rho^\ast_t$ to $\rho^\ast_{t+1}$. In doing so we can guarantee that at most a certain number $T$ of updates will be required, in the following sense:

\begin{lemma}[$\mathsf{Update}$]\label{le:updatecq}
 There exists a sequential procedure to update a classical estimate for some unknown quantum state $\sigma$, that initializes the estimate at time $t=0$ as $\rho_0^{\ast}$ given in Eq.~\eqref{eq:rho0}, and at time $t$, takes as input the current estimate $\rho^*_t$ and some $c\in\mathcal C$ such that either
\begin{enumerate}
    \item $\mu_c - \mathbb E_{\rho_t}[\Pi^{(c)}] \geq \eps$, or
    \item  $\mu_c - \mathbb E_{\rho_t}[\Pi^{(c)}] \leq -\eps$.
\end{enumerate}
and outputs an updated estimate, $\rho^\ast_{t+1}$, such that after 
$$T=O\left(\frac{\log{d}}{\epsilon^3}\left(\log \log d+\log \frac{1}{\epsilon}\right)\right)$$ updates, the estimate $\rho_T$ fulfils
\begin{equation}
|\mu_c-\mathbb{E}_{\rho_T}[\Pi^{(c)}]|\leq \epsilon, \qquad \forall c\in [m].
\end{equation}

\end{lemma}
\begin{proof}

First of all, we construct the following projectors
\begin{equation}
\Pi^{(c)}(l)=\sum_{\vec v\in \{0,1\}^q: ||\vec v||_1=l}\bigotimes_{i=1}^{n}(\Pi^{(c)}+v_i (\mathbf{1}-2\Pi^{(c)})),   
\end{equation}
In other words, $\Pi^{(c)}(l)$ is the sum of all events that accept the projector $\Pi^{(c)}$ at \textit{exactly} $l$ points in $[q]$ and reject $\Pi^{(c)}$ at the remaining $q-l$ points in $[q]$. Finally, we define
\begin{equation}
\Pi^{(c),-}(r)=\sum_{l\leq r}\Pi^{(c)}(l). 
\end{equation}

\begin{equation}
\Pi^{(c),+}(r)=\sum_{l\geq r}\Pi^{(c)}(l). 
\end{equation}
In words, $\Pi^{(c),+}(r)$ is the sum of all events that reject $\Pi^{(c)}$ in \textit{at least} $r$ points in $[q]$ and reject $\Pi^{(c)}$ in all remaining points in $[q]$; $\Pi^{(c),-}(r)$ is the sum of all events that reject $\Pi^{(c)}$ in \textit{at most} $r$ points in $[q]$ and accept $\Pi^{(c)}$ at all remaining points in $[q]$.
We now describe the update procedure. In Case 1, we update $\rho^{\ast}_t$ to $\rho^*_{t+1}$ by post-selecting on the event {$F_{t}^+(c)= \Pi^{(c),+}((\hat{\mu}_{c,t}+\epsilon/2)q)$. In Case 2, we then update $\rho_t$ to $\rho_{t+1}$ post-selecting on the event $F_{t}^-(c)=\Pi^{(c),-}((\hat{\mu}_{c,t}-\epsilon/2)q)$}. Define $F_t$ as the appropriate accepting event at time $t$, i.e. $F_t := F_t^+$ in Case 1 and $F_t := F_t^-$ in Case 2.

We will now upper bound the probability that the first $t$ post-selection steps all succeed, following an argument of~\cite{Aaronsonshadow20}. This is 
\beq
p_t = \Tr(F_0\rho_0^{\ast})\ldots \Tr(F_{t-1}\rho_{t-1}^{\ast})
\eeqc
and we will be able to upper-bound it by $p_t\leq (1-\epsilon)^{\Omega(t)}$. Indeed, we can show this by considering the random variable 
\begin{align}
\mathsf{X} &= \text{Number of acceptances resulting from measuring }\nonumber \\
&\Big\{ \Pi^{(c)}, \mathbf{1}-\Pi^{(c)} \Big\} \text{ on } \rho_t^{\ast}.
\end{align}
Applying Markov's inequality, we can then see that, by writing $\overline{\Pi^{(c)}}:=\Pi^{(c)}\otimes \mathbf{1}\otimes ....\otimes  \mathbf{1}+\mathbf{1}\otimes \Pi^{(c)}\otimes ....\otimes  \mathbf{1}+\mathbf{1}\otimes \mathbf{1}\otimes ....\otimes \Pi^{(c)})$,
\beq
\mathbb{E}[\mathsf{X}] = \Tr[\rho_t^{\ast} \overline{\Pi^{(c)}}] = q\mathbb{E}_{\rho_t}[\Pi^{(c)}]
\eeqc
and, in Case 1,
\begin{align}
&\Pr[\mathsf{X} \geq (\hat{\mu}_{c,t} + \eps/2)q] = \Tr[F_t^+(c) \rho_t^{\ast}] \\
&\leq \frac{\hat{\mu}_{c,t}q}{(\hat{\mu}_{c,t}+\epsilon/2)q}\leq 1-\Omega(\epsilon)
\end{align}
while an analogous calculation yields that in Case 2, 
\beq
\Tr[F_t^-(c) \rho_t^{\ast}] \leq 1-\Omega(\epsilon)
\eeqp
This implies $p_t\leq (1-\epsilon)^{\Omega(t)}$.

Now, we lower-bound $p_t$. Hypothetically, suppose that at time $t$ we were to apply the measurement $F_{t}$ (which depends on which case we are in) to ${\sigma}^{\otimes q} =\left(\sum_{i=1}^{n}\frac{1}{n}\ketbra{i}{i}\otimes \rho_i\right)^{\otimes q}$. By the promise that we are either in case 1) or 2), and by Chernoff's bound, at each step the measurements reject with probability $1-\Tr[F_t{\sigma}^{\otimes q}]\leq e^{-\Omega(q\epsilon^2)}$. Applying the measurements in sequence, $F_t\ldots F_0$ also always accepts with high probability by the quantum union bound (e.g.~\cite{wilde2013sequential}). In particular the probability of accepting on every step, until step $t$, is 
\beq
q_t\geq 1-O\left(\sqrt{t}e^{-\Omega(q\epsilon^2)}\right)
\eeqp
But then, for every $\rho_i$, $i\in [n]$, it is possible to write $\frac{\mathbf{1}}{d}=\frac{1}{d}\rho_i+\left(1-\frac{1}{d}\right)\omega_i$ for some state $\omega_i$. Then, recalling that  $\rho_0^{\ast} = (\sum_{i=1}^{n}\frac{1}{n}\ketbra{i}{i}\otimes \frac{\mathbf{1}}{d})^{\otimes q}$
\begin{equation}
\rho_0^{\ast} = \rho_0^{\otimes q}=\frac{1}{d^q}{\sigma}^{\otimes q}+\left(1-\frac{1}{d^q}\right)\omega,
\end{equation}
for some positive semi-definite and trace-1 $\omega$.

This decomposition (specifically the fact that $\omega$ is PSD) elucidates the following lower bound:
\begin{align}
p_t &= \Tr(\rho_0^{\ast} F_0) \ldots \Tr(\rho_{t-1}^{\ast} F_{t-1}) \\
&\geq \frac{1}{d^q}\left(1-
O\left(\sqrt{t}e^{-\Omega(q\epsilon^2)}\right)\right). 
\end{align}

By taking $q=\frac{C}{\epsilon^2}\left(\log \log d+\log \frac{1}{\epsilon}\right)$ we have $p_t\geq \frac{1}{d^q}\left(1-\frac{\sqrt{t}\epsilon^2}{\log d}\right)$. 

Putting this together with the upper-bound $p_t\leq(1-\epsilon)^{\Omega(t)}$, we need $T=O\left(\frac{\log{d}}{\epsilon^3}\left(\log \log d+\log \frac{1}{\epsilon}\right)\right)$, where $T$ is the number of updates after which we can estimate all the expectation values at the desired precision.
\end{proof}

To obtain $c$ satisfying the promise of Lemma~\ref{le:updatecq}, we can use the following Lemma, which is a simpler variant of Lemma~\ref{lem:simpleere}.

\begin{lemma}\label{lem:averageshadow}
Given access to $2k$ product states 
\begin{equation}
\varrho_{s}=\rho_{s,1}\otimes...\otimes \rho_{s,l}, \qquad s=1,..., 2k,
\end{equation} and a collection of lists of projectors $\{\Pi^{(c)}_{s,j}\}_{c=1,...,m,s=1,...,2k, j=1,...,l}$ and numbers $0\leq\mu_c\leq 1$, $c\in[m]$  such that

\begin{equation}\label{cond:conv3}
\left|\frac{1}{n}\sum_{j=1}^{n}\mathbb{E}_{\rho_{s,j}}[\Pi^{(c)}_{s,j}]-\mu_{c}\right|\leq \epsilon/4, \, \forall\,c\in [m], s\in[2k],
\end{equation}
and numbers $\{\lambda_{c}\}_{c=1,...,m}$. Then if we are guaranteed that
\begin{equation}
  (\log m+C_2)^2<C_1 l \epsilon^2,
\end{equation}
for appropriate constants $C_1,C_2$, and at the same time
\begin{equation}
  l>\frac{\log(k/\delta)}{\epsilon^2},
\end{equation}
for a large enough $k=O(\log\frac{1}{\delta} \log\frac{1}{\epsilon})$,
there is an algorithm that, with probability larger than $\delta$, { if there exists some $c$ such that $|\mu_c -\lambda_c|\geq 2\epsilon$, either the algorithm declares failure or it outputs a $c^*$ such that  $|\mu_{c^*} -\lambda_{c^*}|\geq \epsilon/2$.} 

\end{lemma}

\begin{proof} Use $\mathsf{ThresholdSearch}$ (Lemma~\ref{lem:threshold}) on the states $\varrho_s$, $s$ odd with 
\begin{itemize}
\item the list of projectors $\{\Pi^{(c)}_{s,1}, \Pi^{(c)}_{s,2}, ..., \Pi^{(c)}_{s,l}\}$ and threshold $\lambda_c+7/4\epsilon$, 
\item together with the list of projectors $\{\mathbf{1}-\Pi^{(c)}_{s,1}, \mathbf{1}-\Pi^{(c)}_{s,2}, ..., \mathbf{1}-\Pi^{(c)}_{s,l}\}$ and threshold $1-\lambda_c-7/4\epsilon$,
\end{itemize}
$c\in[m]$, and precision parameter $\epsilon/4$.
For each odd $s$, if $c$ is output from the search, we measure 
\begin{align}
&\Pi^{(c)}_{s+1,1}\otimes \mathbf{1}\otimes ....\otimes  \mathbf{1}+\mathbf{1}\otimes \Pi^{(c)}_{s+1,2}\otimes ....\otimes  \mathbf{1}\nonumber\\
&+\mathbf{1}\otimes \mathbf{1}\otimes ....\otimes \Pi^{(c)}_{s+1,l}
\end{align} on $\varrho_{s+1}$, with a resulting random variable $\mathsf{X}_{c,s+1}$.

If $|\frac{\mathsf{X}_{c,s+1}}{l}-\lambda_c|>\epsilon$, we declare the check passed and output $c$. If all the checks are not passed, we declare failure.

The analysis is as follows:
If there exists $c$ such that $|\mu_c -\lambda_c|\geq 2\epsilon$, then by Eq.~(\ref{cond:conv2}) we have $|\frac{1}{l}\sum_{j=1}^{l}\mathbb{E}_{\rho_{s,j}}[\Pi^{(c)}_{s,j}]-\lambda_c|>7/4\epsilon$, therefore either $\frac{1}{l}\sum_{j=1}^{l}\mathbb{E}_{\rho_{s,j}}[\Pi^{(c)}_{s,j}]\geq\lambda_c+7/4\epsilon$, or $\frac{1}{l}\sum_{j=1}^{l}\mathbb{E}_{\rho_{s,j}}[\mathbf{1}-\Pi^{(c)}_{s,j}]\leq1-\lambda_c-7/4\epsilon$. 
Since the promise of $\mathsf{ThresholdSearch}$ is fulfilled either for $\{\Pi^{(c)}_{s,j}\}_{j\in[l]}$ or $\{\mathbf{1}-\Pi^{(c)}_{s,j}\}_{j\in[l]}$, each time $\mathsf{ThresholdSearch}$ is performed, with probability larger than $0.03$ it outputs a concept $c$ such that $|\frac{1}{l}\sum_{j=1}^{l}\mathbb{E}_{\rho_{s,j}}[\Pi^{(c)}_{s,j}]-\lambda_c|>6/4\epsilon$.
 In this case, by Chernoff bound, Eq.~(\ref{eqchernoff}) we have
\begin{widetext}
\begin{align}
&\Pr(|\mathsf{X}_{c,s+1}-l\lambda_c|\leq l\epsilon)= \Pr(|\mathsf{X}_{c,s+1}-\sum_{j=1}^{l}\mathbb{E}_{\rho_{s,j}}[\Pi^{(c)}_{s,j}]+ \sum_{j=1}^{l}\mathbb{E}_{\rho_{s,j}}[\Pi^{(c)}_{s,j}]-l\lambda_c|\leq l\epsilon)\nonumber\\
&\leq \Pr\left(\Big||\mathsf{X}_{c,s+1}-\sum_{j=1}^{l}\mathbb{E}_{\rho_{s,j}}[\Pi^{(c)}_{s,j}]|-|\sum_{j=1}^{l}\mathbb{E}_{\rho_{s,j}}[\Pi^{(c)}_{s,j}]-l\lambda_c|\Big|\leq l\epsilon\right)\nonumber\\&\leq \Pr(|\mathsf{X}_{c,s+1}-\sum_{j=1}^{l}\mathbb{E}_{\rho_{s,j}}[\Pi^{(c)}_{s,j}]|\geq l\epsilon/2 )\nonumber\\
&\leq 2 e^{-l\frac{\epsilon^2}{6}}.
\end{align}

On the other hand, if a concept is selected with $|\mu_c -\lambda_c|\leq \epsilon/2$, we have

\begin{align}&\Pr(|\mathsf{X}_{c,s+1}-l\lambda_c|\geq l\epsilon)= \Pr(|\mathsf{X}_{c,s+1}-\sum_{j=1}^{l}\mathbb{E}_{\rho_{s,j}}[\Pi^{(c)}_{s,j}]+ \sum_{j=1}^{l}\mathbb{E}_{\rho_{s,j}}[\Pi^{(c)}_{s,j}]-l\lambda_c|\geq l\epsilon)\nonumber\\
&\leq \Pr(|\mathsf{X}_{c,s+1}-\sum_{j=1}^{l}\mathbb{E}_{\rho_{s,j}}[\Pi^{(c)}_{s,j}]|+|\sum_{j=1}^{l}\mathbb{E}_{\rho_{s,j}}[\Pi^{(c)}_{s,j}]-l\lambda_c|\geq l\epsilon)\nonumber\\&\leq \Pr(|\mathsf{X}_{c,s+1}-\sum_{j=1}^{l}\mathbb{E}_{\rho_{s,j}}[\Pi^{(c)}_{s,j}]|)\geq l\epsilon/2 )\nonumber\\
&\leq 2 e^{-l\frac{\epsilon^2}{6}}.
\end{align}
\end{widetext}
By an argument identical to the analysis of the error probability in Lemma~\ref{lem:simpleere}, the probability of error is bounded as $p_{err}\leq 0.97^k+2(k+1)e^{-l\frac{\epsilon^2}{6}}$, and the choice of $k$ and $l$ in the Lemma statement makes it less than $\delta$.

\end{proof}

Thus, our algorithm for ERE simply starts from an estimate dependent on the empirical distribution of measurement outcomes of the classical register, and then interleaves the $\mathsf{ThresholdSearch}$ and $\mathsf{Update}$ subroutines to progressively update this estimate. This is summed up in the following algorithm:

\begin{algorithm}[H]
\textbf{Input:} Product states
\beq
\varrho_{s}=\rho_{s,1}\otimes...\otimes \rho_{s,2lk},\qquad s\in[T]
\eeqc

\textbf{Parameters}: $T, q$. \\
\begin{algorithmic}[1]
\State Initialize on a classical computer, the classical estimate
\beq
\rho^*_0 := \left(\sum_{s=1}^{n}  \ketbra{s}{s} \otimes \frac{\mathbf{1}}{d}\right)^{\otimes q}
\eeqp
\For{$t = 1, \ldots, T$}
    \State $c \gets \mathsf{ThresholdSearch}$ on the state $\varrho_s$ with projectors and parameters as described in the proof of Lemma \ref{lem:averageshadow}.
If $\mathsf{ThresholdSearch}$ declares failure then Break.
    \State $\rho^*_{t+1} \gets \mathsf{Update}(\rho^*_t, c)$.
\EndFor\\
Output estimates $\mu_{c,T}=\mathbb E_{\rho_T}[\Pi^{(c)}]$, $\forall c \in \mathcal{C}$.
    \end{algorithmic}
\caption{\label{algo:shadow}Empirical risk estimation}
\end{algorithm}

We can now prove 

\begin{theorem}[{\sf{Quantum empirical risk estimation for projector-valued functions}} (Theorem \ref{th:theoavsh0}, refined)] \label{th:theoavsh}
Given access to a product state
\begin{equation}
\varrho=\rho_{1}\otimes...\otimes \rho_{n}
\end{equation} 
and a collection of lists of projectors $\{\Pi^{(c)}_{1},...,\Pi^{(c)}_{n}\}_{c=1,...,m}$, with 
\begin{equation}
\mu_c=\frac{1}{n}\sum_{i=1}^{n}\Tr[\rho_{i}\Pi^{(c)}_{i}],
\end{equation}
there is an algorithm which outputs estimates $\hat \mu_{c}$ such that
\begin{equation}
\Pr(|\hat{\mu}_{c} -\mu_c|\geq 2\epsilon)\leq \delta \qquad \forall c \in [m]
\end{equation}
if $n$ is large enough; in fact we can take $n=\frac{T k}{\epsilon^2} \cdot O(\max(\log( Tkm/\delta),(\log m+C_1)^2))$, 
with $T=O\left(\frac{\log{d}}{\epsilon^3}\left(\log \log d+\log \frac{1}{\epsilon}\right)\right)$ and $k=O(\log (T/\delta))$.
\end{theorem}

\begin{proof}
The algorithm for ERE runs as follows. Prepare the classical guess $\rho^*_0$ as in Lemma~\ref{le:updatecq}. Divide the product states into $T$ batches of $2k$ product states. For $t=1,...,T$, run the algorithm of Lemma~\ref{lem:averageshadow} on the corresponding batch. If a concept $c$ is selected, use it as an update for the algorithm in Lemma~\ref{le:updatecq} and continue to the next $t$, otherwise terminate and update the collection of $\mu_c$ as obtained from $\rho^*_t$.

By Lemma~\ref{lem:multisamplerep}, if $n\geq 6Tkl$ we can obtain $2Tk$ samples without replacement of length $l$ from $[n]$. By identifying $x_{c,i}=\Tr[\Pi^{(c)}_{i}\rho_{i}]$, we get $2Tk$ product states 
\begin{equation}
\varrho_{s}=\rho_{s,1}\otimes...\otimes \rho_{s,l}, \qquad s=1,..., 2T{}k,
\end{equation} and a collection of lists of projectors $\{\Pi^{(c)}_{s,j}\}_{c=1,...,m,s=1,...,2T{}k, j=1,...,l}$ such that

\begin{equation}\label{cond:conv4}
\left|\frac{1}{l}\sum_{j=1}^{l}\mathbb{E}_{\rho_{s,j}}[\Pi^{(c)}_{s,j}]-\mu_{c}\right|\leq \epsilon/4, \, \forall\,c\in [m], s\in[2T{}k].
\end{equation}
with probability $p_{err}^{(1)}\leq 2Tkme^{-2l \epsilon^2/64}$.

If Eq.~(\ref{cond:conv4}) is true, then condition Eq.~(\ref{cond:conv3}) is satisfied and we can apply the algorithm of Lemma~\ref{lem:averageshadow} to obtain, if $\max_c|\hat\mu_c-\mu_c|\geq 2\epsilon$ an estimate $\hat{\mu}_c$, $|\hat\mu_c-\mu_c|\geq \epsilon/2$ with probability of error $p_{err}^{(2)}\leq (0.97^k+2(k+1)e^{-l \epsilon^2/6})$. Otherwise, if $\max_c|\hat\mu_c-\mu_c|\leq 2\epsilon$ and we don't get any $|\hat\mu_c-\mu_c|\geq \epsilon/2$ we are satisfied. Anyway, after $T=O\left(\frac{\log{d}}{\epsilon^3}\left(\log \log d+\log \frac{1}{\epsilon}\right)\right)$ updates we also have $|\hat\mu_c-\mu_c|\leq \epsilon/2$ by Lemma~\ref{le:updatecq}. The probability of error is then less than $Tp_{err}^{(2)}$.

We can make $p_{err}^{(1)}+Tp_{err}^{(2)}\leq \delta$ with $T$ and $k$ as in the theorem statement and $$l=O(\max(\log (Tkm/\delta)/\epsilon^2), (\log m+C_2)^2/\epsilon^2))$$ and  thus $n$ as stated.
\end{proof}

\subsection{Empirical risk minimization for state-valued functions}

A key subroutine we will need to introduce is based on \textit{hypothesis selection}, which is a way of choosing a classical or quantum probability distribution that best fits some observed data. 
We'll use a generalized version of the algorithm of~\cite{BO21} for hypothesis selection (which applied to i.i.d states) to find the empirical risk minimizer, with the loss given by the trace distance, 
in a set of candidate state-valued processes. 

\begin{theorem}[{\sf{Quantum empirical risk minimization for state-valued functions}} (Theorem~\ref{theo:states_finite0}, refined)]\label{theo:states_finite}
Let ${\cal C}=\{\sigma_i:[n]\rightarrow D(\mathcal{H}^{(d)})\}_{i=1}^m$ be a class of state-valued functions and 
\begin{equation}
\varrho=\rho_1\otimes...\otimes\rho_n.
\end{equation}
There exists an algorithm which given $\varrho$ outputs $i$ such that

\begin{equation}\label{eq:ER_states}
    \frac{1}{n}\sum_{s=1}^n d_{\rm tr}(\sigma_i(s),\rho_s)\leq 3\eta + 4\epsilon,
\end{equation}
where 
\begin{equation}\label{eq:optimalQHS1}
    \eta:=\min_{i\in[m]} \frac{1}{n}\sum_{s=1}^n d_{\rm tr}(\sigma_i(s),\rho_s),
\end{equation}
with probability of error less than $\delta$ if $n$ is large enough, i.e. we can take in fact
$$n=\frac{T k}{\epsilon^2} \cdot O(\max(\log(Tkm/\delta),(\log m+C_1)^2)),$$ 
with $T=O\left(\frac{\log{d}}{\epsilon^3}\left(\log \log d+\log \frac{1}{\epsilon}\right)\right)$ and $k=O(\log (T/\delta))$.

\end{theorem}
\begin{proof}
Let us define the following classical-quantum states for $k=1,\cdots,m$:
\begin{equation}
    \sigma_k:=\frac{1}{n}\sum_{s=1}^n \ketbra{s}{s} \otimes \sigma_k(s),
\end{equation}
and
\begin{equation}
    \rho:=\frac{1}{n}\sum_{s=1}^n \ketbra{s}{s} \otimes \rho_s,
\end{equation}
Then $\eta=\min_i d_{\rm tr}(\sigma_i,\rho)$ and the algorithm has to output a hypothesis $k$ such that $d_{\rm tr}(\sigma_k,\rho)\leq 3\eta+\epsilon$. 

The algorithm works as follows: for each $s$, $k$ and $i<j$ define the Helstrom projectors
\begin{equation}\label{eq:diff}
    A_{ij}(s):=\left(\sigma_i(s)-\sigma_j(s)\right)_+, 
\end{equation}
where $(\cdot)_+$ is the projector on the positive part of the argument, and their block-sum $A_{ij}:=\sum_{s\in[n]}\ketbra{s}{s}\otimes A_{ij}(s)$. 
By construction, these projectors satisfy
\begin{equation}
    d_{\rm tr}(\sigma_i,\sigma_j)=\Tr(\sigma_i A_{ij}) - \Tr(\sigma_j A_{ij}).
\end{equation}
The algorithm then uses Algorithm \ref{algo:proj2} on $\varrho$ to perform ERE of the projector-valued functions $\{A_{ij}(s):~ i<j\}$, outputting estimates $\mu_{ij}$ such that, with probability at least $\delta$, it holds 
\begin{equation}\label{eq:STforQHS}
    |\frac{1}{n}\sum_{s=1}^{n}\Tr[\rho_s A_{ij}(s)]-\mu_{ij}|\leq 2\epsilon.
\end{equation} 
We also denote 

\begin{equation}
\nu_{kij}:=\frac{1}{n}\sum_{s=1}^n\Tr[\sigma_k(s) A_{ij}(s)].
\end{equation}
Finally, the algorithm employs a classical subroutine to minimize the quantity
\begin{equation}
    \Delta_k:=\max_{i<j}|\nu_{kij}-\mu_{ij}|,
\end{equation}
finding $k^*:={\rm argmin}_k \Delta_k$. We can show that this is a good enough candidate, i.e., $d_{\rm tr}(\sigma_{k^*},\rho)$ is sufficiently small.
Indeed, if we define the optimal hypothesis attaining Eq.~\eqref{eq:optimalQHS1} as $i^*:={\rm argmin}_i d_{\rm tr}(\sigma_i,\rho)$, by the triangle inequality it holds 
\begin{align}
    &d_{\rm tr}(\sigma_{k^*},\rho)
    \leq d_{\rm tr}(\sigma_{k^*},\sigma_{i^*}) + d_{\rm tr}(\sigma_{i^*},\rho) \\
    &= \left|\Tr(\sigma_{k^*}A_{k^*i^*}) - \Tr(\sigma_{i^*}A_{k^*i^*})\right| + \eta \\
    &\leq \left|\Tr(\sigma_{k^*}A_{k^*i^*}) - \mu_{k^*i^*}\right| \nonumber\\
    &+ \left|\mu_{k^*i^*} - \Tr(\sigma_{i^*}A_{k^*i^*})\right| + \eta. \label{eq:boundMinTrace}
\end{align}
The first term in Eq.~\eqref{eq:boundMinTrace} can be bounded in the following way:
\begin{align}
    &\left|\Tr(\sigma_{k^*}A_{k^*i^*}) - \mu_{k^*i^*}\right|
    = \left|\nu_{k^*k^*i^*} - \mu_{k^*i^*} \right| = \Delta_{k^*} \\
    &\leq \Delta_{i^*} = \max_{i<j} \left|\Tr(\sigma_{i^*}A_{ij}) - \mu_{ij}\right|\\ 
    &\leq \max_{i<j} \left|\Tr(\sigma_{i^*}A_{ij}) - \Tr(\rho A_{ij})\right| \nonumber\\
    &+ \max_{i<j} \left|\Tr(\rho A_{ij}) - \mu_{ij}\right| \\
    &\leq \eta + 2\epsilon 
\end{align}
Note that the last two inequalities above are implied by~(\ref{eq:STforQHS}) and therefore the overall result holds with probability at least $1-\delta$. Proceeding similarly for the second term in Eq.~\eqref{eq:boundMinTrace} we obtain
\begin{equation}
     \left|\Tr(\sigma_{i^*}A_{k^*i^*}) - \mu_{k^*i^*}\right|\leq \eta +2\epsilon
\end{equation}
as well.

Therefore, by taking $\epsilon' =\epsilon/3$ we can conclude that $d_{\rm tr}(\sigma_{k^*},\rho)\leq 3\eta + 4\epsilon$ with probability at least $1-\delta$.
\end{proof}

\section{Statistical learning for classical-quantum processes}\label{sec:statlearn}
The previous section presented Theorems for ERM on finite concept classes. In this section we extend our tools to learning even infinite concept classes. The main result of this section is a statistical learning theorem (Theorem \ref{theo_ERM}) for classical-quantum processes which gives conditions on their learnability in terms of the covering number of the concept class $\C$ from which they are drawn. Our proof is constructive and provides an explicit algorithm to learn the given concept class. { Note however that the algorithm relies on constructing empirical covering nets, which can be demanding; hence we expect that faster algorithms can be found to improve the performance. We also remark that the following results are a consequence of the guarantees on our ERM algorithms for both projector-valued and state-valued concept classes, proved in Section~\ref{sec:Empirical risk minimization}, together with standard classical statistical learning theory guarantees on uniform convergence of the empirical risks in terms of the growth functions $\Gamma_{1}$ (see Sections \ref{sec:statlearnprel} and \ref{sec:statlearntheo}). It is important to note that our sufficient condition for learnability is instead expressed in terms of the growth functions $\gamma_{1,q}\geq \Gamma_{1}$, which are sufficient to guarantee both ERM and uniform convergence.}  

In Subsection~\ref{uniconv} we formulate a lemma that identifies, even for continuous-valued concept classes, a finite cardinality concept class on which we apply the algorithm of the previous section. In Subsection~\ref{part1}, we prove our main Theorem \ref{theo_ERM} for concept classes that map to projector-valued functions. Through shadow tomography, this can be extended to estimating all the empirical risks, and we do this in Section \ref{parttomo}. In Subsection \ref{part2}, we prove Theorem \ref{theo_ERM} for concept classes that map to state-valued functions. In Appendix~\ref{subsec:warmup} we also show that the algorithm of~\cite{ChungLin21} to learn concept classes that output pure states, in the realizable case, works also when the growth function is slowly growing.

\subsection{Uniform convergence}\label{uniconv} 

In the rest of this section we will perform ERE and/or ERM with respect to an $\epsilon$-net of the concept class which depends on the classical data. 
The following Lemma ensures that ERE/ERM also gives a good solution for the estimation/minimization of the true risks. 

\begin{lemma}\label{th:uniconv}
Given $l_0$ copies of
$$\rho=\sum_{x\in \mathcal{X}}\mathcal{D}(x)\ketbra{x}{x}\otimes \rho(x),$$ by looking at the classical register, we can find a finite subset $\mathcal C_{\epsilon}$ of the concept class $\mathcal C$ with loss function ($L_s$ or $L_p$) such that for every $c\in\mathcal C$, there is a (known) $c^*(c)\in \mathcal C_{\epsilon}$ such that
\begin{equation}
|R(c)-R(c^*(c))|\leq\epsilon,
\end{equation}
with cardinality at most 
$\gamma_{1,q}(l_0, \epsilon / 2, \mathcal C),$
and at the same time
\begin{equation}
\forall c \in \mathcal{C}:|R(c)-\hat{R}(c)|<\epsilon/4
\end{equation}
with probability of error
\begin{equation}
p_{err}\leq  4\gamma_{1,q}(2 l_0, \epsilon / 64, \mathcal{C}) e^{-\frac{l_0 \epsilon^{2}}{512}}.
\end{equation}
\end{lemma}

\begin{proof}
This is a simple consequence of Theorem~\ref{convergence_proj}. In fact, given a sample $S$ consisting of $l_0$ samples of the random variable $x\in\mathcal X$ associated to the classical register, we have that
\begin{align}
&\Pr_{S \sim \mathcal{D}^{l_0}}[\exists c \in \mathcal{C}:|R(c)-\hat{R}(c)|\geq \epsilon/4]\\
&\leq 4 \gamma_{1,q}(2 l_0, \epsilon / 64, \mathcal{C}) e^{-\frac{l_0 \epsilon^{2}}{512}}.
\end{align}
By taking an $\epsilon/2$-net $\mathcal C_{\epsilon/2}(\vec x)$ according to the appropriate loss function $L$ on the data $\vec x$, we have that for every $c\in\mathcal C$ there is $c^*(c)\in \mathcal C_{\epsilon/2}(\vec x)$ such that, via Eqs.~\eqref{contstates},~\eqref{contproj}
\begin{equation}\label{neteps}
|\hat R(c)-\hat R(c^*(c))|\leq \epsilon/2.
\end{equation}

If both conditions are satisfied, we have

\begin{align}
&|R(c)-R(c^*(c))| \leq |R(c)-\hat{R}(c)|\nonumber\\
&+|R(c^*(c))-\hat{R}(c^*(c))|+|\hat R(c)-\hat{R}(c^*(c))|\\
&\leq |R(c)-\hat{R}(c)|+|R(c^*)-\hat{R}(c^*)| +\epsilon/2 \\
&\leq \sup_{c\in C} 2|R(c)-\hat{R}(c)|+\epsilon/2\\
&\leq\epsilon.
\end{align}

The bound on the cardinality comes from the definition of $\gamma_{1,q}$.

\end{proof}

Via this Lemma, whenever we want to solve the risk estimation for a classical quantum source, we can use part of the classical data to extract a good $\epsilon$-net on the space of \textit{true risks} (meaning exactly that for every concept $c$ there is a known concept $c^*_c$ with true risk $\epsilon$-close), and use the quantum data to obtain the best hypothesis on this finite cardinality concept class, using the algorithms shown in the previous sections. 
\textit{Remark:} As the cardinality of the $\epsilon$-net on the data can increase with the size of the sample, one could also construct the $\epsilon$-net only on a part of the classical data, such that uniform convergence is ensured at the desired level. If $\lim_{n\rightarrow\infty}\log \gamma_{1,q}(2n,\epsilon/64,\mathcal C)/n=0$, there will be a finite $l_0$ such that $p_{err}\leq\delta$. The cardinality of the  $\epsilon$-net will be then $\gamma_{1,q}(l_0,\epsilon/2,\mathcal C)$. This observation is used in Appendix~\ref{appC}.

\subsection{Learnability of quantum processes that map to projectors (Proof of Thm~\ref{theo_ERM} for loss function $L_p$)} \label{part1}

We now present one of our main technical results about learning unknown quantum processes without input control. This was Theorem \ref{theo_ERM}, which we reproduce below for convenience:
\begin{theorem}[(Theorem \ref{theo_ERM}, repeated) Learning quantum processes via ERM]\label{theo_ERM_2}
Suppose the concept class $\C$ consists of classical-quantum processes mapping to projectors or states and let $\epsilon>0$ be the accuracy parameter. { Furthermore, let $S = (x_i,\rho(x_i))_{i=1}^n$ be the training set, with $x_i \xleftarrow{\mathcal{D}} X$ and $\rho(\cdot)$ an unknown classical-quantum channel. 

Then, the appropriate ERM algorithm of Theorems~\ref{th:theoere0},~\ref{theo:states_finite0}, run on an $\epsilon$-net of the concept class $\mathcal{C}$ (according to the appropriate pseudometric determined by $x_1,...,x_n$), provide an \textit{agnostic} learning algorithm $\mathcal{A}:\X^n \times  \mathcal{L}(\mathcal{H}^{(d)})^{\otimes n} \rightarrow \C$. This algorithm} outputs an hypothesis ${\cal A}(S)$ satisfying, for some fixed $\eta,\xi\geq 1$, and $n$ large enough,
\begin{equation}\label{eq:ERM2}
\Pr_{S}[R_\rho(\mathcal{A}(S)) -\eta\inf_{h\in\mathcal C} R_\rho(h)<\xi\epsilon] > 1-\delta,
\end{equation}
if
\beq\label{eq:learnability}
\lim_{n\rightarrow \infty}\frac{\log^2 \gamma_{1,q}(n, \eps,\mathcal{C})}{n}=0,\qquad \forall\epsilon>0
\eeqp
In particular, this applies to risks defined via the loss functions $L_p$ (in this case $\eta=1$, $q=\infty$) and $L_s$ (in this case $\eta=3$, $q=1$) for projector-valued and state-valued concept classes $\mathcal{C}$ respectively.

\end{theorem}
We prove this theorem in two parts. In this section, we prove that Eq.~(\ref{eq:learnability}) gives a sufficient condition for learning concept classes of projector-valued functions (i.e. using loss function $L_p$). In Section~\ref{part2}, we will prove that the same condition suffices for learning a concept class whose functions map to mixed states (i.e. using loss function $L_s$). { Note that, in the case where the functions in the concept class map to pure states, this latter subcase would follow almost immediately from the former. }

We will actually prove a more refined, quantitative statement:

\begin{theorem}[Learning projector-valued functions via ERM]\label{ermrefined}
Suppose the concept class $\C$ consists of quantum processes mapping to projectors and let $\epsilon>0$ be the accuracy parameter, and suppose that
\begin{equation}
\lim_{n\rightarrow \infty}\frac{\log^2 \gamma_{1,\infty}(n,\epsilon,\mathcal C)}{n}=0, \qquad \forall\epsilon>0.
\end{equation}
Given as input a training set $S = (x_i,\rho(x_i))_{i=1}^n$ with $x_i \xleftarrow{\mathcal{D}} X$ and $\rho(\cdot)$ an unknown classical-quantum channel, there is an \textit{agnostic} learning algorithm $\mathcal{A}:\X^n \times \mathcal{L}(\mathcal{H}^{(d)})^{\otimes n} \rightarrow [0,1]^{\C}$ (i.e. the algorithm of Theorem~\ref{th:theoere} using an $\epsilon$-net of $\mathcal C$ as concept class), such that it outputs $c^*$ together with an estimate $\hat{\mu}_{c^*}$ of $\hat{R}_{\rho}(c^*)$ and
\begin{equation}\label{eq:ERMp}
\Pr_{S}[|\hat{\mu}_{c^*}-\inf_{c\in\mathcal C} R_\rho(c)|\geq 7\epsilon\cup |\hat{\mu}_{c^*}-\hat{R}_{\rho}(c^*)|\geq 6\epsilon] =: p_{\text{err}}.
\end{equation}

With $n=6Tkl$, for large enough $T=O(\log\frac{1}{\epsilon})$, $k=O(\log(1/\delta\log(1/\epsilon)))$, there exist constants $C_1, C_2, C_3$ such that,
as long as $l$ satisfies
\beq\label{eq:l0}
(\log\gamma_{1,\infty}(6Tkl, \eps/2,\mathcal{C})+C_2)^2\leq C_1l\epsilon^2/9
\eeqc
we have
\begin{align}
p_{\text{err}}&\leq \frac{\delta}{2}+C_3 Tk\gamma_{1,\infty}(6Tkl, \epsilon / 2, \mathcal{C}) e^{-\frac{l \epsilon^{2}}{72 }} \nonumber\\
&+4\gamma_{1,\infty}(12Tkl, \epsilon / 64, \mathcal{C}) e^{-\frac{6Tkl \epsilon^{2}}{512 }}.
\end{align}

\end{theorem}

Once this is established, Part 1 of Theorem~\ref{theo_ERM} follows by redefining $\epsilon$ and taking $l$ large enough.

We emphasize that this $\mathsf{ThresholdSearch}$ subroutine differs from the original ones \cite{Aaronsonshadow20,BO21} because it does not require identical copies of the state to output a concept above threshold. Exploiting this algorithm and the convergence of the empirical risk to the true risk as stated in Theorem~\ref{convergence_proj}, we can finally prove Theorem~\ref{theo_ERM} for the case of projectors.

\begin{proof}
We remind readers that the empirical risk for projector-valued functions can be written as

\beq
\hat R_{\rho}(c) = \frac{1}{n}\sum_{i=1}^n (1-\Tr[\rho(x_i)\Pi^{(c)}(x_i)])
\eeq

The algorithm looks at the classical register $\vec x$ and and obtains an $\epsilon$-net of the empirical risk $\mathcal C|_{\vec{x}}$, which is also a $2\epsilon$-net on the true risks, with probability of error bounded as in Lemma~\ref{th:uniconv}

\begin{equation}
p_{err, unif}(6Tkl)\leq  4\gamma_{1,\infty}(12Tkl, \epsilon / 64, \mathcal{C}) e^{-\frac{6Tkl \epsilon^{2}}{512}}.
\end{equation}

Then it runs the algorithm of Theorem~\ref{th:theoere} on the product state obtained from the full dataset $\vec{x}$, with the concepts in $\mathcal C|_{\vec{x}}$. In this way we obtain $c^*$ such that $|\hat{\mu}_{c^*}-\inf_{c\in \mathcal C|_{\vec{x}_0}}\hat{R}(c)|\leq 6\epsilon$ and $\hat{\mu}_{c^*}-\hat{R}(c^*)\leq 6\epsilon$, implying $\hat{R}(c^*)< \inf_{c\in \mathcal C|_{\vec{x}_0}}\hat{R}(c)+6\epsilon$ with probability bounded as in the proof of Theorem~\ref{th:theoere}
\begin{equation}
p_{err,ERM}\leq T(0.97^k+(k+1)e^{-l\epsilon^2/72})+2Tkme^{-2l\epsilon^2/64}.
\end{equation}

If this is true we have that, for the selected concept $c^*$

\begin{align}
R(c^*)&\leq \hat R(c^*)+\epsilon/4\\
&\leq \inf_{c\in \mathcal C|_{\vec{x}}} \hat R(c)+6\epsilon+\epsilon/4\\
&\leq \inf_{c\in \mathcal C} \hat R(c)+6\epsilon+3\epsilon/4\tag{$\epsilon/2$-net}\\&\leq  \inf_{c\in \mathcal C} R(c)+7\epsilon.
\end{align}

The probability of error is then upper bounded as
\begin{equation}
p_{err,unif}+p_{err,ERM}.
\end{equation}

With the choices made for the parameters, we obtain the thesis.

\end{proof}

\subsection{Shadow tomography for classical-quantum states}\label{parttomo}
 We can also estimate the empirical risks of all the concepts in a class, a task that is strictly related to shadow tomography. By uniform convergence, these will be close to the true risks. In fact, we can show an \textit{improved} algorithm for shadow tomography of classical-quantum states. Using this algorithm, we can not only find the minimum empirical risk in the concept class, but also simultaneously estimate empirical risk for all concepts. 

 \begin{theorem}[(Theorem~\ref{theo_shadow}, refined) Improved shadow tomography of classical-quantum states]\label{thm:minprojector}
Consider a collection of projector-valued functions $\mathcal C=\{f_{c}:x\rightarrow \Pi^{(c)}(x)\}$ with domain ${\cal X}$ and image in the projectors of a Hilbert space of dimension $d$.
Given access to $n$ copies of a classical-quantum state $\rho=\sum_{x\in{\cal X}} \mathcal{D}(x) \ketbra{x}{x}\otimes\rho(x)$, the algorithm of Theorem~\ref{th:theoavsh0} used with an $\epsilon$-net of $\mathcal C$ outputs values $\{\mu_c\}_{c\in \mathcal{C}}$ such that, except with probability $\delta$, for all $f_c\in \mathcal C$ it holds 
\begin{equation}\label{conv_shadow}
    \left|\sum_x \mathcal{D}(x)\Tr[\rho(x) \Pi^{(c)}(x)]-\mu_c\right|\leq 3\epsilon
\end{equation}
if either one of these conditions holds:
\begin{itemize}
\item $\mathcal C$ is finite. Then the minimal number of copies satisfies $n=\tilde O(\frac{\log^2{m}\log d \log(1/\delta)}{\epsilon^5})$ 
\item ${\cal C}$ is infinite but it holds \begin{equation}\label{eq:convergence}
\lim_{n\rightarrow \infty}\frac{\log^2 \gamma_{1,\infty}(n,\epsilon,\mathcal C)}{n}=0,\qquad \forall\epsilon>0
\end{equation}
\end{itemize}
and the number of copies is large enough. 

In particular, with $n=6Tkl$, for large enough $T=O\left(\frac{\log{d}}{\epsilon^3}\left(\log \log d+\log \frac{1}{\epsilon}\right)\right)$ and $k=O(\log (T/\delta))$, there exist constants $C_1, C_2, C_3$ such that, as long as $l$ satisfies
\beq\label{eq:l02}
(\log\gamma_{1,\infty}(6Tkl, \eps/2,\mathcal{C})+C_2)^2\leq C_1l\epsilon^2/9
\eeqc 
have that the probability of error is bounded as
\begin{align}
p_{err}&\leq \frac{\delta}{2}+C_3 Tk\gamma_{1,\infty}(6Tkl, \epsilon / 2, \mathcal{C}) e^{-\frac{l \epsilon^{2}}{32 }}\nonumber \\
&+4\gamma_{1,\infty}(12Tkl, \epsilon / 64, \mathcal{C}) e^{-\frac{6Tkl \epsilon^{2}}{512 }}.
\end{align}

\end{theorem}
\begin{proof}
The only difference with Theorem~\ref{ermrefined} is to replace the ERE subroutine with the shadow tomography subroutine to estimate the empirical risks. The probability of error of the ERE, $p_{err,ERE}$, becomes bounded as in the proof of Theorem~\ref{th:theoavsh0}, with $m=\gamma_{1,\infty}(6Tkl, \epsilon / 2, \mathcal{C})$:

$p_{err,ERE}\leq T((0.97)^k+2(k+1)e^{-l\epsilon^2/6})+2Tk\gamma_{1,\infty}(6Tkl, \epsilon / 2, \mathcal{C}) e^{-2l\epsilon^2/64}$.

On the other hand, the probability of error of the uniform convergence is the same as in the previous Theorem. By summing the two probabilities of error we get the bound in the thesis. Notice that with Theorem~\ref{th:theoavsh0} we get estimates of the empirical risks at precision $2\epsilon$. By Theorem~\ref{th:uniconv} this translates into a precision $3\epsilon$ on the full class.

\end{proof}

\subsection{Learnability of quantum processes that map to states (Thm~\ref{theo_ERM} with loss function $L_s$)}\label{part2}

In this section, we will prove that quantum processes that output states are also efficiently learnable (having shown the analogous statement for processes that output projectors in Section \ref{part1}). This also constitutes the second half of the statement of Theorem \ref{theo_ERM}. In this setting the empirical risk takes the form
\beq\label{eq:risk_states}
\hat{R}_{\rho}(h) = \frac{1}{n}\sum_{i=1}^n d_{\mathrm{tr}}(\sigma_h(x_i),\rho(x_i))
\eeqp

We can now prove Theorem~\ref{theo_ERM} for state-valued functions, with loss function $L_s$. In fact, we prove a refined statement:

\begin{theorem}[Refinement of Theorem~\ref{theo_ERM} for state-valued functions]\label{theoermproj}
Suppose the concept class $\C$ consists of quantum processes mapping to states and let $\epsilon>0$ be the accuracy parameter. Suppose that

\begin{equation}
\lim_{n\rightarrow \infty}\frac{\log^2 \gamma_{1,1}(n,\epsilon,\mathcal C)}{n}=0,\qquad \forall\epsilon>0.
\end{equation}
Given as input a training set $S = (x_i,\rho(x_i))_{i=1}^n$ with $x_i \xleftarrow{\mathcal{D}} X$ and an unknown $\rho(x)$, there is a learning algorithm $\mathcal{A}:\X^n \times \mathcal{L}(\mathcal{H}^{(d)})^{\otimes n} \rightarrow \C$ obtained from the algorithm of Theorem~\ref{theo:states_finite} applied to an $\epsilon$-net of $\mathcal{C}$ (according to the appropriate pesudometric defined by $x_1,...,x_n$), such that
\begin{equation}\label{eq:ERMs}
\Pr_{S}[R_\rho(\mathcal{A}(S)) -3\inf_{c\in\mathcal C} R_\rho(c)<6\epsilon] =: 1-p_{\text{err}},
\end{equation}
and setting $n=6Tkl$, for large enough $T=O\left(\frac{\log{d}}{\epsilon^3}\left(\log\log d+\log\frac{1}{\epsilon}\right)\right)$, $k=O(\log(T/\delta))$, there exist $C_1,C_2,C_3,C_4>0$ constants such that 
as long as
\beq
(\log\gamma_{1,1}(6Tkl, \eps/2,\mathcal{C})+C_2)^2\leq C_1l\epsilon^2
\eeqc
 we have 
\begin{align}
p_{\text{err}}&\leq \frac{\delta}{2}+C_3\gamma_{1,1}(12 Tkl, \eps/16,\mathcal{C})e^{-6Tkl\epsilon^2/128}\nonumber\\
&+C_4Tk(\gamma_{1,1}(6Tkl, \epsilon/ 2, \mathcal{C}))^2 e^{-\frac{l{\epsilon}^{2}}{32 }}.
\end{align}

\end{theorem}

\begin{proof}[Proof of Theorem~\ref{theoermproj}]
The proof is a straightforward generalization of the proof of Theorem~\ref{theo:states_finite} using the algorithm of Theorem~\ref{thm:minprojector}. In this case, $\mathcal C$ is possibly of infinite cardinality.
The algorithm works as follows. Following the notation of Theorem~\ref{thm:minprojector}, first we measure the classical register of the copies of $\rho$, obtaining a vector $\vec x$ and an $\epsilon$-net $\mathcal C_{\vec x}$ , with cardinality less than the covering number  
$\gamma_{1,1}(6Tkl,\epsilon/2, \mathcal C)$. 
By Theorem~\ref{th:uniconv}, this $\epsilon/2$-net gives an $\epsilon$-net for the true risks, with probability of error bounded as
\begin{equation}
p_{{err,unif}}\leq  4 \gamma_{1,1}(12 Tkl, \eps/16,\mathcal{C})e^{-6Tkl\epsilon^2/128},
\end{equation}

Therefore we have that if we can find $c^*$ such that $\hat{R}(c^*)\leq 3\inf_{c\in \mathcal C_{\vec x}} \hat{R}(c)+4\epsilon$, then
{\begin{align}
R(c^*)&\leq \hat R(c)+\epsilon/4\\
&\leq 3\inf_{c\in \mathcal C_{\vec{x}}} \hat R(c)+4\epsilon+\epsilon/4\\
&\leq 3\inf_{c\in \mathcal C} \hat R(c)+4\epsilon+7\epsilon/4\tag{$\epsilon/2$-net}\\&\leq  3\inf_{c\in \mathcal C} R(c)+6\epsilon.
\end{align}}
Therefore, we also have that $R(c^*)\leq 3\inf_{c\in \mathcal C} R(h)+6\epsilon$ with high probability.
We can then apply Theorem~\ref{theo:states_finite} to the quantum data with the concept class $\mathcal C_{\vec x}$. The probability of error of the algorithm of Theorem~\ref{theo:states_finite} is in fact the probability of error of the ERE algorithm for the projector-valued concept class constructed from $\mathcal C_{\vec x}$ as in Theorem~\ref{theo:states_finite}, therefore, we have again
\begin{align}
p_{\text{err,ERM}}&\leq T((0.97)^k+2(k+1)e^{-l\epsilon^2/6})\nonumber\\
&+2Tk\gamma^2_{1,1}(6Tkl,\epsilon/2,\mathcal{C}) e^{-2l\epsilon^2/64}.
\end{align}
Putting all together we obtain the full error bound.

\end{proof}

\section{Classes with bounded covering numbers}\label{sec:ex_covering}

In this section we present some models of physically motivated concept classes for which upper bounds on their covering numbers can be found. We find these bounds on covering numbers using continuity and bounds on covering numbers of real functions and matrices. This opens up the question of whether it is possible to define a combinatorial dimension (such as the VC dimension or the fat-shattering dimension) corresponding to the covering numbers we introduced. 

Whenever needed in the following, we employ a real-valued function class $\mathcal F$, and we assume that for this function class we have a bound on the covering number in terms of some combinatorial dimension, for example the fat-shattering dimension (see~\cite{Anthony1999} for definitions and proofs of these bounds). It holds that, for a class of functions $\mathcal F:\mathbb R \rightarrow [0,B]$ with fat-shattering dimension $D=\mathrm{fat}_{\mathcal F}(\epsilon/4)$, the covering number is bounded as
\begin{equation}
\Gamma_{p}(n,\epsilon,\mathcal F) < 2\left(\frac{4 n B^2}{\epsilon^2}\right)^{D\log_2 (4 e B n/(D \epsilon))}.
\end{equation}

Furthermore, we can have finite covering numbers for sets of $k\times k$ matrices, seen as functions with trivial input to $\mathbb C^{k}\times \mathbb C^{k}$. These techniques were used for the purpose of studying generalization bounds in~\cite{Caro2022}, and uniform convergence for learning quantum channels on multiqubit systems of finite size, in the controlled input setting, in~\cite{huang2022quantum}. The following Lemma is crucial.

\begin{lemma}[Size of an $\epsilon$-net over unitary and bounded hermitian matrices]\label{lem:vershynin}
\beq
 N_{in}\left( \epsilon, B_R(x),|| \cdot || \right)  \leq \left(1+\frac{2 R}{\epsilon}\right)^K \leq \left(\frac{3 R}{\epsilon}\right)^K
\eeq
where $B_R(x)$ is a norm ball of radius $R\geq \varepsilon$ about some point $x\in \mathbb{R}^K$. We can use this to calculate the size of $\epsilon$-nets over matrices with respect to the operator norm, by using the fact that for a set of unitaries $S_U$ in dimension $k$, $S_U \subset B_1(0)$ with $K=(2k)^2$, and for a set of hermitian matrices $S_H$  in dimension $k$ with norm bounded by $b$, which we denote by $\mathcal M_b^{(k)}$, $\mathcal M_b^{(k)} \subset B_{b}(0)$ with $K=k^2$.
\begin{align}
|N_{in}\left(\epsilon, S_U,\lVert\cdot\rVert\right)| &\leq |N_{in}\left(\epsilon/2, B_1(0),\lVert\cdot\rVert\right)|\nonumber\\ &\leq \left(\frac{6}{\epsilon}\right)^{4k^2},
\end{align}
\begin{align}
|N_{in}\left(\epsilon, \mathcal M_b^{(k)},\lVert\cdot\rVert\right)| &\leq |N_{in}\left(\epsilon/2, B_b(0),\lVert\cdot\rVert\right)| \nonumber\\ &\leq \left(\frac{6b}{\epsilon}\right)^{k^2}.
\end{align}
\end{lemma}

The bound on the covering number of the Euclidean ball is given in \cite{vershynin}, and the argument generalizes to any norm in a finite-dimensional space. The first inequalities are a consequence of triangular inequality. For specific matrix classes one can have improved bounds. 

In the following, we first exhibit concept classes based on circuits, which are valuable since they are associated to quantum states and measurement that can actually be produced efficiently on a quantum computer: not only the data that we receive could be states of this form, but also the resulting output hypothesis would be efficiently produced for any future need. 
In a further subsection we also present concept classes that are more motivated by physical scenarios, ideally exhibiting toy models for the setting described in the introduction, inspired by sensing and hamiltonian learning.

\subsection{Quantum Circuits}

We consider concept classes based on circuits. By choosing circuits that do not depend on the data, 
we will first give several examples of concept classes with covering numbers that are not just slowly growing, but in fact bounded by a quantity independent of $n$, the length of the sample. Then, including a finite number of gates that depend on the data through real functions with finite fat-shattering dimension, we will obtain concept classes that are slowly growing with $n$.

\subsubsection{Quantum circuits that give rise to state-valued functions}\label{subsec:examples}

First, we will look at an example of a concept class that maps to quantum states. We remind readers that this corresponds to choosing as the loss function the trace distance $L_s$.  
Consider the concept class $\mathcal{C}_m \subseteq \{c: \mathcal{X}\rightarrow  \mathcal{L}(\mathcal{H}_m)\}$ consisting of $m$-qubit quantum circuits acting on arbitrary input states. That is, 
\beq\label{eq:Cm}
\C_m = \{c_U(x) := U\rho(x) U^{\dag}\}_{U\in S_m}
\eeq
where $\rho:\mathcal{X} \rightarrow D(\mathcal{H}_m)$ is a fixed process preparing a mixed state that depends on a classical random variable $\mathcal{X}$, and $U$ is an arbitrary $m$-qubit unitary chosen from a set $S_m$. { For instance, we could be interested in studying processes corresponding to quantum circuits given by a particular architecture (specified by $S_m$) acting on pure computational-basis states.} This class is obtained in our formalism by setting $\mathcal{X} = \01^m$ and noting that the process that prepares the pure computational basis state $\ketbra{x}{x}$ corresponds to $\rho(x) = \ketbra{x}{x}$. Specific relevant examples of $S_m$ will be discussed later.

Then an upper-bound on $\gamma_{1,1}(n,\eps, \mathcal{C}_m)$ is given by the size of an $\eps/8$-net over the $m$-qubit unitaries in $S_m$, where the distance metric is taken to be the spectral norm $\lVert \cdot \rVert$ (see Eq. \eqref{eq:spec}).

In the following, we remind readers that $n$ is the number of data points observed, while $m$ is a parameter of the concept class -- the number of qubits in the circuits in the concept class.
\begin{proposition}\label{corr:packing}
$\gamma_{1,1}(n,\eps, \mathcal{C}_m) \leq  |N_{\text in}(\eps, S_m, \lVert \cdot \rVert)|$

\end{proposition}

\begin{proof}
To see this, observe that for an arbitrary state $\rho$ and unitaries $U,V$

\begin{equation}\label{eq:pertunit}
\frac{1}{2} \lVert U \rho U^{\dag} - V \rho V^{\dag} \rVert_1 \leq \lVert U-V \rVert.
\end{equation}

Indeed, with the spectral decomposition $\rho = \sum_i p_i \ketbra{\psi_i}{\psi_i}$, 
we have

\begin{align}
&\frac{1}{2} \lVert U \rho U^{\dag} - V \rho V^{\dag} \rVert_1 \nonumber\\&= \frac{1}{4} \lVert (U+V) \rho (U^{\dag} - V^{\dag}) + (U-V) \rho (U^{\dag}+V^{\dag}) \rVert_1 \label{eq:1}\\
&\leq \frac{1}{4} \sum_i p_i \left(\lVert (U+V)  \ketbra{\psi_i}{\psi_i} (U^{\dag} - V^{\dag}) \rVert_1 \right.\nonumber\\&+ \left.\lVert (U-V)  \ketbra{\psi_i}{\psi_i} (U^{\dag} + V^{\dag}) \rVert_1 \right)\\
&= \frac{1}{4} \sum_i p_i \left( \lVert (U+V) \ket{\psi_i}\rVert_2 \lVert (U - V) \ket{\psi_i} \rVert_2\right. \nonumber\\&+\left.  \lVert (U-V) \ket{\psi_i}\rVert_2 \lVert (U + V) \ket{\psi_i} \rVert_2 \right)\\
& \leq \frac{1}{2} \max_{\ket{\psi_i}} \lVert (U+V) \ket{\psi_i}\rVert_2 \, \max_{\ket{\psi_j}} \lVert (U - V) \ket{\psi_j} \rVert_2\\
& \leq \frac{1}{2} \lVert U+V \rVert \lVert U-V \rVert\\
& \leq \lVert U-V \rVert \label{eq:inf}
\end{align}

where in the third line we have used Fact \ref{fact:1norm2norm}.

Therefore, for any two functions $c_U, c_V \in \mathcal{C}_m$, we have that

\begin{align}
&\frac{1}{n} \sum_{i=1}^{n} \frac{1}{2} \bigg\lVert c_U(x_i) - c_V(x_i)\bigg\rVert_1\\
&= \frac{1}{n} \sum_{i=1}^{n} \frac{1}{2} \bigg\lVert U\rho(x_i) U^{\dag} - V\rho(x_i) V^{\dag} \bigg\rVert_1\\
&\leq \lVert U - V \rVert \label{eq:U-V},
\end{align}
where the second inequality is inequality \eqref{eq:inf}. 

Now consider the set $\mathcal{C}_m^{(\epsilon)} \subseteq \mathcal{C}_m$ induced by the $\eps$-net of $S_m$ (denoted as $S_{m}^{(\epsilon)}$),  and
\beq
\mathcal{C}_m^{(\epsilon)} := \left\{c_V: c_V(\vec{x}) = V \rho(x_i)V^{\dag} \right\}_{V\in S_{m}^{(\epsilon)}}
\eeqc
Eq.~\eqref{eq:U-V} implies that, if for all $U \in S_m, \, \exists V \in S_{m}^{(\epsilon)}$ such that $\lVert U-V \rVert < \eps$, then for all $c_U \in \mathcal{C}_m, \, \exists c_V \in \mathcal{C}_m^{(\epsilon)}$ such that $\lVert c_U -c_V \rVert_{1,\vec{x}} < \eps$. This is a sufficient condition for $\mathcal{C}_m^{(\epsilon)}$ to be an $\eps$-net of $\mathcal{C}_m$.
\end{proof}

Proposition \ref{corr:packing} makes an important reduction: to upper-bound $\gamma_{1,1}(n,\eps, \mathcal{C}_m)$, where $\mathcal{C}_m$ is a concept class induced by a set of $m$-qubit quantum circuits $S_m$, it suffices to upper-bound the size of an $\eps$-net in the spectral norm over the set $S_m$. 

We now compute this covering number for circuit classes $S_m$ corresponding to particular architectures of relevance. 
\begin{itemize}
    \item Consider the class of one-dimensional local random quantum circuits on $m$ qubits of depth $\ell$, which was defined by Brandao, Harrow and Horodecki \cite{Brandao2016} as applying a Haar-random $2$-qubit nearest-neighbor gate on a uniformly-random pair of neighboring qubits at each time-step $t=1, \ldots \ell$. (Note that the distribution from which the circuits are drawn is actually immaterial to us; the concept class merely consists of the support of that distribution.) We shall call this class $\text{LQC}(m,\ell)$, and we will call the induced loss function class $\mathcal{G}_{\text{LQC}(m,\ell)}$. 

\begin{corollary}[Bound on covering net size for 1d random local quantum circuits]\label{corr:LQC}
\begin{equation}\label{eq:NinuLQC}
N_{in}(\eps, \text{LQC}(m,\ell), \lVert \cdot \rVert) \leq  \left[m \cdot \left(\frac{6 \ell}{\eps}\right)^{32}\right] ^{\ell},
\end{equation}
and hence by Proposition \ref{corr:packing}, 
\begin{align}\label{eq:net_LQC}
&\Gamma_1(n,\eps,\mathcal{G}_{\text{LQC}(m,\ell)}) \leq\gamma_{1,1}(n,\eps,\mathcal{C}_{\text{LQC}(m,\ell)}) \nonumber\\&\leq \left[m \cdot \left(\frac{6\ell}{\eps}\right) ^ {32}\right]^\ell.
\end{align}

\end{corollary}
\begin{proof}
The argument follows along the lines of that made in \cite{huang2022quantum,Caro2022},
and so we only sketch it briefly. Let $\mathcal{N}_{\tilde{\eps}}$ be an $\tilde{\eps} = \frac{\eps}{d}$-net over a single 2-qubit unitary. 
Let us now consider the net 
\beq
\tilde{\mN}_{\eps}= \left\{U_{1,p(1)} U_{2,p(2)} \ldots U_{d,p(\ell)} \Big| U_i \in \mN_{\tilde{\eps}}, p(i)\in [\ell] \right\}
\eeqp
where $U_{i,p(i)}$ denotes the unitary $U_i$ acting on qubits $p(i),p(i)+1$. As a special case of Lemma \ref{lem:vershynin}, an $\eps$-net over the set of 2-qubit unitaries (a set that lies within $B_{\lVert \mathbbm{1}_{(\mathbb{C}^2)^{\otimes 2}} \rVert}(0)$, a set of matrices with $4\times 4 \times 2$ parameters), has size
$\left(\frac{ 6}{\eps}\right) ^ {32}$. Combining this with the fact that there are $m-1$ choices of qubit pairs on which the unitary can act, $\tilde{\mathcal{N}}_{\eps}$ has size $
\left[\left(\frac{ 6}{\tilde{\eps}}\right) ^ {32} \cdot (m-1)\right]^\ell$. An arbitrary circuit in $LQC(m,\ell)$ has the form 
\beq
V = V_1,\ldots V_\ell
\eeq
where each $V_i$ is a 2-qubit nearest-neighbor gate. We now show that $\tilde{\mN}_{\eps}$ is an $\eps$-net over the set $LQC(m,\ell)$. For $V_i$ acting on qubits $p(i)$, $p(i)+1$, there exists some $U_{i} \in \mN_{\tilde{\eps}}$ such that $\lVert U_{i}-V_{i}\rVert < \tilde{\eps},$ and we can enforce that it acts on qubits $p(i)$, $p(i)+1$. This means that $U_1 \ldots U_\ell \in \tilde{\mN}_{\eps}$, and moreover

\begin{align}\label{triangcirc}
    &\lVert V_1 \ldots V_\ell - U_1 \ldots U_\ell \rVert \nonumber\\&\leq \lVert V_1 \ldots V_{\ell-1} V_\ell - V_1 \ldots V_{\ell-1}U_\ell \rVert \nonumber\\&+ \lVert V_1 \ldots V_{\ell-1}U_\ell - U_1 \ldots U_{\ell-1} U_\ell \rVert \\
    &\leq \lVert V_\ell - U_\ell \rVert + \lVert V_1 \ldots V_{\ell-1} - U_1 \ldots U_{\ell-1} \rVert \\
    &\leq \eps/\ell + \lVert V_1 \ldots V_{\ell-1} - U_1 \ldots U_{\ell-1} \rVert \label{eq:secondlast}\\
    &\leq \eps
\end{align}

where the first inequality is the triangle inequality, the second inequality follows from unitary invariance of the spectral norm, the third inequality follows from the assumption that $\mN_{\tilde{\eps}}$ is an $\eps/\ell$-net and the last inequality follows from repeating the above three steps on the remaining term in \eqref{eq:secondlast}. 
\end{proof}

\item Exactly analogously,  we could also consider the set $S_{\text{b}}(m,\ell)$ of all `brickwork' local quantum circuits, where every pair of neighboring qubits has a nearest-neighbor two-qubit gate acting on them at every time step. We then obtain
\beq\label{eq:net_brick}
\Gamma_1(n,\eps,\mathcal{G}_{S_{\text{b}}(m,\ell)}) \leq\gamma_{1,1}(n,\eps,\mathcal{C}_{S_{\text{b}}(m,\ell)}) \leq \left(\frac{6m\ell}{\eps}\right) ^ {32m\ell}
\eeqp
\item Alternatively, we could consider choosing our unitaries from the set of all possible unitaries on $m$ qubits. Let us call this set $S_m$. Immediately from Lemma \ref{lem:vershynin} we have
\beq\label{eq:Ninunitaries}
\gamma_{1,1}(n,\eps,\mathcal{C}_{S_m}) \leq \left(\frac{6}{\eps}\right)^{2^{2m+2}}
\eeqp
\end{itemize}
Observe that in all of these cases, the covering number bound (which determines learnability of the corresponding class) is independent of the number of sampled points, $n$.

\subsubsection{Quantum circuits that give rise to projector-valued functions}

In the previous subsection, we computed the covering number of state-valued concept classes. In this subsection, we will give an example of a concept class consisting of \textit{projector}-valued functions, and compute its covering number. 

\textbf{Side note:} An immediate example of a projector-valued concept class is, in fact, the state-valued concept classes obtained by circuits applied to computational basis states, as a special case of the settings considered in the last subsection. This is because the concepts output pure states which are rank-1 density matrices, i.e. rank-1 projectors. Even for higher ranks, the calculation of the covering number for this concept class proceeds similarly. Some slight tweaks are needed to account for the use of a different loss function (namely $L_p$) and notion of distance between concepts (namely the spectral norm instead of the trace norm).

We present a class of projector-valued concepts that have an explicit $x$ dependence but each of them is based on a single circuit. We fix a projector-valued function $\Pi(x)$ and we define
$\C_{\Pi}  \subseteq \{c:\mathcal{X} \rightarrow \mathcal{L}(\mathcal{H}_{m})\}$, where $S_m$ is a set of circuits:
\beq\label{eq:pi}
\C_{\Pi} = \{\Pi_U:=U \Pi(x) U^{\dag}\}_{U \in S_m}
\eeqp

To calculate the covering numbers, it suffices to find the size of an $\eps/8$-net over the isometries $S_m = \{V_h\}_h$ according to the diamond distance. To see that this suffices, we will need two standard consequences of the definition of operator norms, as follows

\begin{proposition}\label{prop:netoverprojectors}
For any $\Pi_1$ and $\Pi_2$ which are projectors onto $\mathcal{L}(\mathcal{H}_m)$,
\begin{equation}
    \max_{\rho \in D(\mathcal{H}_m)} |\Tr[\rho (\Pi_1 - \Pi_2)]| = \lVert \Pi_1 - \Pi_2 \rVert.
\end{equation}
\end{proposition}
\begin{proof}
$\Pi_1 - \Pi_2$ is Hermitian, hence it is diagonalizable and it has real eigenvalues, so   $\lVert \Pi_1 - \Pi_2 \rVert$ equals the largest absolute values among those of the eigenvalues of $\Pi_1 - \Pi_2$, say $|\lambda|$, with corresponding eigenvector $\ket{\lambda}$. Then, $\lVert \Pi_1 - \Pi_2 \rVert=|\Tr[\ketbra{\lambda}{\lambda} (\Pi_1 - \Pi_2)]|$. Conversely, by the duality between trace norm and operator norm, $\max_{\rho \in D(\mathcal{H}_m)} |\Tr[\rho (\Pi_1 - \Pi_2)]| \leq \lVert \Pi_1 - \Pi_2 \rVert$. 

\end{proof}

\begin{proposition} \label{prop:diamond} 
For any two unitaries $U, V: \mathcal{L}(\mathcal{H}_m) \rightarrow \mathcal{L}(\mathcal{H}_m)$,
\begin{align}\label{submult}
&||UXU^{\dagger}-VXV^{\dagger}||\\
&\leq ||U X(U^{\dagger}-V^{\dagger})||+||(U-V) X V^{\dagger})||\\
&\leq 2||X||||U-V||.
\end{align}

\end{proposition}

Let us fix any internal $\eps$-net $S_m^{(\epsilon)}$ of cardinality $N_{in}(\eps, S_m, \lVert \cdot \rVert)$ over the set of unitaries $S_m$. We can then prove that the set induced by $S_m^{(\epsilon)}$ 
\begin{align}
\mathcal{C}_\Pi^{(\epsilon)}&:= \Big\{\Pi_V:\mathcal{X}^{n}\rightarrow L(\bigotimes_{i=1}^{n}\mathcal{H}_{m}),\nonumber\\
&\Pi_V(\vec{x}) = V\Pi(x)V^{\dag}\Big\}_{V\in S_m^{(\epsilon)}}
\end{align}
{ lifts to a $2\epsilon$-net over $\C_{\Pi}$, as follows: we'll show that for any $\Pi_V \in \C_{\Pi}$, we can find $\Pi_U \in \mathcal{C}_\Pi^{(\epsilon)}$ at most distance $2\eps$ away. For, let $V\in S_m$ be the unitary associated with $\Pi_V$. Then, choose $U$ to be the nearest point in $S_m^{(\epsilon)}$ to $V$. Then we claim that $\Pi_U$ suffices to be such a point:}
\begin{align}
 &\frac{1}{n} \sum_{i=1}^{n} ||\Pi_V(x_i) - \Pi_U(x_i)||\nonumber\\
 &\leq \frac{1}{n} \sum_{i=1}^{n} 2||\Pi(x_i)||\cdot||V - U|| \leq 2\eps \label{eq:3}
\end{align}
where the second to last inequality follows 
from Proposition \ref{prop:diamond} and the last one from the definition of $\epsilon$-net.

Finally, exactly analogously to Proposition \ref{corr:packing} and Proposition \ref{corr:LQC}, we obtain from Eq.~\eqref{eq:3}:

\begin{proposition}[Bounding $\Gamma_1$ and $\gamma_{1,\infty}$ for projector-valued function classes]\label{corr:LQC_proj}
For the class of projector-valued functions $\C_{\Pi}$ 
\beq
\gamma_{1,\infty}(n,\eps,\mathcal{C}_{\Pi})  \leq  N_{in}(\epsilon/2,S, \lVert \cdot \rVert) 
\eeq
\end{proposition}

\subsection{Data dependent circuits}
The previous subsections have studied concept classes based on circuits that are independent of the data. As a result, the upper bounds on the covering numbers do not depend on the size of the dataset. Without deep changes in the structure of the proofs of the bounds, we can also treat cases where this dependence is in fact explicit.  Such
data-dependent circuits are considered in variational quantum machine learning models with
data re-uploading (see, e.g.~\cite{PerezSalinas2020datareuploading}). Consider circuits described by a circuit model $S$ of $k$ two-qubit-gates, as in the previous sections,
and modify it by inserting in specific places in the sequence a number $k'$ of gates of the form $e^{i H_{j'}g_{j'}(x)},\, j'=1,...,k'$, $||H_i||\leq b$, $H_{j'}$ fixed and $g_{j'}(x)$ belonging to a concept class $\mathcal F$ with fat-shattering dimension $D$, where the functions in the class map to $[0,B]$ as explained in the beginning of this section. Each circuit in the concept class will be made of a sequence of $k$ data-independent qubit gates and $k'$ data-dependent gates, in specified positions. Since for any unitarily invariant matrix norm we have~\cite{van1977sensitivity}

\begin{equation}
||e^{X}-e^{Y}||\leq ||X-Y||e^{||X-Y||}e^{||X||}\label{expmatrixbound},
\end{equation}  

we have that, choosing $B=1/b$, whenever $|g_{i}(x)-g'_{i}(x)|\leq \epsilon/(4bet)$, $\epsilon<1$

\begin{align}
&||e^{i H_ig_i(x)}-e^{i H_ig'_i(x)}||\nonumber\\
&\leq ||H_ig_i(x)- H_ig'_i(x)||e^{||H_ig_i(x)-H_ig'_i(x)||} e^{||H_ig_i(x)||}\nonumber\\
&\leq b||g_i(x)- g'_i(x)||e^{b||g_i(x)-g'_i(x)||} e^{1}\nonumber\\
&\leq b\frac{\epsilon}{4bet} e^{b\frac{\epsilon}{4bet}} e^{1}\leq \frac{\epsilon}{2t}. 
\end{align}
Now, take two circuits $U$ and $V$ such that the non-data-dependent parts are such that $\sum_{j=1}^{k}||U_{j}-V_{j}||\leq\frac{\epsilon}{2}$, and also $||g_{j'}(x)-g'_{j'}(x)||_{1,\vec{x}}\leq \epsilon/(4bek')$, for every $j$, $j'$.
By triangular inequalities, for two state valued concepts $c_{U},c_{V}$ associated to circuits $U$, $V$, repeating the reasoning behind Eq.~(\ref{triangcirc}) and following equations, we get
\begin{align}
&\frac{1}{n} \sum_{i=1}^{n} \frac{1}{2} \bigg\lVert c_U(x_i) - c_V(x_i)\bigg\rVert_1
\leq \frac{1}{2n} \sum_{i=1}^{n}||U(x_i)-V(x_i)||\nonumber\\ 
&\leq\sum_{j=1}^k||U_j-V_j||\nonumber\\&+\frac{1}{n}\sum_{i=1}^n\sum_{j'=1}^{k'}||e^{i g_{j'}(x_i)H_{j'}}-e^{i g'_{j'}   (x_i)H_{j'}}||\nonumber\\
&\leq \epsilon/2+\frac{1}{n}\sum_{i=1}^n \sum_{j'=1}^{k'}\frac{\epsilon}{2k'}\nonumber\\
&\leq \epsilon/2+\epsilon/2=\epsilon.
\end{align}
The same inequality holds if we take projector-valued classes based on the circuits. This means that the covering number for a data-dependent class of circuits $\mathcal C'$ constructed from a qubit model $S$ is bounded as

\begin{align}
&\gamma_{1,\infty}(n,\epsilon, \mathcal C')\leq {N}_{in}(\epsilon/2,S,||\cdot||_S) \nonumber\\
&\times  2^{k'}\left(\frac{64 ne^2(k')^2}{\epsilon^2}\right)^{k'D\log_2 ( 16 e^2k'  n/(D \epsilon))}.
\end{align}

where $||U-V||_S:=\sum_{i=1}^{k}||U_i-V_i||$ with $\{U_i\}$ and $\{V_i\}$ being the gates of $U$ and $V$ respectively and ${N}_{in}(\epsilon/2,S,||\cdot||_S)$ can be bounded as in the previous section.

\subsection{Learning physical processes}
We now present some concept classes inspired by physical scenarios. We will always assume that the quantum register has dimension $d$.

We use facts about matrix continuity to obtain a bound on the covering number for our state-valued and projector-valued concept classes. In particular, we need the following.

\begin{itemize}
\item For any unitarily invariant matrix norm we have~\cite{van1977sensitivity}

\begin{equation}
||e^{X}-e^{Y}||\leq ||X-Y||e^{||X-Y||}e^{||X||}\label{expmatrixbound2}.
\end{equation}
\item~[Lemma 16 in~\cite{brandao2017quantum}] For any Hermitian matrices $H,H'$, 
\begin{equation}||\frac{e^{-H}}{\Tr[e^{-H}]}-\frac{e^{-H'}}{\Tr[e^{-H'}]}||_1\leq 2(e^{||H-H'||}-1)\label{Brandlemma}
\end{equation}
\item Weyl's perturbation theorem~\cite{bhatia2013matrix}. Let $A$ and $B$ be Hermitian matrices with ordered eigenvalues respectively $\lambda_{1}(A)\geq \lambda_{2}(A)\geq ...\geq \lambda_{n}(A) $ and $\lambda_{1}(B)\geq \lambda_{2}(B)\geq ...\geq \lambda_{n}(B)$. Then
\begin{equation}
\max_{i=1,...,n}|\lambda_{i}(A)-\lambda_i(B)|\leq ||A-B||.
\end{equation}
\item Given two projectors $E, F$, $||E-F||=||E^{\perp}F||$ (\cite{bhatia2013matrix}, Exercise VII.1.11).
\item (\cite{bhatia2013matrix}, Theorem VII.3.2) Let $A$ and $B$ be Hermitian operators, and let $S_1$, $S_2$ any two intervals such that their distance $\sup_{x\in S_1,y\in S_2}|x-y|=\delta>0.$ Let $E$ be the projector on eigenspaces of $A$ with eigenvalues contained in $S_1$, and $F$ be the projector on eigenspaces of $B$ with eigenvalues contained in $S_2$. Then, for any unitarily invariant norm $|||\cdot|||$,
\begin{equation}
|||EF||\leq \frac{\pi}{2\delta}|||A-B|||,
\end{equation}

\item (Corollary of previous points) Given two Hermitian matrices $A$ and $B$ such that $||A-B||\leq\epsilon$ and such that $A$ does not have eigenvalues in the interval $(\theta-2\epsilon,\theta+2\epsilon)$, (and thus $B$ does not have eigenvalues in the interval $(\theta-\epsilon,\theta+\epsilon)$ by Weyl perturbation's theorem) and $E,F$ the projectors on the eigenspaces associated to the eigenvalues less than $\theta-\epsilon$ satisfy
\begin{equation}
||E-F||\leq \frac{\pi}{4\epsilon}||A-B||.\label{bhatiapert}
\end{equation}
This follows from the previous point since $E$ coincides with the projector on the eigenspaces of $A$ with eigenvalues contained in $(-\infty,\theta-\epsilon]$ and $F^{\perp}$ coincides with the projector on the eigenspaces of $B$ with eigenvalues contained in $[\theta+\epsilon, +\infty)$, and the two sets $(-\infty,\theta-\epsilon]$ and $[\theta+\epsilon, +\infty)$ have distance $2\epsilon$.
\end{itemize}

We are now equipped to prove the following Propositions. Let $\F$ be a function class with fat-shattering dimension $D$, containing functions  $g\in F$ such that $|g(x)|\leq B$.  We remind that $\mathcal{M}_b^{(d)}$ indicates Hermitian matrices of operator norm bounded by $b$.

\begin{proposition}[Gibbs states of a perturbed Hamiltonian]Consider a concept class given by a set of Gibbs states obtained by perturbing a Hamiltonian with a field-dependent term, but the specific dependence on the field is not known, namely
\begin{align}
    \mathcal C=\{f|f(x)=\frac{e^{-H_0-g(x)V}}{\Tr[e^{-H_0-g(x)V}]},\,\, g\in\mathcal{F}, V\in \mathcal M_b^{(d)}\},
\end{align}
Its covering number is bounded as
\begin{align}
&\mathcal \gamma_{1,1}(n,2e^{\epsilon(B+b)}\epsilon (b+B),\mathcal C)\nonumber\\
&\leq 2\left(\frac{4 n B^2}{\epsilon^2}\right)^{D\log_2  (4 e B n/(D \epsilon))}\times \left(\frac{6b}{\epsilon}\right)^{d^2}.
\end{align}
\begin{proof}

Let $\mathcal{F}_{\epsilon}$ be an $\epsilon$-net  of $\mathcal{F}$ and ${\mathcal M_{b,\epsilon}^{(d)}}$ an $\epsilon$-net of $\mathcal M_{b}^{(d)}$. Then, for every $g\in \mathcal{F}$ there is $g\in \mathcal{F}_{\epsilon}$ such that $||g-g'||_{1,\vec x}\leq \epsilon$ and 
for every $V\in \mathcal M_{b}^{(d)}$ there is $V'\in {\mathcal M_{b,\epsilon}^{(d)}}$ such that $||V-V'||\leq \epsilon$. Let us consider the corresponding functions $f(x)=\frac{e^{-H_0-g(x)V}}{\Tr[e^{-H_0-g(x)V}]}\in \mathcal {C}$, $f'(x)=\frac{e^{-H_0-g'(x)V'}}{\Tr[e^{-H_0-g'(x)V'}]}\in \mathcal {C}$.
Then, using Eq.~(\ref{Brandlemma}) we obtain

\begin{align}
&||\frac{e^{-H_0-g(x)V}}{\Tr[e^{-H_0-g(x)V}]}-\frac{e^{-H_0-g'(x)V'}}{\Tr[e^{-H_0-g'(x)V'}]}||_1\nonumber\\&\leq 2(e^{||g(x)V-g'(x)V'||}-1)\nonumber\\
&\leq 2 (e^{(||V (g(x)-g'(x))||+||(V -V')g'(x))||}-1) \nonumber\\&\leq 2e^{\epsilon(B+b)}(||V (g(x)-g'(x))||+||(V -V')g'(x))||).
\end{align}

Therefore, 
\begin{equation}
||f-f'||_{1,1,\vec x}\leq 2e^{\epsilon(B+b)}\epsilon (b+B).
\end{equation}
 which implies that 

$ \mathcal C_{\epsilon}=\{f|f(x)=\frac{e^{-H_0-g(x)V}}{\Tr[e^{-H_0-g(x)V}]},\,\, g\in\mathcal{F}_{\epsilon}, V\in \mathcal M_{b,\epsilon}^{(d)}\},$
 
 is a $2e^{\epsilon(B+b)}\epsilon(b+B)$-net for $\mathcal C$.
\end{proof}

\end{proposition}

\begin{proposition}[Phase-shifts with position-dependent depth]\label{projectorcover}
 Consider a spatially local channel that acts as a power of an unknown unitary channel at position $x$ according to some classical variable at position $x$ (for example a thickness), measured as $g(x)$, which we probe with some position dependent state $\rho(x)$, possibly entangled with a different reference system for each $x$. We can model this class as either a state-valued class, but also as a projector-valued function class if the states $\rho(x)$ are pure:
\begin{align}
    &\mathcal C=\{f|f(x)=U^{g(x)}\rho(x){(U^\dagger)}^{g(x)}\,\, g\in\mathcal F,\\
    &U=e^{i H}\in \mathrm{U}(d), H\in \mathcal M_b^{(d)}\}.
\end{align}
The covering number of this class is bounded as

\begin{align}
&\mathcal \gamma_{1,1}(n,2 e^{\epsilon(B+b)+Bb}\epsilon (b+B),\mathcal C)\\
&\leq 2\left(\frac{4 n B^2}{\epsilon^2}\right)^{D\log_2  (4 e B n/(D \epsilon))} \times \left(\frac{6b}{\epsilon}\right)^{d^2},
\end{align}

and the same bound holds for $\mathcal \gamma_{1,\infty}(n,2 e^{\epsilon(B+b)+Bb}\epsilon (b+B),\mathcal C)$ when $\rho(x)$ are pure.
\end{proposition}
\begin{proof}
Let $\mathcal{F}_{\epsilon}$ be an $\epsilon$-net  of $\mathcal{F}$ and ${\mathcal M_{b,\epsilon}^{(d)}}$ an $\epsilon$-net of $\mathcal M_{b}^{(d)}$. Then, for every $g\in \mathcal{F}$ there is $g\in \mathcal{F}_{\epsilon}$ such that $||g-g'||_{1,\vec x}\leq \epsilon$ and 
for every $H\in \mathcal M_{b}^{(d)}$ there is $H'\in {\mathcal M_{b,\epsilon}^{(d)}}$ such that $||H-H'||\leq \epsilon$. 
Let us consider the corresponding functions $f(x)=U^{g(x)}\rho(x){(U^\dagger)}^{g(x)}\in \mathcal {C}$, $f'(x)={U'}^{g'(x)}\rho(x){({U'}^\dagger)}^{g'(x)}\in \mathcal {C}$.

Using Eq.~(\ref{expmatrixbound2})
therefore, since {$||H||\leq b$ and  $|g(x)|\leq B$}
\begin{align}
&||e^{i H g(x)}-e^{i H'g'(x)}||\nonumber\\
&\leq ||H g(x)- H'g'(x)||e^{||H g(x)- H'g'(x)||}e^{||H g(x)||} \nonumber\\
&\leq (||H (g(x)-g'(x))||+||(H -H')g'(x))||)e^{3(B+b)}
\end{align}

Moreover, applying Eq.~\eqref{eq:pertunit} (or Proposition~\ref{prop:diamond} for the projector-valued case) 

 to the pseudometric for our class, we obtain that for both $q=1,\infty$ 
\begin{align}
&||f-f'||_{1,q,\vec x}=\frac{1}{n}\sum_{i=1}^{n}||e^{iH g(x_i)}\rho(x_i)e^{-iH g(x_i)}\nonumber\\&-e^{iH' g'(x_i)}\rho(x_i)e^{-iH' g'(x_i)}||_{q}\\&\leq 2 e^{\epsilon(B+b)+Bb}\epsilon (b+B),
\end{align}
 which implies that
 $ \mathcal C_{\epsilon}=\{f|f(x)=U^{g(x)}\rho(x){(U^\dagger)}^{g(x)},\,\, g\in\mathcal{F}_{\epsilon}, V\in \mathcal M_{b,\epsilon}^{(d)}\},$
 is
 an $2 e^{\epsilon(B+b)+Bb}\epsilon (b+B)$-net for $\mathcal C$.
\end{proof}
\begin{proposition}[Projectors onto low-energy subspace of a position-dependent Hamiltonian]
Consider a class that projects on the eigenspaces of low energy of a set of unitary equivalent Hamiltonians, defining $\Pi_E$ as the projector on eigenstates of $H_0$ with energy less than $E$. We can model this as
\begin{align}
    &\mathcal C=\{f|f(x)=U^{g(x)}\Pi_E{(U^\dagger)}^{g(x)}\,\, g\in\mathcal F,\nonumber\\
    &U=e^{i H}\in \mathrm{U}(d), H\in \mathcal M_b^{(d)}\}.
\end{align}
The covering number of this class is bounded as
\begin{align}
&\mathcal \gamma_{1,\infty}(n,2 e^{3(B+b)}\epsilon (b+B),\mathcal C)\nonumber\\
&\leq 2\left(\frac{4 m B^2}{\epsilon^2}\right)^{D\log_2 (4 e B n/(D \epsilon))} \times \left(\frac{6b}{\epsilon}\right)^{d^2}.
\end{align}
\end{proposition}
\begin{proof}
The proof is identical to the one of Proposition~\ref{projectorcover} in the projector-valued case, exchanging $\rho(x)$ with $\Pi_E(x)$. 
\end{proof}
\begin{proposition}[Projectors on low-energy subspace of perturbed Hamiltonians]
Similarly, we can consider a class that projects on the eigenspaces of low energy of a set of perturbed Hamiltonians, defining $\Pi_E$ as the projector on eigenstates of $H_0+g(x)V$ with energy less than $E$, where $H_0$ does not have eigenvalues in the interval $[E-2\delta, E+2\delta]$, denoted $P_E(H_0+g(x)V)$.
\begin{equation}
    \mathcal C=\{f|f(x)=\Pi_E(H_0+g(x)V),\,\, g\in\mathcal F, V\in \mathcal M_{\delta/B}^{(d)} \},
\end{equation}
where we choose parameters in such a way that the perturbation is small, $||g(x)V||<\delta$.
The covering number of this class is bounded as

\begin{align}
&\mathcal \gamma_{1,\infty}(n,\frac{\pi\epsilon}{4\delta}(B+\delta\frac{1}{B}),\mathcal C)\nonumber\\
&\leq 2\left(\frac{4 m B^2}{\epsilon^2}\right)^{D\log_2 (4 e B n/(D \epsilon))} \times \left(\frac{6b}{\epsilon}\right)^{d^2}.
\end{align}
\end{proposition}
\begin{proof}
Let $\mathcal{F}_{\epsilon}$ be an $\epsilon$-net  of $\mathcal{F}$ and ${\mathcal M_{b,\epsilon}^{(d)}}$ an $\epsilon$-net of $\mathcal M_{b}^{(d)}$. Then, for every $g\in \mathcal{F}$ there is $g\in \mathcal{F}_{\epsilon}$ such that $||g-g'||_{1,\vec x}\leq \epsilon$ and 
for every $H\in \mathcal M_{b}^{(d)}$ there is $H'\in {\mathcal M_{b,\epsilon}^{(d)}}$ such that $||H-H'||\leq \epsilon$. 
Let us consider the corresponding functions $f(x)=\Pi_E(H_0+g(x)V)\in \mathcal {C}$, $f'(x)=\Pi_E(H_0+g'(x)V')\in \mathcal {C}$.

Using Eq.~(\ref{bhatiapert}), we have

\begin{align}
&||\Pi_E(H_0+g(x)V)-\Pi_E(H_0+g'(x)V')||\nonumber\\
&\leq \frac{\pi}{4\delta}||H g(x)- H'g'(x)||\nonumber\\&\leq \frac{\pi}{4\delta} \left(\frac{\delta}{B}|g(x)-g'(x)|+B||H -H'||\right).
\end{align}

Therefore, 
we get
\begin{align}
&||f-f'||_{1,\infty,\vec x}\nonumber\\
&=\frac{1}{n}\sum_{i=1}^{n}||\Pi_E(H_0+g(x_i)V)-\Pi_E(H_0+g'(x_i)V')||\\
&\leq\frac{\pi\epsilon}{4\delta}(B+\delta\frac{1}{B}),
\end{align}

which implies that  $ \mathcal C_{\epsilon}=\{f|f(x)=\Pi_E(H_0+g(x)V),\,\, g\in\mathcal{F}_{\epsilon}, V\in \mathcal M_{b,\epsilon}^{(d)}\},$ is a $\frac{\pi\epsilon}{4\delta}(B+\delta\frac{1}{B})$-net for $\mathcal C$.
\end{proof}

\section*{Acknowledgements}
The authors thank Ryan O'Donnell and Costin Bădescu for valuable correspondence. MF thanks Michalis Skotiniotis for discussions on quantum imaging. 
MF is also supported by a Juan de la Cierva Formaci\'on fellowship (Spanish MCIN project FJC2021-047404-I), with funding from MCIN/AEI/10.13039/501100011033 and European Union NextGenerationEU/PRTR. 
MF was also supported by Spanish MICINN PCI2019-111869-2 and 
 Spanish  Agencia Estatal de Investigación, project PID2019-107609GB-I00/AEI/$\backslash$$\backslash$10.13039/501100011033 and co-funded by the European Union Regional Development Fund within the ERDF Operational Program of Catalunya (project QuantumCat, ref. 001-P-001644), and currently by the Spanish MCIN (project PID2022-141283NB-I00) with the support of FEDER funds, by the Spanish MCIN with funding from European Union NextGenerationEU (grant PRTR-C17.I1) and the Generalitat de Catalunya, as well as the Ministry of Economic Affairs and Digital Transformation of the Spanish Government through the QUANTUM ENIA ``Quantum Spain'' project with funds from the European Union through the Recovery, Transformation and Resilience Plan - NextGenerationEU within the framework of the "Digital Spain 2026 Agenda". 
 YQ is supported by the Alexander von Humboldt Foundation. 
MR acknowledges support from the PNRR project 816000-2022-SQUID - CUP F83C22002390007, financed by the European Union - Next Generation EU program, the Einstein Foundation and the European Space Agency. 
\bibliography{qlearning.bib}

\begin{appendix}
\section{Proof of threshold reporting for product state}

\begin{proof}[Proof of Theorem~\ref{bernoulli_new}]\label{bernoulli_new_proof}
We will use the following facts already listed in \cite{BO21}:
\begin{itemize}
    \item  For $\mathsf{S}$ a random variable and $f : \mathbb{R} \to \mathbb{R}$ 1-Lipschitz,  
    \begin{equation}
        \mathrm{Var}[f(\mathsf{S})] \leq \mathrm{Var}[\mathsf{S}]\label{variance}
    \end{equation}
    (Proposition 3.1 in~\cite{BO21}).
    \item Fix $0 \leq p \leq 1$,  $q = 1-p$. Then for $C = (e-1)^2 \leq 3$, we have
  \begin{equation}
        q + pe^{2 \lambda} \leq (1 + Cpq \lambda^2) \cdot (q + pe^{\lambda})^2
    \quad \forall \lambda \in [0,1].\label{lem:inequality-1}
  \end{equation}
(Lemma 3.2 in~\cite{BO21}).
\item Let $\mathsf{S}$ be a discrete random variable, and let $B$ be an event on the same probability space with $\Pr[B] <
  1$. For each outcome~$s$ of~$\mathsf{S}$, define $f(s) = \Pr[B \mid \mathsf{S} = s]$. Then
  \begin{equation}
    \dist[\chi^2]((\mathsf{S} \mid \overline{B}), \mathsf{S})
    = \mathrm{Var}[f(\mathsf{S})] \bigm/ \Pr[\overline{B}]^2.\label{chi-squared}
  \end{equation}
  (Proposition 3.3 in~\cite{BO21}).

\end{itemize}
    We take $\mathsf{X}$ to be an exponential random variable with average $\mathbb{E}[\mathsf{X}]=\frac{1}{\lambda}$, and we assume $\lambda \leq \frac{1}{\sqrt{\sum_{i=1}^np_iq_i}}$ and $\lambda \leq 1$.
    
    Note that 
    \begin{align}
        f(s) &= \Pr[\mathsf{X} > \theta n - s] = \min(1, g(s)), \nonumber\\ g(s) &=  \exp(-\lambda(\theta n - s)).
    \end{align}
    Therefore we have from Eq.~(\ref{chi-squared})
    \[
    \dist[\chi^2]((\mathsf{T} \mid \overline{B}), \mathsf{T})
    = \mathrm{Var}[f(\mathsf{T})] \bigm/ \Pr[\overline{B}]^2\leq \frac{16}{9}\mathrm{Var}[f(\mathsf{T})],
    \]
    and we need to show that
      \[
        \mathrm{Var}[f(\mathsf{T})] \lesssim \Pr[B]^2 \cdot \left(\sum_{i=1}^np_iq_i\right) \lambda^2,
    \]
    
    We can use the fact that $y \mapsto
    \min(1,y)$ is 1-Lipschitz, therefore from Eq.~\eqref{variance} we have that
    $\mathrm{Var}[f(\mathsf{T})] \leq \mathrm{Var}[g(\mathsf{T})]$.
    
    As in the original proof, $\mathrm{Var}[g(\mathsf{T})]$ can be computed using the moment-generating function, this time of the Poisson binomial distribution $\mathsf{T}$, namely $\mathbb{E}[\exp(t \mathsf{T})]
    = \prod_{i=1}^n(q_i + p_ie^t)$:
    \begin{align*}
        \mathbb{E}[g(\mathsf{T})] &= \mathbb{E}[\exp(- \lambda(\theta n - \mathsf {T}))]
                            \nonumber\\&= \exp(- \lambda \theta n) \cdot\prod_{i=1}^n(q_i + p_ie^\lambda), \\
        \mathbb{E}[g(\mathsf{T})^2] &= \mathbb{E}[\exp(- 2 \lambda(\theta n - \mathsf{T}))]
                            \nonumber\\&= \exp(- 2 \lambda \theta n) \cdot \prod_{i=1}^n(q_i + p_ie^{2\lambda}).
    \end{align*}
    Thus
    \begin{align*}
        \mathrm{Var}[g(\mathsf{T})] 
         &= \mathbb{E}[g(\mathsf{T})]^2 \cdot \left(\frac{\mathbb{E}[g(\mathsf{T})^2]}{\mathbb{E}[g(\mathsf{T})]^2} - 1\right)
    \nonumber\\&=  \mathbb{E}[g(\mathsf{T})]^2 \cdot \left(\left(\prod_{i=1}^n\frac{q_i + p_ie^{2 \lambda}}{(q_i + p_ie^{\lambda})^2}\right) - 1\right) \\
    &\le\mathbb{E}[g(\mathsf{T})]^2 \cdot \left(\prod_{i=1}^n(1+3p_iq_i\lambda^2) - 1\right) \tag{Eq.~\ref{lem:inequality-1}}\\
     &\le\mathbb{E}[g(\mathsf{T})]^2 \cdot \left((1+3\frac{1}{n}\sum_{i=1}^{n}p_iq_i\lambda^2)^n - 1\right) \tag{AM-GM}\\
    &\leq 3 \mathbb{E}[g(\mathsf{T})]^2 \cdot \sum_{i=1}^np_iq_i \lambda^2 \tag{as $\lambda^2 \leq \frac{1}{\sum_{i=1}^np_iq_i}$}.
    \end{align*}
  
   Furthermore
    \begin{align}\label{ineq}
        \mathbb{E}[g(\mathsf{T})] &= \exp(- \lambda \theta n) \cdot \prod_{i=1}^n(q_i + p_ie^{\lambda})\nonumber\\&\leq \exp(- \lambda \theta n) \cdot \left(\overline {q} + \overline{p}e^{\lambda}\right)^n,
    \end{align}
    
    where $\overline{q}:=\frac{1}{n}\sum_{i=1}^nq_i$ and $\overline{p}:=\frac{1}{n}\sum_{i=1}^np_i$, and we will show
    \begin{equation}
    \exp(- \lambda \theta n) \cdot \left(\overline {q} + \overline{p}e^{\lambda}\right)^n \leq D\Pr[B],\label{boundgen}
    \end{equation}
    for some other constant $D>0$
    
    We consider two cases: $\overline{p}\geq \frac{1}{n}$ and $\overline{p} < \frac 1 n$.

    \paragraph{Case 1:} $\overline{p} \geq \frac {1} {n}$.  In this case we use that for a Binomial random variable $\overline{\mathsf{T}}$ with mean $\overline{p}n$, we have $\Pr[\overline{\mathsf{T}} > \overline{p}n] \geq \frac14$ (see, e.g.,~\cite{Doe18}). Moreover for the Poisson binomial distribution it holds that $\Pr(\mathsf T> k)\geq \Pr(\overline{\mathsf T}> k)$ if $0\leq k\leq n\overline{p}-1$ \cite{Hoe56,tang2019poisson}, therefore $\Pr[\mathsf{T} \geq \overline{p}n-1]\geq \Pr[\overline{\mathsf{T}} > \overline{p}n-1]\geq 1/4$.
    
    Now observe that: (i)~ $\theta \geq \overline{p}-\frac{1}{n}$, since we are assuming $\Pr[B] = \Pr[\mathsf{T} + \mathsf{X} > \theta n] < \frac14$; and, (ii)
    \begin{align}
        \Pr[B] &\geq \Pr[\mathsf{T} \geq \overline{p}n-1] \cdot  \Pr[\mathsf{T}+ \mathsf{X} \geq  \theta n|\mathsf{T}\geq\overline p n-1] \nonumber\\&\geq\frac{1}{4} \Pr[\mathsf{X} \geq (\theta-\overline{p})n+1] \nonumber\\&\geq \frac14 \exp(-\lambda((\theta-\overline{p})n+1)),
    \end{align} where the first inequality used independence of $\mathsf{T}$ and~$\mathsf{X}$, the second inequality the fact that $\Pr[\mathsf{T} \geq \overline{p}n-1]\geq 1/4$ and the third inequality used $(\theta -\overline{p})n +1\geq 0$ (by~(i)).
    Thus to establish our claim, it remains to show
    \[
        \exp(- \lambda \theta n) \cdot (\overline{q} + \overline{p}e^{\lambda})^n \lesssim \exp(-\lambda((\theta-\overline{p})n+1)),
    \]
    which is equivalent to
 \[
        \overline{q} + \overline{p}e^{\lambda} \lesssim \exp(\lambda \overline{p}).
    \]
    But this last inequality indeed holds, as for $\lambda \leq 1$ we have $\overline{q}+\overline{p}e^{\lambda} {\color{black} <} \overline{q}+\overline{p}(1+2\lambda) = 1+2\overline{p}\lambda$, and $1+2\overline{p}\lambda \leq 2\exp(\lambda \overline{p})$ always.

    \paragraph{Case 2:} $\overline{p} < \frac1n$.  Here $\overline{q} + \overline{p}e^{\lambda}< 1 + 2\overline{p}\lambda \leq 1 + \frac{2}{n}$, and so $(\overline{q}+\overline{p}e^{\lambda})^n \lesssim 1$, which implies that Eq.~\eqref{boundgen} follows from $\Pr[B] \geq \Pr[\mathsf{X} > \theta n] = \exp(-\lambda \theta n)$.
\end{proof}

\section{Results on sampling without replacement}\label{app:worep}

The following Theorem due to Hoeffding holds~\cite{Hoeffding_samplingworep} (see also ~\cite{Bardenet_samplingworep})
\begin{theorem}\label{hoeffrep}
Let $\mathcal{X}={x_1,...,x_N}$ a finite population of points, and $\{X_{1},...,X_{n}\}$ a random sample drawn without replacement from $\mathcal{X}$. Let
\begin{equation}
a=\min_{1\leq i\leq N}x_{i}\qquad b=\max_{1\leq i\leq N}x_{i}.
\end{equation}
Then it holds
\begin{equation}
\mathrm{Pr}\left(\sum_{i=1}^{n}X_i-n\mu\geq n\epsilon\right)\leq e^{-\frac{2n \epsilon^2}{(b-a)^2}}.
\end{equation}
where $\mu:= \frac{1}{N}\sum_i x_i$ is the empirical mean of $\mathcal{X}$.
\end{theorem}

Now we apply the above Theorem to $M$ different finite populations of $N$ points. 
\begin{theorem}\label{th:multinorep}
Let $\mathcal{Y}={1,...,N}$, and let $\textsf{indices} = \{X_{1},...,X_{n}\}$ a random sample drawn without replacement from $\mathcal{Y}$. Furthermore, consider $M$ different finite populations of $N$ numbers $\{\mathcal{X}_i\}_{i=1}^M = \{\{x_{i,j}\}_{j=1}^N\}_{i=1}^M$ with $0\leq x_{i,j}\leq 1$, for every $i, j$, and for each $\mathcal{X}_i$, use the $\textsf{indices}$ to choose a subset of the points $\{x_{i,j}\}_{j \in \textsf{indices}}$. 
 
Then it holds
\begin{equation}
\mathrm{Pr}\left(\max_{i=1,...,M}|\sum_{j=1}^{n}x_{i,X_{j}}-n\mu_{i}|\geq n\epsilon\right)\leq 2Me^{-2n \epsilon^2},
\end{equation}
where $\mu_i:= \frac{1}{N}\sum_{j=1}^N x_{i,j}$ is the empirical mean of $\mathcal{X}_i$.
\end{theorem}
\begin{proof}

From Theorem~\ref{hoeffrep}, consider the $2M$ finite populations $\{\mathcal{X}_i, -\mathcal{X}_i\}_{i=1,...,M}$ (the negative sign is to handle the absolute value). The $\textsf{indices}$ determine the choice of a sample without replacement from each of the above finite populations. The Proposition follows from the union bound.

\end{proof}

We can thus prove the following Theorem.
\begin{widetext}

\begin{theorem}\label{thmappsampl}
Let $\mathcal{Y}=\{1,...,N\}$, $N\geq 3Kn$ and let the $K$ sets of indices $\{\{X_{nk+i}\}_{i=1}^n\}_{k=0,...,K-1}$ be random samples drawn without replacement from $\mathcal{Y}$. Furthermore, consider $m$ different finite populations of $N$ numbers $\{\{x_{i,j}\}_{j=1}^{N}\}_{i=1}^m$, with $0\leq x_{i,j}\leq 1$. From each population $\{x_{i,j}\}_{j=1}^{N}$ obtain $K$ subsets of size $n$, as $\{x_{i,X_{nk+1}},...,x_{i,X_{nk+n}}\}_{k=0,...,K-1}$. 
Then it holds
\begin{equation}
\mathrm{Pr}\left(\underset{\substack{ i=1,...,m\\ k=0,...,K-1}}{\max}|\sum_{j=1}^{n}x_{i,X_{nk+j}}-n\mu_{i}|\geq n\epsilon\right)\leq 2Kme^{-2n \epsilon^2/4}.
\end{equation}
where $\mu_i = \frac{1}{N}\sum_{j=1}^{N} x_{i,j}$.
\end{theorem}

\begin{proof}
    
We compute an upper bound recursively using
\begin{align}
&\mathrm{Pr}\left(\underset{\substack{ i=1,...,m\\ k=0,...,t}}{\max}|\sum_{j=1}^{n}x_{i,X_{nk+j}}-n\mu_{i}|\geq n\epsilon\right)\nonumber\\&\leq \mathrm{Pr}\left(\underset{\substack{ i=1,...,m\\ k=0,...,t-1}}{\max}|\sum_{j=1}^{n}x_{i,X_{nk+j}}-n\mu_{i}|\geq n\epsilon\right)\nonumber\\
&+\mathrm{Pr}\left(\max_{i=1,..,m}|\sum_{j=1}^{n}x_{i,X_{nt+j}}-n\mu_{i}|\geq n\epsilon \Big|\underset{\substack{ i=1,...,m\\ k=0,...,t-1}}{\max}|\sum_{j=1}^{n}x_{i,X_{nk+j}}-n\mu_{i}|\leq n\epsilon\right).
\end{align}

Defining $\mu_{i,t}$ such that

\begin{equation}
N\mu_i=(N-t n)\mu_{i,t}+\sum_{k=0}^{t-1} \sum_{j=1}^{n}x_{i,X_{nk+j}},
\end{equation}

given the condition $\underset{\substack{ i=1,...,m\\ k=0,...,t-1}}{\max}|\sum_{j=1}^{n}x_{i,X_{nk+j}}-n\mu_{i}|\leq n\epsilon$, 

we have

\begin{align}
N\mu_i\leq(N-tn)\mu_{it}+tn\mu_i+tn\epsilon\\
N\mu_i\geq(N-tn)\mu_{it}+tn\mu_i-tn\epsilon,
\end{align}
and therefore, for $t\leq K-1$
\begin{align}
\mu_{it}\geq \frac{N-tn}{N-tn}\mu_i-\frac{tn}{N-tn}\epsilon\geq \mu_i-\frac{K-1}{2K+1}\epsilon\geq \mu_i-\epsilon/2\nonumber\\
\mu_{it}\leq \frac{N-tn}{N-tn}\mu_i+\frac{tn}{N-tn}\epsilon\leq \mu_i+\frac{K-1}{2K+1}\epsilon\leq \mu_i+\epsilon/2,
\end{align}

meaning that 

\begin{equation}
|n\mu_{it}-n\mu_{i}|\leq n\epsilon/2
\end{equation}

\begin{align}
&\mathrm{Pr}\left(\max_{i=1,..,m}|\sum_{j=1}^{n}x_{i,X_{nt+j}}-n\mu_{i}|\geq n\epsilon |\underset{\substack{ i=1,...,m\\ k=0,...,t-1}}{\max}|\sum_{j=1}^{n}x_{i,X_{nk+j}}-n\mu_{i}|\leq n\epsilon\right)\nonumber\\
&=\mathrm{Pr}\left(\max_{i=1,..,m}|\sum_{j=1}^{n}x_{i,X_{nt+j}}-n\mu_{it}+n\mu_{it}-n\mu_{i}|\geq n\epsilon \Big|\underset{\substack{ i=1,...,m\\ k=0,...,t-1}}{\max}|\sum_{j=1}^{n}x_{i,X_{nk+j}}-n\mu_{i}|\leq n\epsilon\right)\nonumber\\
&\leq \mathrm{Pr}\left(\max_{i=1,..,m}|\sum_{j=1}^{n}x_{i,X_{nt+j}}-n\mu_{it}|+|n\mu_{it}-n\mu_{i}|\geq n\epsilon \Big|\underset{\substack{ i=1,...,m\\ k=0,...,t-1}}{\max}|\sum_{j=1}^{n}x_{i,X_{nk+j}}-n\mu_{i}|\leq n\epsilon\right)\nonumber\\
&\leq \mathrm{Pr}\left(\max_{i=1,..,m}|\sum_{j=1}^{n}x_{i,X_{nt+j}}-n\mu_{it}|\geq n\epsilon/2 \Big|\underset{\substack{ i=1,...,m\\ k=0,...,t-1}}{\max}|\sum_{j=1}^{n}x_{i,X_{nk+j}}-n\mu_{i}|\leq n\epsilon\right)
\end{align}

Using Theorem~\ref{th:multinorep} we have
\begin{equation}
\mathrm{Pr}\left(\max_{i=1,..,m}|\sum_{j=1}^{n}x_{i,X_{nt+j}}-n\mu_{it}|\geq n\epsilon/2 \Big|\underset{\substack{ i=1,...,m\\ k=0,...,t-1}}{\max}|\sum_{j=1}^{n}x_{i,X_{nk+j}}-n\mu_{i}|\leq n\epsilon\right)\leq 2me^{-\frac{2n \epsilon^2}{4}},
\end{equation}

and therefore

\begin{align}
&\mathrm{Pr}\left(\underset{\substack{ i=1,...,m\\ k=0,...,t}}{\max}|\sum_{j=1}^{n}x_{i,X_{nk+j}}-n\mu_{i}|\geq n\epsilon\right)\leq 2Kme^{-\frac{2n \epsilon^2}{4}}.
\end{align}

\end{proof}
\end{widetext}

\section{Realizable learning with pure states}\label{subsec:warmup}
We first consider learning in the setting where there is a guarantee that the unknown concept $f:\X \rightarrow D(\mathcal{H})$ comes from the concept class $\C$. This setting was first considered in Ref. \cite{ChungLin21}. Additionally, Ref. \cite{ChungLin21} gives two algorithms for learning in the realizable setting, one for pure states and one for mixed states. 
In this section we show that the pure-state algorithm can be straightforwardly improved if one looks at growth functions. 

We will use the shorthands:
\begin{align}
\Delta(c,c')&=\sum_{x}\mathcal{D}(x)d_{\text{tr}}(c(x),c'(x))\nonumber\\ \tilde{\Delta}(c,c', \vec {x})&=\frac{1}{|\vec{x}|}\sum_{i=1}^{|\vec{x}|}d_{\text{tr}}(c(x_i),c'(x_i)).
\end{align}

First let us make some definitions.
\begin{defn}[Measurement output distribution]
For any state $\sigma \in S$, and for a $d$-outcome POVM ${\cal M} := \{M_z\}_{z=1}^d$ where $M_z \in \mathbb{C}^d$, define the measurement output distribution $D_{{\cal M},\sigma}:[d] \rightarrow [0,1]$ as the distribution over outcomes upon measuring the state $\sigma$ with ${\cal M}$:
\begin{equation}
D_{{\cal M},\sigma}(z) = \Tr[\sigma M_z]
\end{equation}
\end{defn}

\begin{defn}[$\eps$-far set] For any $c\in \C$, define the $\eps$-far set:
$\C_{\geq \eps}(c):= \{h \in \C: \Delta(h,c) \geq \eps\}$ 
\end{defn}

\begin{algorithm}[H]
\textbf{Input:} $\left(x_{1}, \ket{c^{\ast}(x_{1})}\right),\left(x_{2}, \ket{c^{\ast}(x_{2})}\right), \ldots,\left(x_{n}, \ket{c^{\ast}(x_n)}\right)$ \\

\begin{algorithmic}[1]
\State Restrict the concept class $\C$ into $\C|_{\vec{x}}$.
\State Do a random $d$-outcome orthonormal basis measurement ${\cal M}^{(i)}$ on each output state $\ket{c^{\ast}(x_{i})}$ (if necessary augmenting the space with auxiliary variables) so that the measurement on the overall state is $\bigotimes_{i=1}^n{\cal M}^{(i)}:=\{\bigotimes_{i=1}^n M_{z_i}^{(i)}\mid\vec{z}\in[d]^n\}$. Let the outcome of the measurement ${\cal M}^{(i)}$ be $z_i \in [d]$.
\State Output the concept $h \in \C|_{\vec{x}}$ that is most likely to have produced the measurement outcomes recorded in Line 2:
\beq
h=\underset{c \in \C|_{\vec{x}}}{\arg \max }\,\, \prod_{i =1}^T D_{{\cal M}^{(i)}, c(x_i)}(z_i) \quad \text{(Maximum Likelihood)}
\eeqp
\end{algorithmic}
\caption{\label{algo:pure}Learning channels mapping to pure states}
\end{algorithm}

The error occurs if maximum likelihood chooses a concept $h$ which is $2\eps$-far from the true one (i.e. $\Delta(h,c^{\ast})$), which could happen for two reasons: 
\begin{enumerate}
    \item There is a concept $h$ such that $\Delta(h,c^{\ast})>2\eps$ but $\tilde{\Delta}(h,c^{\ast},\vec{x}) < \eps$, i.e. the sampled points do not give a good approximation of $h$'s distance from the true concept, which affects the construction of a good covering.
    \item There is a concept $h$ such that $\Delta(h,c^{\ast})> 2\eps$ but maximum likelihood on samples from $D_{\bigotimes_i {\cal M}^{(i)}, \bigotimes_i c^{\ast}(x_i)}$ outputs $h$, i.e., the sampled points do not give a good approximation of the measurement output distribution of the true concept. The probability of this happening is smaller than 
    \beq
    1-TV(D_{\bigotimes_i {\cal M}^{(i)}, \bigotimes_i c^{\ast}(x_i)},D_{\bigotimes_i {\cal M}^{(i)}, \bigotimes_i h(x_i)}),
    \eeq 
    as proven in \cite{ChungLin21}.
\end{enumerate}

To bound the probability of error from the first source, we note that for fixed $h \in \C|_{\vec{x}}$,
\begin{align}
    &\Pr_{\vec{x}\sim \mathcal{D}}\left[  \tilde{\Delta}(h,c^{\ast},\vec{x}) < \eps  \text{ and } \Delta(h,c^{\ast})\geq 2\eps \right]\nonumber\\&\leq  \Pr_{\vec{x}\sim \mathcal{D}}\left[ \left| \tilde{\Delta}(h,c^{\ast},\vec{x}) - \Delta(h,c^{\ast}) \right|\geq \eps\right] \leq 2 \exp(-2\eps^2 n).
\end{align}
where the second inequality follows from H\"{o}ffding's inequality. Hence, with probability $1-\exp(-2\eps^2 T)$ over the distribution $\mathcal{D}$, the samples $x_1,\ldots x_T$ are such that 
\beq\label{eq:condEventRealizable}
\Delta(h,c^{\ast}) \geq 2\eps \Rightarrow \tilde{\Delta}(h,c^{\ast},\vec{x}) \geq \eps
\eeqp

Let us now condition on this event and bound the probability of error from the second source. We use the fact \cite{Sen06} that for any two $d$-dimensional quantum states $\ket{\psi}, \ket{\psi'}$ with $r:= \operatorname{rank}(\ket{\psi}-\ket{\psi'})$ and a random orthogonal measurement basis $M$, with probability $1- \exp(-kd/r)$ which can be assumed to be sufficiently close to $1$ (by padding auxiliary quantum registers) over the choice of the basis,
\begin{align}\label{eq:Sen}
    TV(D_{{\cal M},\ket{\psi}},D_{{\cal M},\ket{\psi'}}) \geq \frac{k}{\sqrt{2}} d_{\rm tr} (\ket{\psi},\ket{\psi'}),
\end{align}
for $k$ a universal constant. Condition on this event too. Hence for any $h \in \C_{>2\eps}(c^{\ast})$,
\begin{align}\label{eq:empiricalTV}
&\frac{1}{n} \sum_{i=1}^n TV(D_{M,h(x_i)},D_{M,c^{\ast}(x_i)}) \nonumber\\&> \frac{k}{\sqrt{2}} \frac{1}{n} \sum_{i=1}^n d_{\rm tr} (\ket{h(x_i)},\ket{c^{\ast}(x_i)})\geq \frac{k}{\sqrt{2}} \eps.
\end{align}
The first inequality follows by linearity of Eq.\eqref{eq:Sen}. The second inequality follows from the conditioning event Eq.~\eqref{eq:condEventRealizable}. Finally, the probability of maximum likelihood on samples from $D_{\bigotimes_i {\cal M}^{(i)}, \bigotimes_i c^{\ast}(x_i)}$ erroneously assigning them to $D_{\bigotimes_i {\cal M}^{(i)}, \bigotimes_i h(x_i)}$ is at most 
\begin{align}
&1- TV(D_{\bigotimes_i {\cal M}^{(i)}, \bigotimes_i c^{\ast}(x_i)}, D_{\bigotimes_i {\cal M}^{(i)}, \bigotimes_i h(x_i)}) \nonumber\\&\leq 2^{-\Omega(n \cdot (\frac{1}{n} \sum_{i=1}^n TV(D_{{\cal M},h(x_i)},D_{{\cal M},c^{\ast}(x_i)}))^2)}\nonumber\\& \leq 2^{-\Omega(n \cdot k^2\eps^2/2)}
\end{align}
The first inequality follows from Lemma 6 of \cite{ChungLin21}; the second inequality follows from Eq.\eqref{eq:empiricalTV}. 

Having accounted for both sources of error for a fixed hypothesis $h$ in the 'bad' set $\C|_{\vec{x}} \cap \C_{>2\eps}(c^{\ast})$, we now take a union bound over all bad hypotheses. Thus the total probability of error is at most
\begin{align}\label{eq:perror}
&\left|\C|_{\vec{x}} \cap \C_{>2\eps}(c^{\ast}) \right| \, ( 2^{-2\eps^2 T}  + 2^{-\Omega(T \cdot k^2\eps^2/2)})\nonumber\\& \leq G(T) 2^{-\Omega(\eps^2 T)}.
\end{align}
Thus, there exists a finite number of samples that makes the right-hand-side of Eq.~\eqref{eq:perror} smaller than $\delta$ if
\beq\label{convcond}
\lim_{T\rightarrow \infty}\frac{1}{\eps^2} \log \left( \frac{G(T)}{\delta}\right) /T=0
\eeqp
Note that if $G(T) \leq \C$ for all $T<\infty$ one recovers the sample complexity of \cite{ChungLin21}.

\section{Learning via restricted restricted empirical risk estimation.}\label{appC}

The following variant of risk minimization via ERM on a projector-valued concept class which is not necessarily an $\epsilon$-net on the full input data shows that risk minimization can be done even if a less restrictive condition on $\gamma_{1,\infty}(n,\epsilon,\mathcal{C})$ holds. Risk estimation for projector-valued concept classes and risk minimization for state-valued concept classes can also be realized with a similar modification.

\begin{theorem}[Learning quantum processes via ERM, alternative statement]
Suppose the concept class $\C$ consists of quantum processes mapping to projectors and let $\epsilon>0$ be the accuracy parameter, and suppose that
\begin{equation}
\lim_{n\rightarrow \infty}\frac{\log \gamma_{1,\infty}(n,\epsilon,\mathcal C)}{n}=0, \qquad \forall\epsilon>0.
\end{equation}
Given as input a training set $S = (x_i,\rho(x_i))_{i=1}^n$ with $x_i \xleftarrow{\mathcal{D}} X$ and $\rho(\cdot)$ an unknown classical-quantum channel, there is an \textit{agnostic} learning algorithm $\mathcal{A}:\X^n \times \mathcal{L}(\mathcal{H}^{(d)})^{\otimes n} \rightarrow \C$ 
such that
\begin{equation}\label{eq:ERMp2}
\Pr_{S}[R_\rho(\mathcal{A}(S)) -\inf_{h\in\mathcal C} R_\rho(h)<9\epsilon]=: 1-p_{\text{err}}.
\end{equation}
Writing $n=Tkl$, for large enough $T=O(\log\frac{1}{\epsilon})$, $k=O(\log(1/\delta\log(1/\epsilon)))$, and $m_0\leq Tkl$, such that

\begin{equation}
4 \gamma_{1,q}(2 m_0, \epsilon / 64, \mathcal{C}) e^{-\frac{m_0 \epsilon^{2}}{512}}\leq\delta/4,
\end{equation}
there exist constants $C_1, C_2, C_3$ such that,
if $l$ satisfies

\beq\label{eq:l03}
(\log\gamma_{1,\infty}(m_0, \eps/2,\mathcal{C})+C_2)^2\leq C_1l\epsilon^2/9
\eeqc
we have
\begin{align}
p_{\text{err}}&\leq \frac{\delta}{2}+4\gamma_{1,\infty}(6Tkl, \epsilon / 2, \mathcal{C}) e^{-\frac{Tkl \epsilon^{2}}{32 }} \nonumber\\&+C_3 \log\frac{1}{\epsilon}\log\left(\frac{1}{\delta}\log\frac{1}{\epsilon}\right)\gamma_{1,\infty}(m_0, \epsilon / 2, \mathcal{C}) e^{-\frac{l \epsilon^{2}}{72 }}.
\end{align}

\end{theorem}

Once this is established, it can be verified that $p_{err}$ can be made arbitrarily small if
$$\lim_{n\rightarrow \infty}\frac{\log \gamma_{1,\infty}(n,\epsilon,\mathcal C)}{n}=0, \qquad \forall\epsilon>0.$$

\begin{proof}

The algorithm takes the first $m_0$ classical variables $\vec{x}_0$ and obtains an $\epsilon$-net on the true risks $\mathcal C|_{\vec{x}_0}$, with probability of error bounded as in Lemma~\ref{th:uniconv}

\begin{equation}
p_{err, net}(m_0)\leq  4\gamma_{1}(2 m_0, \epsilon / 64, \mathcal{C}) e^{-\frac{m_0 \epsilon^{2}}{512}}.
\end{equation}

Then it runs algorithm~\ref{algo:proj2} on the product state obtained from the full dataset $\vec{x}$, with the concepts in $\mathcal C|_{\vec{x}_0}$. In this way we obtain $c^*$ such that $\hat{R}(c)< \max_{c\in \mathcal C|_{\vec{x}_0}}\hat{R}(c)+6\epsilon$ with probability bounded as in the proof of Theorem~\ref{th:theoere}
\begin{align}
p_{err,ERM}&\leq T(0.97^k+(k+1)e^{-l\epsilon^2/72})\nonumber\\&+2Tk\gamma_{1}(m_0, \epsilon / 2, \mathcal{C})e^{-2l\epsilon^2/64}.
\end{align}

Finally, we require that uniform convergence occurs also for the full dataset (not just the first $m_0$ data).
This happens with probability
\begin{align}
p_{err,unif}&=\Pr_{S \sim \mathcal{D}^{n}}[\exists c \in \mathcal{C}:|R(c)-\hat{R}(c)|\nonumber\\& \geq \epsilon]\leq 4\gamma_{1,\infty}(12Tkl, \epsilon / 8, \mathcal{C}) e^{-\frac{6Tkl \epsilon^{2}}{32}}.
\end{align}

If this is true we have that, for the selected concept $c$

\begin{align}
R(c)&\leq \hat R(c)+\epsilon\\&\leq \sup_{c\in \mathcal C|_{\vec{x}_0}} \hat R(c)+7\epsilon\\&\leq \sup_{c\in \mathcal C|_{\vec{x}_0}} R(c)+8\epsilon\\&\leq  \sup_{c\in \mathcal C} R(c)+9\epsilon.
\end{align}

The probability of error is then upper bounded as
\begin{equation}
p_{err,net}+p_{err,ERM}+p_{err,unif}.
\end{equation}

With the choices made for the parameters, we obtain the thesis.
\end{proof}
\end{appendix}

\end{document}